\documentclass{amsart}
\usepackage[foot]{amsaddr}
\usepackage[margin=1in]{geometry}
\usepackage{amsmath, amsfonts, amssymb}
\usepackage{color}
\usepackage{algpseudocode}
\usepackage{algorithm}
\usepackage{siunitx}
\usepackage{graphicx}
\usepackage{subfig}
\usepackage{physics}
\usepackage{bbm}
\usepackage{tikz}
\usepackage{mathtools}
\usepackage{mathrsfs}
\usetikzlibrary{patterns,snakes,arrows.meta}
\colorlet{ColorPink}{black!10}

\usepackage{hyperref}
\usepackage{cleveref}

\newcommand\deleted[1]{}

\newcommand{\dt}{\Delta t}
\newcommand{\e}{\mathrm{e}} % Euler's number e
\newcommand{\id}{\mathrm{Id}} % Identity operator
\newcommand{\ii}{\mathrm{i}} % Imaginary unit, subscription i
\newcommand{\ff}{\mathrm{f}}
\newcommand{\sgn}{\mathrm{sgn}}
\newcommand{\Ls}{\mathcal{L}}
\newcommand{\Ms}{\mathcal{M}}
\newcommand{\bs}{\boldsymbol{s}}
\newcommand{\UijI}{\widehat{\mathscr{U}{}}{}_{\!\!ij}^{\!I}}
\newcommand{\Uij}{\widehat{\mathscr{U}{}}{}_{\!\!ij}}

\DeclareMathOperator{\SPAN}{span}

\newtheorem{theorem}{Theorem}
\newtheorem{lemma}[theorem]{Lemma}

\newtheorem{proposition}[theorem]{Proposition}
\newtheorem{definition}[theorem]{Definition}

\theoremstyle{remark}

\title{The Bold-Thin-Bold Diagrammatic Monte Carlo Method for Open Quantum Systems}

\author{Zhenning Cai}
\address[Zhenning Cai]{Department of Mathematics, National University of Singapore,
  Level 4, Block S17, 10 Lower Kent Ridge Road, Singapore 119076}
\email{matcz@nus.edu.sg}

\author{Geshuo Wang}
\address[Geshuo Wang]{Department of Mathematics, National University of Singapore,
  Level 4, Block S17, 10 Lower Kent Ridge Road, Singapore 119076}
\email{geshuowang@u.nus.edu}

\author{Siyao Yang}
\address[Siyao Yang]{Department of Mathematics, National University of Singapore,
  Level 4, Block S17, 10 Lower Kent Ridge Road, Singapore 119076}
\email{matsiya@nus.edu.sg}

\thanks{Zhenning Cai and Siyao Yang's work was supported by the Academic Research Fund of the Ministry of Education of Singapore under grant A-0004592-00-00.}


\begin{document}

\maketitle

\begin{abstract}
We present two diagrammatic Monte Carlo methods for quantum systems coupled with harmonic baths, whose dynamics are described by integro-differential equations. The first approach can be considered as a reformulation of Dyson series, and the second one, called ``bold-thin-bold diagrammatic Monte Carlo'', is based on resummation of the diagrams in the Dyson series to accelerate its convergence. The underlying mechanism of the governing equations associated with the two methods lies in the recurrence relation of the path integrals, which is the most costly part in the numerical methods. The proposed algorithms give an extension to the work [``Fast algorithms of bath calculations in simulations of quantum system-bath dynamics", Computer Physics Communications, to appear], where the algorithms are designed based on reusing the previous calculations of bath influence functionals. Compared with the algorithms therein, our methods further include the reuse of system associated functionals and show better performance in terms of computational efficiency and memory cost. We demonstrate the two methods in the framework of spin-boson model, and numerical experiments are carried out to verify the validity of the methods.       
\end{abstract}

\section{Introduction}

It is a great challenge to numerically simulate a quantum system 
 coupled to its environment
 to a non-negligible extent.
Nevertheless,
 open quantum systems are 
 playing increasingly important roles
 in many fields including 
 but not limited to
quantum computation \cite{nielsen2002quantum},
 quantum optical systems \cite{breuer2002theory,carmichael2009open},
 quantum sensing
 \cite{correa2017enhancement,salado2021spectroscopy},
 and chemical physics \cite{leggett1987dynamics,weiss2012quantum}.
In general,
 the main difficulty comes from
 the coupling of the system 
 and the bath, which leads to
 a non-Markovian process.
The non-locality of the simulation
 generally requires large memory cost.
By using an influence functional
 to describe the effect of the bath
 to the system \cite{feynman1963theory},
 the iterative quasi-adiabatic propagator path integral (i-QuAPI) \cite{makri1995numerical,makri1998quantum}
 assumes finite memory length
 to prevent the unrestrained growth of the memory cost.
While the bottleneck of the i-QuAPI method is its memory cost,
 some methods based on similar ideas
 have been developed to further reduce the memory cost
 and improve the computational efficiency.
For instance,
 the blip-summed decomposition
 eliminates a majority of system paths with small contributions \cite{makri2014blip,makri2017blip};
 the differential equation based path integral (DEBPI) studies the continuous form of the i-QuAPI 
 and formulates a system of differential equations \cite{wang2021differential}
 so that classical numerical methods can be applied.
Recently, the small matrix decomposition of the path integral (SMatPI) is proposed as a stable and efficient algorithm that overcomes
 the exponential scaling of the memory cost with respect to the memory length \cite{makri2020small,makri2021small1,makri2021small2}.

Another main idea to simulate the system-bath coupling
 is to use the Nakajima-Zwanzig generalized quantum master equation (GQME) \cite{nakajima1958quantum,zwanzig1960ensemble,shi2003new,zhang2006nonequilibrium}, which is an integro-differential equation describing the non-Markovian process of the system density matrix.
For a weakly coupled system,
 one may neglect the memory effect
 and work only on the Markovian approximation, resulting in the Lindblad equation \cite{davies1974markovian,lindblad1976generators}.
Methods based on the non-Markovian GQME includes momentum jump-GQME \cite{Kelly2013efficient}, which utilizes the surface hopping method in the computation of the memory kernel.
The transfer tensor method (TTM) for open quantum systems \cite{cerrillo2014non,buser2017initial,chen2020non} can also be considered as a discretization of GQME.
It deduces the memory kernel of a time-nonlocal quantum master equation and reconstructs the dynamical operators, which can extend simulations to very long times.

Except for the deterministic methods,
 another natural idea is to apply stochastic methods for the high-dimensional integrals in the Dyson series \cite{dyson1949radiation}.
The diagrammatic quantum Monte Carlo (dQMC) method uses diagrams to better visualize the coupling between the system and the bath \cite{prokof1998polaron,werner2009diagrammatic}.
In this method, the memory cost is no longer an issue, but it suffers from the numerical sign problem \cite{loh1990sign,cai2020numerical},
where the variance increases exponentially 
with respect to the simulation time $t$.
To tame the sign problem, 
 different techniques are developed.
The inchworm Monte Carlo method is one of the recent successful improvements, which applies ``bold-line diagrams'' to
 accelerate the convergence of series \cite{chen2017inchworm1,chen2017inchworm2,cai2020inchworm,yang2021inclusion,cai2020numerical}. It has achieved many successes in the computation of impurity models \cite{cohen2015taming, ridley2018numerically, eidelstein2020multiorbital}. The idea of using bold lines can be traced back to \cite{prokof1997bold}, and has been applied to the impurity model in \cite{gull2010bold}. In general, a bold diagram is the sum of a number of original diagrams denoting the terms in the Dyson series. By clustering some diagrams in the Dyson series, the total number of diagrams to be summed is reduced, and hence improves the numerical efficiency.
  
% Multilevel blocks samples ``blocks'' instead of paths \cite{egger2000path,muhlbacher2003crossover} so that the sign problem is relieved.
%Nevertheless, for long time simulations,
% the number of samples required is getting large,
% resulting in low computation efficiency.
In \cite{cai2022fast},
 the authors design a fast algorithm
 that reuses the samples generated in previous time steps to reduce the calculations of the bath influence functional without losing accuracy of dQMC and the inchworm Monte Carlo method.
In this paper,
 we further discuss the reuse of previously calculated results based on the intrinsic invariance of the quantum propagators.
For the computation of Dyson series,
 we explore a recurrence relation of the series to be computed at each time step, and
 propose a fast algorithm
 that requires a much smaller amount of new samples at each time step.
While this idea can not be directly applied to the inchworm Monte Carlo method, we modify the inchworm Monte Carlo method by allowing some thin lines in the diagrams, and develop
 a new ``bold-thin-bold'' (BTB) Monte Carlo method
 where the recurrence relation can again be formulated and the fast algorithm can also be designed. The name ``bold-thin-bold'' describes the typical structure of the diagrams, where a thin line is sandwiched between two bold sections.

For the rest sections of this paper,
 we present the derivation of the fast algorithms for Dyson series and the BTB method.
In \Cref{sec:dyson},
 we first introduce the spin-boson model
 and the Dyson series,
 and then design a fast algorithm
 for the computation of Dyson series
 based on a recurrence formula. \Cref{section_semi_inchworm} contributes to the derivation of our new diagrammatic Monte Carlo method.
In \Cref{sec_review_imc},
 we give a brief introduction to inchworm Monte Carlo method
 and its diagrammatic representation.
In \Cref{Modified_IMCM},
 the BTB method is presented intuitively by diagrams
 and formally by equations.
The fast algorithm
 for BTB method is designed 
 based on the recurrence relation
 in \Cref{sec:semi_inchworm_numerical_method}.
At the end of both \Cref{sec:dyson} and \Cref{section_semi_inchworm},
 some remarks about the number of samples required
 and the comparison between methods are given for better understanding. In \Cref{section_numerical_results},
 we carry out some numerical experiments
 to illustrate the validity of our algorithms. 
Some discussions about the computation efficiency 
 and the variances of two methods
 are also given in \Cref{section_numerical_results}.
In \Cref{final_conclusion},
 we briefly conclude the idea of the whole paper
 and discuss the possible future work.

\section{Fast algorithm for Dyson series}
 \label{sec:dyson}
\subsection{Spin-boson model and Dyson series}
We study the density matrix $\rho(t)$ of the system-bath dynamics in an open quantum system, which can be described by the following von Neumann equation:
\begin{equation}
    \ii \frac{\dd \rho(t)}{\dd t} 
   = H\rho(t) - \rho(t)H
    \label{vonNeumannEq}
\end{equation}
The Hamiltonian $H$ is an operator on the Hilbert space $\mathcal{H} = \mathcal{H}_s \otimes \mathcal{H}_b$, where $\mathcal{H}_s$ and $\mathcal{H}_b$
 represent the spaces associated with the system and bath, respectively. The Hamiltonian $H$ can be written as the sum of three parts, the system, the bath and their coupling: 
\begin{equation*}
    H = H_s \otimes \id_b + \id_s \otimes H_b + W,
\end{equation*}
where $H_s$ are operators on $\mathcal{H}_s$ and $H_b$ are operators on $\mathcal{H}_b$. The identity operators
$\id_s,\id_b$ are defined over the system space $\mathcal{H}_s$ and bath space $\mathcal{H}_b$, respectively. 
The coupling term is assumed to have the tensor-product form 
\begin{displaymath}
W = W_s \otimes W_b. 
\end{displaymath}
In such an open quantum system,
 we are generally indifferent to the evolution of the bath.
However, as the system is coupled with its bath,
 the effect of the bath cannot be neglected, which introduces significant difficulties in the simulation of open quantum systems.
In this paper, we only consider spin-boson model,
 where the system is a single spin
 and the bath is modeled by a large number of harmonic oscillators. While the algorithms proposed in this work can be easily generalized to any multiple-state open quantum systems, the spin-boson model contains most difficulties for the simulations of system bath dynamics and thus is used as a fundamental example throughout this paper to demonstrate our methods.  
 
 \subsubsection{Spin-boson model}
In the spin-boson model, the Hilbert spaces $\mathcal{H}_s,\mathcal{H}_b$ are
\begin{equation*}
    \mathcal{H}_s = \SPAN\{\ket{-1},\ket{1}\}, \quad
    \mathcal{H}_b = \bigotimes_{j} \left(\mathscr{L}^2\left(\mathbb{R}^3\right)\right)
\end{equation*}
with $\mathscr{L}^2(\mathbb{R}^3)$ being the $\mathscr{L}^2$ space over $\mathbb{R}^3$
 and $j$ is the index of harmonic oscillators in the bath.
Additionally, the Hamiltonians and coupling operators are given by
\begin{equation*}
H_s = \epsilon \hat{\sigma}_z + \Delta \hat{\sigma}_x, \quad
H_b = \sum_{j} \frac{1}{2} \left(\hat{p}_j^2 + \omega_j^2 \hat{q}_j^2\right), \quad
W_s = \hat{\sigma}_z, \quad
W_b = \sum_j c_j \hat{q}_j.
\end{equation*}
The notations used in the equations above are defined as follows:
\begin{itemize}
\item Pauli matrices: $\hat{\sigma}_x = \ket{-1}\bra{1} + \ket{1}\bra{-1}$, $\hat{\sigma}_z = \dyad{1}{1}-\dyad{-1}{-1}$.
\item $\epsilon$: The energy difference between two spin states.
\item $\Delta$: The frequency of the spin flipping.
\item $\hat{p}_j$: The momentum operator of the $j$th harmonic oscillator, $\psi(\cdots, q_j, \cdots) \mapsto \partial_{q_j} \psi(\cdots, q_j, \cdots)$.
\item $\hat{q}_j$: The position operator of the $j$th harmonic oscillator, $\psi(\cdots, q_j, \cdots) \mapsto q_j \psi(\cdots, q_j, \cdots)$.
\item $\omega_j$: The frequency of the $j$th harmonic oscillator.
\item $c_j$: The coupling intensity between the spin and the $j$th harmonic oscillator.
\end{itemize}
The solution to \eqref{vonNeumannEq} can be formally written as $\rho(t) = \e^{-\ii tH} \rho(0) \e^{\ii t H}$.
With the assumption that the system and bath are not entangled initially $\rho(0) = \rho_s \otimes \rho_b$,
 the expectation of a given observable over the system $\hat{O}=\hat{O}_s \otimes \id_b$ can be computed by
\begin{equation*}
    \expval{\hat{O}(t)} = \tr\left(\hat{O}\rho(t)\right)
    = \tr\left(\hat{O}\e^{-\ii tH} \rho(0) \e^{\ii t H}\right)
    = \tr\left(\rho_s\otimes \rho_b \e^{\ii tH} \hat{O}_s \e^{-\ii t H} \right).
\end{equation*}
We define the propagator
\begin{equation*}
    G(-t,t) = \tr_b\left( \rho_b \e^{\ii t H} \hat{O}_s \e^{-\ii t H} \right)
\end{equation*}
where $\tr_b$ is the partial trace operator with respect to bath. 
With this propagator, $\expval{\hat{O}(t)} = \tr_s \left( \rho_s G(-t,t) \right)$ with $\tr_s$ is the partial trace operator with respect the system.
One may check that this propagator is Hermitian, i.e, $(G(-t,t))^\dagger = G(-t,t)$.
\subsubsection{Dyson series}
It is practically prohibitive to compute $\e^{\pm \ii t H}$ directly
 because of the high dimensionality of the bath space $\mathcal{H}_b$.
In mathematical physics, 
 the time evolution operator in the interaction picture can be expanded as Dyson series \cite{dyson1949radiation,cai2020inchworm,cai2022fast}:
\begin{equation}
    G(-t,t) = \e^{\ii t H_s} \hat{O}_s \e^{-\ii t H_s}
    + \sum_{m=1}^{+\infty}
    \ii^m
    \int_{-t\leqslant \boldsymbol{s} \leqslant t}
    \dd \boldsymbol{s}
    \left[(-1)^{\#\{\boldsymbol{s}<0\}}
    \mathcal{U}^{(0)} (-t,\boldsymbol{s},t)
    \mathcal{L}_b(\boldsymbol{s})\right] \text{~for~} t \geqslant 0, 
    \label{fullPropagator}
\end{equation}
where the variable of integration $\boldsymbol{s}$ is an $m$-dimensional vector
 and its integral domain $-t\leqslant \boldsymbol{s} \leqslant t$ is an $m$-dimensional simplex
%  When $x=0$, $G(0,0) = O_s$ actually means that $\displaystyle \lim_{s\rightarrow 0^+} G(-s,s) = O_s$. \zc{What does ``When $x=0$'' mean?}
% Actually the binary function $G(s_\ii,s_\ff)$ 
%  is discontinuous at $s_\ii=0$ and $s_\ff=0$,
%  which will be further discussed in \Cref{section_semi_inchworm}. \zc{I think there is no need to introduce the discontinuity at this point. It does not really give the readers any necessary information.}
% In this section, we only consider the univariate function $G(-t,t)$ with $t\geqslant 0$. 
\begin{equation*}
    \left\{
    (s_1,\cdots,s_m)\in\mathbb{R}^m \vert
    -t \leqslant s_1 \leqslant s_2 \leqslant \cdots \leqslant s_m \leqslant t
    \right\},
\end{equation*}
and thus the integral in \eqref{fullPropagator} should be interpreted as
\begin{equation*}
    \int_{-t\leqslant \boldsymbol{s} \leqslant t} \dd \boldsymbol{s}
    = \int_{-t}^{t} \dd s_m
    \int_{-t}^{s_m} \dd s_{m-1}
    \cdots
    \int_{-t}^{s_2} \dd s_1.
\end{equation*}
In the integrand, the notation $\#\{\boldsymbol{s}<0\}$ denotes the number of negative components 
in $\boldsymbol{s} = (s_1,\cdots,s_m)$, and
$\mathcal{U}^{(0)}(-t,\boldsymbol{s},t)$ describes the contribution of the system dynamics, which is defined by
\begin{equation}
    \label{U0}
    \mathcal{U}^{(0)} (-t,\boldsymbol{s},t)
    = G_s^{(0)}(s_m,t) W_s
    G_s^{(0)}(s_{m-1},s_m) W_s
    \cdots
    W_s G_s^{(0)}(s_1,s_2) 
    W_s G_s^{(0)}(-t,s_1)
\end{equation}
where
\begin{equation}
    G_s^{(0)} (s_\ii,s_\ff)
    = \begin{cases}
    \e^{-\ii (s_\ff-s_\ii)H_s}, 
    &\text{if }s_\ii \leqslant s_\ff < 0 \\
    \e^{-\ii (s_\ii-s_\ff)H_s},
    &\text{if }0 \leqslant s_\ii \leqslant s_\ff\\
    \e^{\ii s_\ff H_s}\hat{O}_s \e^{\ii s_\ii H_s},
    &\text{if } s_\ii < 0 \leqslant s_\ff
    \end{cases}.
    \label{Gs0_definition}
\end{equation}
The bath influence functional $\mathcal{L}_b$ satisfies the following Wick's theorem \cite{Negele1988}
\begin{equation*}
    \mathcal{L}_b (s_1,\cdots,s_m) 
    = \begin{cases}
    0, & \quad\text{if $m$ is odd} \\
    \displaystyle\sum_{\mathfrak{q}\in\mathcal{Q}(\boldsymbol{s})}
    \prod_{(s_j,s_k)\in\mathfrak{q}}
    B(s_j,s_k), & \quad\text{if $m$ is even} 
    \end{cases}
\end{equation*}
where $B(\tau_1,\tau_2)$ is a complex-valued functional
\begin{equation}
    B(\tau_1,\tau_2) = \frac{1}{\pi}
    \int_0^{+\infty} J(\omega) \left[
    \coth\left(\frac{\beta\omega}{2}\right)
    \cos(\omega \Delta \tau) - \ii \sin(\omega \Delta \tau)
    \right]\dd \omega
    \label{bath_function_B}
\end{equation}
with $\Delta \tau = \vert \tau_1 \vert - \vert \tau_2 \vert$
and $J(\omega)$ being the spectral density.
The set $\mathcal{Q}(\boldsymbol{s})$ is given by
\begin{equation}
\label{Q}
\begin{split}
    \mathcal{Q}(s_1,\cdots,s_m)
    = \left.\Big\{
    \left\{
    (s_{j_1},s_{k_1}),\cdots,(s_{j_{m/2}},s_{k_{m/2}})   
    \right\}
    \right\vert&
    \{j_1,\cdots,j_{m/2},k_1,\cdots,k_{m/2}\} = \{1,\cdots,m\}, \\
    &j_l < k_l \text{ for any }l=1,\cdots,m/2
    \Big\}.
\end{split}
\end{equation}
For example, 
\begin{align*}
     &\mathcal{Q}(s_1,s_2) = \big\{\{(s_1,s_2)\}\big\},\\
     &  \mathcal{Q}(s_1,s_2,s_3,s_4) = \big\{\{(s_1,s_2),(s_3,s_4)\},\{(s_1,s_3),(s_2,s_4)\},\{(s_1,s_4),(s_2,s_3)\}\big\}
\end{align*}
and the corresponding bath influence functionals are formulated as 
\begin{align*}
    &\Ls_b(s_1,s_2) = B(s_1,s_2),\\
    & \Ls_b(s_1,s_2,s_3,s_4) = B(s_1,s_2)B(s_3,s_4)+ B(s_1,s_3)B(s_2,s_4) + B(s_1,s_4)B(s_2,s_3).
\end{align*}
Inserting the definition of $\Ls_b$ into \eqref{fullPropagator}, the series now only sums over the even terms as the bath influence functional vanishes when $m$ is odd:   
\begin{equation}
    \begin{split}
    G(-t,t) = & \ \e^{\ii t H_s} \hat{O}_s \e^{-\ii t H_s}
    + \sum_{\substack{m=2\\m \text{~is even}}}^{+\infty}
   \int_{-t\leqslant \boldsymbol{s} \leqslant t}
    \dd \boldsymbol{s} \sum_{\mathfrak{q}\in\mathcal{Q}(\boldsymbol{s})}
    (-1)^{\#\{\boldsymbol{s}<0\}} \ii^m
    \mathcal{U}^{(0)} (-t,\boldsymbol{s},t)
    \prod_{(s_j,s_k)\in\mathfrak{q}}
    B(s_j,s_k)  \\
    = & \ \e^{\ii t H_s} \hat{O}_s \e^{-\ii t H_s}
    +  \int_{-t\leqslant s_1 \leqslant s_2 \leqslant t} \dd s_1 \dd s_2  (-1)^{\#\{\bs<0\}}\ii^2   \mathcal{U}^{(0)} (-t,s_1,s_2,t) B(s_1,s_2)\\
    & + \int_{-t\leqslant s_1 \leqslant s_2 \leqslant s_3 \leqslant s_4 \leqslant t} \dd s_1 \dd s_2 \dd s_3 \dd s_4  (-1)^{\#\{\bs<0\}}\ii^4   \mathcal{U}^{(0)} (-t,s_1,s_2,s_3,s_4,t) B(s_1,s_2)B(s_3,s_4)\\
    & + \int_{-t\leqslant s_1 \leqslant s_2 \leqslant s_3 \leqslant s_4 \leqslant t} \dd s_1 \dd s_2 \dd s_3 \dd s_4  (-1)^{\#\{\bs<0\}}\ii^4   \mathcal{U}^{(0)} (-t,s_1,s_2,s_3,s_4,t) B(s_1,s_3)B(s_2,s_4)\\    
    & + \int_{-t\leqslant s_1 \leqslant s_2 \leqslant s_3 \leqslant s_4 \leqslant t} \dd s_1 \dd s_2 \dd s_3 \dd s_4  (-1)^{\#\{\bs<0\}}\ii^4   \mathcal{U}^{(0)} (-t,s_1,s_2,s_3,s_4,t) B(s_1,s_4)B(s_2,s_3)+\cdots. 
    \end{split}
    \label{fullPropagator Dyson}
\end{equation}

The above expansion of Dyson series can also be expressed by the following diagrammatic representation: 
\begin{equation} \label{fullpropagator diagram}
   \begin{split}
&\begin{tikzpicture}[anchor=base, baseline] 
 \fill [black] (-1.4,-0.13) rectangle (1.4,0.13);
 \draw [thick] (-1.4,-0.18)--(-1.4,0.18); \draw [thick] (1.4,-0.18)--(1.4,0.18);
 \node [below] at (-1.4,-0.18) {$-t$}; 
  \node [below] at (1.4,-0.18) {$t$}; 
\end{tikzpicture}
 = 
 \begin{tikzpicture}[anchor=base, baseline] 
 \draw [thick] (-1.4,0) -- (1.4,0);
 \draw [thick] (-1.4,-0.1)--(-1.4,0.1); \draw [thick] (1.4,-0.1)--(1.4,0.1);
 \node [below] at (-1.4,-0.1) {$-t$}; 
  \node [below] at (1.4,-0.1) {$t$}; 
\end{tikzpicture}
+ 
 \begin{tikzpicture}[anchor=base, baseline] 
 \draw [thick] (-1.4,0) -- (1.4,0);
 \draw [thick] (-1.4,-0.1)--(-1.4,0.1); \draw [thick] (1.4,-0.1)--(1.4,0.1);
 \node [below] at (-1.4,-0.1) {$-t$}; 
  \node [below] at (1.4,-0.1) {$t$}; 
  \draw[-] (-1,0) to[bend left=60] (1,0);
  \draw plot[only marks,mark =*, mark options={color=black, scale=0.5}]coordinates {(-1,0)(1,0)};
  \node [below] at (-1,0) {$s_1$};
    \node [below] at (1,0) {$s_2$};
 \end{tikzpicture}
 + \\
& \hspace{2cm} + \begin{tikzpicture}[anchor=base, baseline] 
 \draw [thick] (-1.4,0) -- (1.4,0);
 \draw [thick] (-1.4,-0.1)--(-1.4,0.1); \draw [thick] (1.4,-0.1)--(1.4,0.1);
 \node [below] at (-1.4,-0.1) {$-t$}; 
  \node [below] at (1.4,-0.1) {$t$}; 
  \draw[-] (-1,0) to[bend left=75] (-0.2,0);
    \draw[-] (0.2,0) to[bend left=75] (1,0);
  \draw plot[only marks,mark =*, mark options={color=black, scale=0.5}]coordinates {(-1,0) (-0.2,0) (0.2,0) (1,0)};
  \node [below] at (-1,0) {$s_1$};  \node [below] at (-0.2,0) {$s_2$};   \node [below] at (0.2,0) {$s_3$};\node [below] at (1,0) {$s_4$};
\end{tikzpicture}
+ \begin{tikzpicture}[anchor=base, baseline] 
 \draw [thick] (-1.4,0) -- (1.4,0);
 \draw [thick] (-1.4,-0.1)--(-1.4,0.1); \draw [thick] (1.4,-0.1)--(1.4,0.1);
 \node [below] at (-1.4,-0.1) {$-t$}; 
  \node [below] at (1.4,-0.1) {$t$}; 
  \draw[-] (-1,0) to[bend left=75] (0.2,0);
    \draw[-] (-0.2,0) to[bend left=75] (1,0);
  \draw plot[only marks,mark =*, mark options={color=black, scale=0.5}]coordinates {(-1,0) (-0.2,0) (0.2,0) (1,0)};
  \node [below] at (-1,0) {$s_1$};  \node [below] at (-0.2,0) {$s_2$};   \node [below] at (0.2,0) {$s_3$};\node [below] at (1,0) {$s_4$};
\end{tikzpicture}
+ \begin{tikzpicture}[anchor=base, baseline] 
 \draw [thick] (-1.4,0) -- (1.4,0);
 \draw [thick] (-1.4,-0.1)--(-1.4,0.1); \draw [thick] (1.4,-0.1)--(1.4,0.1);
 \node [below] at (-1.4,-0.1) {$-t$}; 
  \node [below] at (1.4,-0.1) {$t$}; 
  \draw[-] (-1,0) to[bend left=75] (1,0);
    \draw[-] (-0.4,0) to[bend left=75] (0.4,0);
  \draw plot[only marks,mark =*, mark options={color=black, scale=0.5}]coordinates {(-1,0) (-0.4,0) (0.4,0) (1,0)};
  \node [below] at (-1,0) {$s_1$};  \node [below] at (-0.4,0) {$s_2$};   \node [below] at (0.4,0) {$s_3$};\node [below] at (1,0) {$s_4$};
\end{tikzpicture}+ \cdots 
   \end{split}
\end{equation}
where the propagator $G(-t,t)$ on left-hand side is denoted by a bold line. On the right-hand side, each diagram represents an integral of the form
\begin{displaymath}
\int_{-t\leqslant \bs \leqslant t} \dd \bs  (-1)^{\#\{\bs<0\}}\ii^m   \mathcal{U}^{(0)} (-t,\bs,t) B(\cdot,\cdot)B(\cdot,\cdot)\cdots.
\end{displaymath}
In detail, each thin line segment connecting two adjacent time points
  $s_j$ and $s_{j+1}$ means a bare propagator
  $G_s^{(0)}(s_j,s_{j+1})$ appearing in the system associated functional $\mathcal{U}^{(0)}$; 
  a black dot at $s_j$ introduces a perturbation operator $\sgn(s_j) \ii W_s$ in $\mathcal{U}^{(0)}$; 
  an arc connecting $s_j$ and $s_k$ stands for the two-point correlation $B(s_j,s_k)$.
In particular, the first term $\e^{\ii t H_s} \hat{O}_s \e^{-\ii t H_s}$ on the right-hand side of \eqref{fullPropagator} is denoted by a thin line with no dots or arcs (the first diagram on the right-hand side of \eqref{fullpropagator diagram}).  
  
In general, the integral over $m$ time points
\begin{equation}
\int_{-t\leqslant \boldsymbol{s} \leqslant t}
    \dd \boldsymbol{s} \sum_{\mathfrak{q}\in\mathcal{Q}(\boldsymbol{s})}
    (-1)^{\#\{\boldsymbol{s}<0\}} \ii^m
    \mathcal{U}^{(0)} (-t,\boldsymbol{s},t)
    \prod_{(s_j,s_k)\in\mathfrak{q}}
    B(s_j,s_k),
    \label{integral}
\end{equation}
which appears in the expansion of $G(-t,t)$ (see \eqref{fullPropagator Dyson}), is the sum of all possible combinations of the pairings of $(s_1,\cdots,s_m)$, which includes $(m-1)!!$ diagrams in total. For example, when $m = 4$, the integral \eqref{integral} is diagrammatically represented by
\begin{displaymath}
\begin{tikzpicture}[anchor=base, baseline] 
 \draw [thick] (-1.4,0) -- (1.4,0);
 \draw [thick] (-1.4,-0.1)--(-1.4,0.1); \draw [thick] (1.4,-0.1)--(1.4,0.1);
 \node [below] at (-1.4,-0.1) {$-t$}; 
  \node [below] at (1.4,-0.1) {$t$}; 
  \draw[-] (-1,0) to[bend left=75] (-0.2,0);
    \draw[-] (0.2,0) to[bend left=75] (1,0);
  \draw plot[only marks,mark =*, mark options={color=black, scale=0.5}]coordinates {(-1,0) (-0.2,0) (0.2,0) (1,0)};
  \node [below] at (-1,0) {$s_1$};  \node [below] at (-0.2,0) {$s_2$};   \node [below] at (0.2,0) {$s_3$};\node [below] at (1,0) {$s_4$};
\end{tikzpicture}
+ \begin{tikzpicture}[anchor=base, baseline] 
 \draw [thick] (-1.4,0) -- (1.4,0);
 \draw [thick] (-1.4,-0.1)--(-1.4,0.1); \draw [thick] (1.4,-0.1)--(1.4,0.1);
 \node [below] at (-1.4,-0.1) {$-t$}; 
  \node [below] at (1.4,-0.1) {$t$}; 
  \draw[-] (-1,0) to[bend left=75] (0.2,0);
    \draw[-] (-0.2,0) to[bend left=75] (1,0);
  \draw plot[only marks,mark =*, mark options={color=black, scale=0.5}]coordinates {(-1,0) (-0.2,0) (0.2,0) (1,0)};
  \node [below] at (-1,0) {$s_1$};  \node [below] at (-0.2,0) {$s_2$};   \node [below] at (0.2,0) {$s_3$};\node [below] at (1,0) {$s_4$};
\end{tikzpicture}
+ \begin{tikzpicture}[anchor=base, baseline] 
 \draw [thick] (-1.4,0) -- (1.4,0);
 \draw [thick] (-1.4,-0.1)--(-1.4,0.1); \draw [thick] (1.4,-0.1)--(1.4,0.1);
 \node [below] at (-1.4,-0.1) {$-t$}; 
  \node [below] at (1.4,-0.1) {$t$}; 
  \draw[-] (-1,0) to[bend left=75] (1,0);
    \draw[-] (-0.4,0) to[bend left=75] (0.4,0);
  \draw plot[only marks,mark =*, mark options={color=black, scale=0.5}]coordinates {(-1,0) (-0.4,0) (0.4,0) (1,0)};
  \node [below] at (-1,0) {$s_1$};  \node [below] at (-0.4,0) {$s_2$};   \node [below] at (0.4,0) {$s_3$};\node [below] at (1,0) {$s_4$};
\end{tikzpicture}.
\end{displaymath}
To evaluate $G(-t,t)$ using the series expansion \eqref{fullPropagator Dyson}, one natural idea is to truncate this Dyson series at $m = \bar{M}$ for a sufficiently large even integer $\bar{M}$ and compute the high-dimensional integrals using Monte Carlo method, resulting in the method of bare dQMC \cite{prokof1998polaron,werner2009diagrammatic} which evaluates $G(-t,t)$ by 
\begin{equation}\label{bare dqmc}
  G(-t, t)  \approx  \e^{\ii t H_s} \hat{O}_s \e^{- \ii  t H_s} + \sum^{\bar{M}}_{\substack{m=2 \\ m \text{~is even}}} \frac{1}{\Ms^{(m)}} \  \sum^{\Ms^{(m)}}_{i = 1} \  \frac{(2t)^m}{m!}  \cdot    \ii^{m} 
(-1)^{\#\{\bs_m^{(i)} < 0\}}  \mathcal{U}^{(0)}(-t, \bs_m^{(i)} , t) 
    \mathcal{L}_b(\bs_m^{(i)}) 
\end{equation}
where for simplicity, the Monte Carlo samples $\{\bs_m^{(i)}\}_{i=1}^{\Ms^{(m)}}$ are drawn independently according to the uniform distribution $U(-t \leqslant \bs_m \leqslant t)$ with $\bs_m = (s_1,\cdots,s_m)$, and $\frac{(2t)^m}{m!}$ in the formula above is the volume of the domain of integration $\{-t \leqslant \bs_m \leqslant t\}$. One may also apply other sampling strategy in general. 

Alternatively, we can also solve $G(-t,t)$ via the following integro-differential equation, which is obtained by taking time derivative on both sides of \eqref{fullPropagator} (see \cite{cai2022fast} for the derivation): 
\begin{equation}
    \frac{\dd }{\dd t}G(-t,t) = \ii [H_s,G(-t,t)] + W_s \mathcal{K}(-t,t) + \left(W_s \mathcal{K}(-t,t)\right)^\dagger, \qquad \forall t > 0,
  \label{integro_differential_eq}
\end{equation}
where $\mathcal{K}(-t,t)$ stands for the series:
\begin{equation}\label{K dyson}
\mathcal{K}(-t,t) = \mathcal{K}^{(0)}(-t,t) \coloneqq \sum_{\substack{m=1 \\ m \text{~is odd}}}^{+\infty}  \ii^{m+1} \int_{-t\leqslant \bs \leqslant t} \dd \bs  (-1)^{\#\{\bs < 0\}} \mathcal{U}^{(0)}(-t, \bs , t)  \Ls_b(\bs,t)  
\end{equation}
and the initial condition is given by $G(0,0)=\hat{O}_s$ according to definition \eqref{fullPropagator}. Here we use the notation $\mathcal{K}^{(0)}(-t,t)$ to denote the series with expression \eqref{K dyson} in order to distinguish it from other formulations of $\mathcal{K}(-t,t)$ to be introduced later in this paper.
Note that now in the function $\Ls_b$, the last time point is fixed as $t$, and thus the series on the right-hand side now sums over odd $m$. Diagrammatically, this series is written as  
\begin{equation} \label{diagram dyson}
 \mathcal{K}^{(0)}(-t,t) = \begin{tikzpicture}[anchor=base, baseline] 
 \draw [thick] (-1.4,0) -- (1.4,0);
 \draw [thick] (-1.4,-0.1)--(-1.4,0.1); \draw [thick] (1.4,-0.1)--(1.4,0.1);
 \node [below] at (-1.4,-0.1) {$-t$}; 
  \node [below] at (1.4,-0.1) {$t$}; 
  \draw[-] (-1,0) to[bend left=60] (1.4,0);
  \draw plot[only marks,mark =*, mark options={color=black, scale=0.5}]coordinates {(-1,0)};
  \node [below] at (-1,0) {$s_1$};
 \end{tikzpicture}
+ \begin{tikzpicture}[anchor=base, baseline] 
 \draw [thick] (-1.4,0) -- (1.4,0);
 \draw [thick] (-1.4,-0.1)--(-1.4,0.1); \draw [thick] (1.4,-0.1)--(1.4,0.1);
 \node [below] at (-1.4,-0.1) {$-t$}; 
  \node [below] at (1.4,-0.1) {$t$}; 
  \draw[-] (-1,0) to[bend left=75] (-0.2,0);
    \draw[-] (0.6,0) to[bend left=75] (1.4,0);
  \draw plot[only marks,mark =*, mark options={color=black, scale=0.5}]coordinates {(-1,0) (-0.2,0) (0.6,0)};
  \node [below] at (-1,0) {$s_1$};  \node [below] at (-0.2,0) {$s_2$};   \node [below] at (0.6,0) {$s_3$};
\end{tikzpicture}
+ \begin{tikzpicture}[anchor=base, baseline] 
 \draw [thick] (-1.4,0) -- (1.4,0);
 \draw [thick] (-1.4,-0.1)--(-1.4,0.1); \draw [thick] (1.4,-0.1)--(1.4,0.1);
 \node [below] at (-1.4,-0.1) {$-t$}; 
  \node [below] at (1.4,-0.1) {$t$}; 
  \draw[-] (-1,0) to[bend left=75] (0.6,0);
    \draw[-] (-0.2,0) to[bend left=75] (1.4,0);
  \draw plot[only marks,mark =*, mark options={color=black, scale=0.5}]coordinates {(-1,0) (-0.2,0) (0.6,0)};
  \node [below] at (-1,0) {$s_1$};  \node [below] at (-0.2,0) {$s_2$};   \node [below] at (0.6,0) {$s_3$};
\end{tikzpicture}
+ \begin{tikzpicture}[anchor=base, baseline] 
 \draw [thick] (-1.4,0) -- (1.4,0);
 \draw [thick] (-1.4,-0.1)--(-1.4,0.1); \draw [thick] (1.4,-0.1)--(1.4,0.1);
 \node [below] at (-1.4,-0.1) {$-t$}; 
  \node [below] at (1.4,-0.1) {$t$}; 
  \draw[-] (-1,0) to[bend left=75] (1.4,0);
    \draw[-] (-0.2,0) to[bend left=75] (0.6,0);
  \draw plot[only marks,mark =*, mark options={color=black, scale=0.5}]coordinates {(-1,0) (-0.2,0) (0.6,0)};
  \node [below] at (-1,0) {$s_1$};  \node [below] at (-0.2,0) {$s_2$};   \node [below] at (0.6,0) {$s_3$};
\end{tikzpicture} + \cdots .
\end{equation}
With this integro-differential equation, one can solve $G(-t,t)$ iteratively by Runge-Kutta method. Throughout this paper, we apply the second-order Heun's method, which reads
\begin{equation}
  \label{heun}
  \begin{split}
     & G_\star =  G_n + \Delta t \left( \ii[H_s,G_n] + W_s \mathcal{K}_n + \left(  W_s \mathcal{K}_n\right)^\dagger     \right), \\  
    & G_{\star\star} = G_\star + \Delta t \left( \ii[H_s,G_\star] + W_s \mathcal{K}_{n+1} + \left(  W_s \mathcal{K}_{n+1}\right)^\dagger     \right),  \\
    & G_{n+1} = \frac{1}{2}\left(  G_n + G_{\star\star}  \right)
    \end{split}
\end{equation}
where $\Delta t$ is the time step length. The initial condition is given by $G_0 = \hat{O}_s$. $G_n$ is the approximated solution of $G(-t_n,t_n)$ with $t_n = n\dt$, and $\mathcal{K}_n$ approximates $\mathcal{K}(-t_n,t_n)$ again by truncating the series and replacing the integrals by the average of Monte Carlo samples:  
\begin{equation}
 \label{Kn}
\mathcal{K}_n =  \sum_{\substack{m=1 \\ m \text{~is odd}}}^{\bar{M}} \frac{1}{\bar{\Ms}^{(m)}_n}  \   \sum^{\bar{\Ms}^{(m)}_n}_{i = 1} \ \frac{(2t_n)^m}{m!} \cdot    \ii^{m+1} 
(-1)^{\#\{\bs_m^{(i)} < 0\}}  \mathcal{U}^{(0)}(-t_n, \bs_m^{(i)} , t_n) 
    \mathcal{L}_b(\bs_m^{(i)},t_n) 
\end{equation}
In $n$th time step, the samples $\{\bs_m^{(i)}\}_{i=1}^{\bar{\Ms}^{(m)}_n}$ are drawn from uniform distribution $U(-t_n \leqslant \bs_m \leqslant t_n)$. In practice, the number of samples $\bar{\Ms}^{(m)}_n$ is assumed to be proportional to the integral of the absolute value of bath influence functional
\begin{equation*}
    \bar{\mathcal{M}}^{(m)}_n \propto \int_{-t_n \leqslant \bs_m \leqslant t_n} \dd \bs |\Ls_b(\bs,t_n)|  = \int_{-t_n \leqslant \bs_m \leqslant t_n} \dd \bs  \left| \sum_{\mathfrak{q} \in \mathcal{Q}(\bs,t_n)} \prod_{(s_j,s_k) \in \mathfrak{q}} B(s_j,s_k) \right|.
\end{equation*}
The two-point correlation is then replaced by some empirical constant $\mathcal{B}\in(0,\max\vert B \vert)$, which yields
\begin{equation}\label{Mnm1}
     \bar{\mathcal{M}}^{(m)}_n =    \mathcal{M}_0 \cdot  \left\vert \{-t_n \leqslant \bs_m \leqslant t_n\}\right\vert \cdot m!! \mathcal{B}^{\frac{m+1}{2}} 
=   \mathcal{M}_0 \cdot \frac{(2t_n)^m}{(m-1)!!} \cdot \mathcal{B}^{\frac{m+1}{2}}
\end{equation}
where $\mathcal{M}_0$ is some positive constant. In practice, we may first assign the value of $\mathcal{M}_0$, then the general $\bar{\mathcal{M}}^{(m)}_n$ is set to be the nearest integer of the value \eqref{Mnm1}.

Compared with the bare dQMC \eqref{bare dqmc}, the numerical methods based on \eqref{integro_differential_eq} preserve the Hermitian property of $G(-t,t)$ as the right-hand side of \eqref{integro_differential_eq} is always Hermitian, while this is not guaranteed by bare dQMC which directly evaluates the Dyson series \eqref{fullPropagator}. 
With the time evolution of $G(-t,t)$ provided by the numerical scheme, we are ready to propose our fast algorithms based on calculation reuse which will be detailed in the following section.

% the numerical solutions obtained by \eqref{bare dqmc} suffers the notorious sign problem, meaning that the variance of the solutions grows double exponentially fast with respect to time \cite{cai2020numerical}. 
% This requires us to increase the number of samples $\Ms$ accordingly in order for the accuracy of the numerical solutions.

\subsection{Numerical method based on reuse of calculations}
\label{Sec_Dyson_reuse}
% When we use either \eqref{fullPropagator} 
%  or \eqref{integro_differential_eq}
%  to compute the full propagator,
%  the number of samples should increase very fast when $t$ becomes larger.
% In high dimensional space,
%  the volume of the simplex increases very fast
%  as $t$ increases.
% In \cite{cai2022fast},
%  we propose a fast algorithm to reuse the samples
%  of the bath influence functional $\mathcal{L}_b$.
% In this paper,
%  we further extend the idea of sample reusing
%  to the system part $\mathcal{U}^{(0)}$.
% In order to reuse the Monte Carlo samples,
%  we first study the invariant properties of $G_s^{(0)}$ and $B$:
The main idea of the reuse of calculations concentrates on the following relations of the bare propagator $G_s^{(0)}(\cdot,\cdot)$ and two-point correlation $B(\cdot,\cdot)$:
 \begin{proposition}
 Given any $s_\ii \leqslant s_\ff$ and $\Delta t \geqslant 0$, we have 
 \begin{itemize}
     \item The bare propagator \eqref{Gs0_definition} associated with the system satisfies
     \begin{equation}
     \label{properties_Gs0}
     G_s^{(0)}(s_\ii,s_\ff)
        = \left\{ \begin{array}{l l}
         G_s^{(0)} (s_\ii -\dt, s_\ff-\dt), & \text{for~} s_\ii \leqslant s_\ff < 0, \\
         G_s^{(0)} (s_\ii +\dt, s_\ff+\dt), & \text{for~}  0 \leqslant s_\ii \leqslant s_\ff
              \end{array} \right.
     \end{equation}
     \item The two-point correlation \eqref{bath_function_B} satisfies
     \begin{equation}
     \label{properties_B}
       B(s_\ii,s_\ff)
        = \left\{ \begin{array}{l l}
       B(s_\ii-\dt,s_\ff-\dt), & \text{for~} s_\ii \leqslant s_\ff < 0, \\
        B(s_\ii+\dt,s_\ff+\dt), & \text{for~}  0 \leqslant s_\ii \leqslant s_\ff, \\
        B(s_\ii-\dt,s_\ff+\dt), & \text{for~}  s_\ii < 0 \leqslant s_\ff. 
              \end{array} \right.
     \end{equation}
\end{itemize}
 \end{proposition}
 The rigorous proof for the relations above can be found in the appendix of \cite{cai2022fast}. These relations describe the invariance of $G_s^{(0)}(s_\ii,s_\ff)$ and $B(s_\ii,s_\ff)$ after ``shifting" or ``stretching" $(s_\ii,s_\ff)$ by the time step length $\Delta t$. We hope to utilize such properties to reuse the previous calculations in the scheme \eqref{heun}. In fact, one can easily use \eqref{properties_B} to derive the following property.
 \begin{proposition}
 Given $\bs = (s_1,\cdots,s_m)$ and $t>0$ such that $-t \leqslant s_1 \leqslant \cdots s_m \leqslant t$ and $\Delta t \geqslant 0$, define
\begin{equation} \label{I}
\mathcal{I}_{\dt}(\boldsymbol{s}) = (\tilde{s}_1,\cdots,\tilde{s}_m) \text{~where~}    \tilde{s}_k = \begin{cases}
        s_k + \dt, &\text{if } s_k \geqslant 0 \\
        s_k - \dt, &\text{if } s_k < 0
    \end{cases},
\end{equation}
we have 
\begin{equation}
     \label{shift_properties_Lb}
    \mathcal{L}_b(\mathcal{I}_{\dt}(\boldsymbol{s}),t+\dt)
    = \mathcal{L}_b(\boldsymbol{s},t).
\end{equation}
 \end{proposition}
In \cite{cai2022fast}, fast algorithms that avoid duplicate computations of $\Ls_b$ have been proposed based on such invariance of the bath influence functional.
Specifically, suppose we have sampled time sequences $\{\bs_n^{(i)}\}_{i=1}^{\bar{\Ms}_n}$ and calculated the corresponding bath influence functionals $\{\Ls_b(\bs_n^{(i)},t_n)\}_{i=1}^{\bar{\Ms}_n}$ for $\mathcal{K}_n$ \eqref{Kn}. Then at the next time step, we can use $\{\mathcal{I}_{\Delta t}(\bs_n^{(i)})\}_{i=1}^{\bar{\Ms}_n}$ as part of the samples to compute $\mathcal{K}_{n+1}$ so that $\{\Ls_b(\mathcal{I}_{\Delta t}(\bs_n^{(i)}),t_{n+1})\}_{i=1}^{\bar{\Ms}_n}$ can be directly obtained by \eqref{shift_properties_Lb}. One may refer to \cite{cai2022fast} for more details.
 
However, due to the existence of the observable $\hat{O}_s$ in the definition of $G_s^{(0)}(s_\ii,s_\ff)$ \eqref{Gs0_definition} when $s_\ii < 0 \leqslant s_\ff$, there is no similar invariance for the system associated functional $\mathcal{U}^{(0)}$. Consequently, in the reuse algorithm developed in \cite{cai2022fast}, $\mathcal{U}^{(0)}$ needs to be evaluated with all samples at each time step. In the rest of this section, we propose a new fast algorithm for Dyson series which can achieve the reuse of $\mathcal{K}^{(0)}(-t,t)$ as a whole. To do this, we introduce the linear basis of $G_s^{(0)}(s_\ii,s_\ff)$, $\mathcal{U}^{(0)}(-t,\bs,t)$ and $\mathcal{K}^{(0)}(-t,t)$:
% This means that all samples need to be saved to evaluate $\mathcal{U}^{(0)}$ at each time step, leading to a high memory cost for long time simulations. In this work, we aim to remove this memory constraint.
\begin{definition}
\label{defineBases}
    Given two states $\ket{i},\ket{j} \in \{\ket{-1},\ket{1}\}$ and $t \geqslant 0$, we 
     define
    \begin{equation}
     \label{basis}
        \mathscr{G}_{ij}^{(0)} (s_\ii,s_\ff)
        = \begin{cases}
            G_s^{(0)}(s_\ii,s_\ff),
            &\text{ if }s_\ii \leqslant s_\ff < 0
            \text{ or }0\leqslant s_\ii \leqslant s_\ff \\
            \e^{\ii s_\ff H_s} \dyad{i}{j} \e^{\ii s_\ii H_s},
            &\text{ if } s_\ii < 0 \leqslant s_\ff
        \end{cases},
    \end{equation}
        \begin{equation}\label{U0 basis}
    \mathscr{U}_{ij}^{(0)}(-t,\boldsymbol{s},t)
    = \mathscr{G}_{ij}^{(0)}(s_m,t) W_s \mathscr{G}_{ij}^{(0)}(s_{m-1},s_m)
    W_s \cdots W_s 
    \mathscr{G}_{ij}^{(0)}(s_1,s_2)
    W_s
    \mathscr{G}_{ij}^{(0)}(-t,s_1),
    \end{equation}
    and
        \begin{equation}
        \label{Sscr}
        \mathscr{K}_{ij}^{(0)}(-t,t)
        = \sum^{+\infty}_{\substack{m=1\\m \text{ is odd}}} 
        \ii^{m+1}
        \int_{-t\leqslant \boldsymbol{s} \leqslant t}
        \dd \boldsymbol{s} (-1)^{\#\{\boldsymbol{s}<0\}}
        \mathscr{U}_{ij}^{(0)}(-t,\boldsymbol{s},t) \mathcal{L}_b(\boldsymbol{s},t).
  \end{equation}
\end{definition}
With the definitions above, the bare propagator $G_s^{(0)}(s_k,s_{k+1})$ for $s_k < 0 \leqslant s_{k+1}$ can now be expressed as the linear combination 
\begin{equation}\label{bare propagator basis}
    G_{s}^{(0)}(s_k,s_{k+1}) =  \e^{\ii s_{k+1} H_s} \hat{O}_s\e^{\ii s_k H_s}
    = \sum_{i,j} a_{ij} \mathscr{G}_{ij}^{(0)}(s_k,s_{k+1})
\end{equation}
 where the coefficients $a_{ij}$ satisfies that 
\begin{equation} \label{aij}
 \hat{O}_s = \displaystyle \sum_{i,j} a_{ij} \dyad{i}{j}.
\end{equation}
Inserting \eqref{bare propagator basis} into \eqref{U0}, the system associated functional $\mathcal{U}^{(0)}(-t,\bs,t)$ is written as 
linear combination of $\mathscr{U}_{ij}^{(0)}$:
\begin{equation*}
    \mathcal{U}^{(0)} (-t,\boldsymbol{s},t)
    = \sum_{i,j} a_{ij} \mathscr{U}_{ij}^{(0)}(-t,\boldsymbol{s},t),
\end{equation*}
%  \begin{equation*}
%         \begin{split}
%         &\mathcal{U}^{(0)}(-t,\boldsymbol{s},t) \\
%         =  & \  G_{s}^{(0)}(s_m,t) W_s 
%             \cdots 
%             G_s^{(0)}(s_{k-1},s_{k}) W_s
%             G_s^{(0)}(s_k,s_{k+1})
%             W_s G_s^{(0)}(s_{k+1},s_{k+2})
%             \cdots
%             W_s
%             G_{s}^{(0)}(-t,s_1)  \\ 
%             = & \  
%             \mathscr{G}_{s}^{(0)}(s_m,t) W_s 
%             \cdots 
%             \mathscr{G}_s^{(0)}(s_{k-1},s_{k}) W_s
%             \sum_{i,j} a_{ij} \mathscr{G}_{ij}^{(0)}(s_k,s_{k+1})
%             W_s \mathscr{G}_s^{(0)}(s_{k+1},s_{k+2})
%             \cdots
%             W_s
%             \mathscr{G}_{s}^{(0)}(-t,s_1) \\
%             = & \  
%             \sum_{i,j} a_{ij}\mathscr{U}^{(0)}_{ij}(-t,\boldsymbol{s},t),
%         \end{split}
% \end{equation*}
and thus we may express  
\begin{equation*}
    \mathcal{K}^{(0)}(-t,t) = \sum_{i,j} a_{ij} \mathscr{K}_{ij}^{(0)}(-t,t). 
\end{equation*}

At this point, we can rewrite the original integro-differential equation \eqref{integro_differential_eq} as
\begin{equation}
    \frac{\dd}{\dd t} G(-t,t) 
    = \ii \left[H_s,G(-t,t)\right]
    + \sum_{i,j} a_{ij} W_s\mathscr{K}_{ij}^{(0)}(-t,t)
    + \left(\sum_{i,j} a_{ij} W_s\mathscr{K}_{ij}^{(0)}(-t,t)\right)^\dagger,
    \label{intego_differential_compact_form}
\end{equation}
which can again be numerically solved using a discretization similar to \eqref{heun}. Under such scheme, we aim to find some recurrence relation of the basis $\mathscr{K}_{ij}^{(0)}(-t,t)$ so that it can be computed iteratively.
Below we begin with the iterative relation of $\mathscr{U}_{ij}^{(0)}$:
%The following relation gives the ``invariance" of the system associated functional, which is comparable to \eqref{shift_properties_Lb}:
\begin{lemma}
\label{lemma:Uscr_shift}
   Given two states $\ket{i},\ket{j} \in \{\ket{-1},\ket{1}\}$ and $\dt \geqslant 0$, we have 
   \begin{equation}
            \mathscr{U}_{ij}^{(0)}(-t-\dt,\mathcal{I}_{\dt}(\boldsymbol{s}),t+\dt)
            = \sum_{k,l} b_{kl}^{ij} \mathscr{U}_{kl}^{(0)}(-t,\boldsymbol{s},t),
            \label{shift_properties_Uscr}
        \end{equation}
   where the coefficients $b_{kl}^{ij}$ satisfy that 
   \begin{equation} \label{bij}
       \e^{\ii \dt H_s} \dyad{i}{j} \e^{-\ii \dt H_s} = \sum_{k,l} b_{kl}^{ij}\dyad{k}{l}.
   \end{equation}
   \end{lemma}
       \begin{proof}
       For the bare propergator $\mathscr{G}_{ij}^{(0)}(s_k,s_{k+1})$ with $s_k < 0 \leqslant s_{k+1}$,
        directly using the definition \eqref{basis} yields that 
       \begin{equation} \label{shift_property}
    \begin{split}
        \mathscr{G}_{ij}^{(0)}(s_k-\dt,s_{k+1}+\dt)
        &= \e^{\ii (s_{k+1}+\dt) H_s} \dyad{i}{j} \e^{\ii (s_k -\dt) H_s} 
        = \e^{\ii s_{k+1} H_s} 
        \e^{\ii \dt H_s} \dyad{i}{j} \e^{-\ii \dt H_s}
        \e^{\ii s_k H_s} \\
        &= \e^{\ii s_{k+1} H_s} 
        \sum_{k,l} b_{kl}^{ij} \dyad{k}{l}
        \e^{\ii s_k H_s} = \sum_{k,l} b_{kl}^{ij} \mathscr{G}_{kl}^{(0)}(s_k,s_{k+1}).
    \end{split}
    \end{equation}
Insert the formula above into \eqref{U0 basis}, due to \eqref{properties_Gs0},
 we arrive at 
        \begin{equation*}
            \begin{split}
                &\mathscr{U}_{ij}^{(0)}(-t-\dt,\mathcal{I}_{\dt}(\boldsymbol{s}),t+\dt) \\
                =& \  \mathscr{G}_{ij}^{(0)}(s_m+\dt,t+\dt) W_s  \cdots 
                W_s
                \mathscr{G}_{ij}^{0}(s_{k}-\dt,s_{k+1}+\dt)
                W_s
                \cdots
                W_s
                \mathscr{G}_{ij}^{(0)}(-t-\dt,s_1-\dt) \\
                =& \ \mathscr{G}_{ij}^{(0)}(s_m,t) W_s  \cdots 
                W_s
                \left[\sum_{k,l} b_{kl}^{ij} \mathscr{G}_{kl}^{(0)}
                (s_{k},s_{k+1})\right]
                W_s
                \cdots
                W_s
                \mathscr{G}_{ij}^{(0)}(-t,s_1) 
                = \  \sum_{k,l}b_{kl}^{ij} \mathscr{U}_{kl}^{(0)}(-t,\boldsymbol{s},t). \qedhere
            \end{split}
        \end{equation*}
    \end{proof}
Note that the coefficients $b_{kl}^{ij}$ depend only on the observable, the Hamiltonian $H_s$ and the time step $\dt$. 
Therefore, coefficients $b_{kl}^{ij}$ can be computed offline for any given time step $\Delta t$. At the time step $t+\Delta t$, we now split the $m$-dimensional integral in $\mathscr{K}_{ij}^{(0)}(-t-\Delta t, t + \Delta t)$ as 
 \begin{equation}
    \label{separation_of_integral}
    \begin{split}
    & \int_{-t - \Delta t \leqslant \boldsymbol{s} \leqslant t + \Delta t}
        \dd \boldsymbol{s} (-1)^{\#\{\boldsymbol{s}<0\}}
        \mathscr{U}_{ij}^{(0)}(-t - \Delta t,\boldsymbol{s},t + \Delta t) \mathcal{L}_b(\boldsymbol{s},t +\Delta t) \\  
         & \hspace{4cm}= \int_{\mathcal{I}_{\dt}\left(
            \{ -t\leqslant \boldsymbol{s} \leqslant t\}
            \right)}
            (\text{integrand})
            + \int_{T_{t+\dt}^{(m)}}(\text{integrand})
             \end{split}
\end{equation} 
where the domains of integration on the right-hand side are defined as follows:
\begin{definition}
    Given $t \geqslant 0$, $\Delta t \geqslant 0$ and $S \subset \mathbb{R}^m$, 
    \begin{equation*}
        \mathcal{I}_{\dt}(S) \coloneqq \bigcup\limits_{\boldsymbol{s}\in S}\{\mathcal{I}_{\dt}(\boldsymbol{s})\}
    \end{equation*}
    and 
    \begin{equation*}
    T_{t+\dt}^{(m)} \coloneqq \{ -t-\dt\leqslant \boldsymbol{s}\leqslant t+\dt \ \big| \ \exists \ 
    s_j \emph{~for~} j=1,\cdots,m \emph{~such that~}  |s_j| \leqslant \Delta t 
    \}.
\end{equation*}
\end{definition}
The sum of the first integral on the right-hand side of \eqref{separation_of_integral} over all odd $m$ can be directly obtained using the previous calculations of $ \mathscr{K}_{ij}^{(0)}(-t,t)$: 
\begin{equation*}
    \begin{split}
      &  \sum^{+\infty}_{\substack{m=1\\m \text{ is odd}}} 
            \ii^{m+1} 
            \int_{\mathcal{I}_{\dt}\left(
            \{ -t\leqslant \boldsymbol{s} \leqslant t\}
            \right)}
            \dd \tilde{\boldsymbol{s}} (-1)^{\#\{\tilde{\boldsymbol{s}}<0\}}
            \mathscr{U}_{ij}^{(0)}(-t-\dt,\tilde{\boldsymbol{s}},t+\dt) \mathcal{L}_b\left(\tilde{\boldsymbol{s}},t+\dt\right) \\
            = & \  \sum^{+\infty}_{\substack{m=1\\m \text{ is odd}}} 
            \ii^{m+1} 
            \int_{-t\leqslant \boldsymbol{s} \leqslant t}
            \dd \boldsymbol{s} (-1)^{\#\{\boldsymbol{s}<0\}}
            \mathscr{U}_{ij}^{(0)}(-t-\dt,\mathcal{I}_{\dt}\left(\boldsymbol{s}\right),t+\dt) \mathcal{L}_b\left(\mathcal{I}_{\dt}\left(\boldsymbol{s}\right),t+\dt\right) \\
            = & \ \sum^{+\infty}_{\substack{m=1\\m \text{ is odd}}} 
            \ii^{m+1} 
            \int_{-t\leqslant \boldsymbol{s} \leqslant t}
            \dd \boldsymbol{s} (-1)^{\#\{\boldsymbol{s}<0\}}
            \sum_{k,l} b_{kl}^{ij}\mathscr{U}_{kl}^{(0)}(-t,\boldsymbol{s},t) \mathcal{L}_b(\boldsymbol{s},t) 
            =  \ \sum_{k,l}b_{kl}^{ij} \mathscr{K}_{kl}^{(0)}(-t,t),
    \end{split}
\end{equation*}
where we have used \eqref{shift_properties_Lb} and \eqref{shift_properties_Uscr} in the derivation. The second integral on the right-hand side of \eqref{separation_of_integral} has to be evaluated at each time step. This implies the following recurrence relation of $\mathscr{K}_{ij}^{(0)}(-t,t)$ summarized in the following theorem:
\begin{theorem}
\label{thm:Sscr_shift}
    Given $t \geqslant 0$ and $\Delta t \geqslant 0$,
    we have
    \begin{equation}
        \mathscr{K}_{ij}^{(0)}(-t-\dt,t+\dt) = \sum_{k,l} b_{kl}^{ij} \mathscr{K}_{kl}^{(0)}(-t,t) 
        + \mathscr{D}_{ij}^{\dt}(t+\dt)
        \label{Sscr_shift_property}
    \end{equation}
    where $\mathscr{D}_{ij}^{\dt}(t+\dt)$ is defined by
    \begin{equation*}
        \mathscr{D}_{ij}^{\dt}(t+\dt)
        =\sum_{\substack{m=1\\m\text{ is odd}}}^{+\infty}
        \ii^{m+1} \int_{T_{t+\dt}^{(m)}}
        \dd \boldsymbol{s} (-1)^{\#\{\boldsymbol{s} < 0\}}
        \mathscr{U}_{ij}^{(0)}(-t-\dt,\boldsymbol{s},t+\dt)
        \mathcal{L}_b(\boldsymbol{s},t+\dt).
    \end{equation*}
\end{theorem}
The recurrence relation above will be used in each time step of the Heun's scheme for \eqref{intego_differential_compact_form} to achieve reuse of calculations. Similar as \eqref{Kn}, the series $\mathscr{D}_{ij}^{\dt}$ to be newly computed is truncated by some odd integer $\bar{M}$ and the integrals are evaluated using Monte Carlo method. The details of such a procedure are described in Algorithm \ref{algo:dyson}. In the input, $N$ denotes the number of total time steps, and the number of samples $\mathcal{M}^{(m)}_n$ is similarly given as \eqref{Mnm1} with the volume of simplex $\{-t_n \leqslant \bs_m \leqslant t_n\}$ replaced by the volume of $T^{(m)}_{t_{n-1}+\Delta t}$:    
\begin{equation}\label{Mnm}
    \mathcal{M}^{(m)}_n =  \mathcal{M}_0 \cdot  \left|T^{(m)}_{t_n+\Delta t} \right|  \cdot \mathcal{B}^{\frac{m+1}{2}} \quad \text{where} \quad \left|T^{(m)}_{t_n+\Delta t} \right|  = \frac{(2t_n)^m - (2t_{n-1})^m}{(m-1)!!}. 
\end{equation}

\begin{algorithm}[ht]
  \caption{}\label{algo:dyson}
  \begin{algorithmic}[1]
\State \textbf{input}  $\Delta t$, $\bar{M}$, $N$, $\mathcal{M}^{(m)}_n$
  \medskip
    \State \label{Line2}Set $G_{0}  \gets \hat{O}_s$, $\mathcal{K}_0 \gets  \mathbf{O}_{2\times 2}$ (zero matrix), $\mathscr{D}_{ij}^{n} \gets  \mathbf{O}_{2\times 2}$, $\mathscr{K}_{ij}^0 \gets \mathbf{O}_{2\times 2}$ for $i,j=\pm 1$ and $n=1,\cdots,N$ 
      \medskip
    \State \label{Line3}Compute $a_{ij}$ according to \eqref{aij}, $b_{kl}^{ij}$ according to \eqref{bij} for $i,j,k,l=\pm 1$ 
      \medskip
  \For{$n$ from $0$ to $N-1$}  
  \medskip
    \For{odd $m$ from $1$ to $\bar{M}$}  
  \medskip
  \State Sample $\{\bs^{(k)}_{n+1}\}_{k=1}^{\mathcal{M}^{(m)}_{n+1}}$ from uniform distribution $U\left( T^{(m)}_{t_{n}+\Delta t}\right)$
  \medskip
  \State \label{Line7}Set $\mathscr{S}^{(k)}_{ij} \gets  \left|T^{(m)}_{t_n+\Delta t} \right| 
        \cdot \ii^{m+1} (-1)^{\#\{\bs_{n+1}^{(k)} < 0\}}
        \mathscr{U}_{ij}^{(0)}(-t_{n+1},\bs_{n+1}^{(k)},t_{n+1})
        \mathcal{L}_b(\bs_{n+1}^{(k)},t_{n+1})$
  \medskip
  \State\label{Line8}Set $\mathscr{D}_{ij}^{n+1}
        \gets    \mathscr{D}_{ij}^{n+1} + \frac{1}{\mathcal{M}^{(m)}_{n+1}} \displaystyle \sum_{k=1}^{\mathcal{M}^{(m)}_{n+1}} \mathscr{S}^{(k)}_{ij}$
      \medskip
     \EndFor
      \medskip
    \State Set  $\mathscr{K}_{ij}^{n+1} \gets  \sum_{k,l} b_{kl}^{ij} \mathscr{K}_{kl}^{n} 
        + \mathscr{D}_{ij}^{n+1}$ and $ \mathcal{K}_{n+1} \gets  \sum_{i,j} a_{ij} W_s\mathscr{K}_{ij}^{n+1}$
          \medskip
     \State Compute $G_{n+1}$ by scheme \eqref{heun} 
     \medskip
     \EndFor
   \medskip
  \State \textbf{return} $G_{n}$ for $n = 1,\cdots, N$
   \end{algorithmic}
\end{algorithm}

 To end this section, let us have a brief discussion on the computational and memory cost of Algorithm \ref{algo:dyson}. In particular, we compare it with the direct implementation of Heun's scheme \eqref{heun} without using the recurrence relation, as well as the existing reuse algorithm in \cite{cai2022fast} which avoids duplicate calculations of the bath influence functionals.
 
 Below we list the total computational cost on the Monte Carlo integration in the three approaches: 
  \begin{itemize}
 \item The scheme \eqref{heun} and \eqref{Kn}:
 \begin{equation} \label{eq:algo1}
 \sum_{\substack{m =1\\m \text{~is odd}}}^{\bar{M}} \sum_{n=1}^N \bar{\mathcal{M}}_n^{(m)} [C_s(m)+ C_b(m)].
 \end{equation}
 \item The fast algorithm in \cite{cai2022fast}:
 \begin{equation} \label{eq:algo2}
 \sum_{\substack{m =1\\m \text{~is odd}}}^{\bar{M}} \sum_{n=1}^N [\bar{\mathcal{M}}_n^{(m)} C_s(m) + \mathcal{M}_n^{(m)} C_b(m)].
 \end{equation}
 \item The algorithm in this paper: 
 \begin{equation} \label{eq:algo3}
 \sum_{\substack{m =1\\m \text{~is odd}}}^{\bar{M}} \sum_{n=1}^N \mathcal{M}_n^{(m)} [C_s(m)+ C_b(m)].
 \end{equation}
 \end{itemize}
 Here $C_s(m)$ and $C_b(m)$ are respectively the computational cost of the system associated functional $\mathcal{U}^{(0)}$ and the bath influence functional $\mathcal{L}_b$ with a given $m-$point time sequence. Asymptotically, $\mathcal{U}^{(0)}$ is evaluated with the computational cost at $C_s(m) \sim O(m)$ by the definition \eqref{U0}, and $\mathcal{L}_b$ can be computed with the cost at $C_b(m) \sim O(2^m)$ upon applying the inclusion-exclusion principle \cite{yang2021inclusion}. For large $\bar{M}$, $C_s(m)$ in the three approaches is negligible, so that \eqref{eq:algo2} and \eqref{eq:algo3} will have similar values, which reduce the original cost \eqref{eq:algo1} by a factor $O(N)$ (see the complexity analysis in \cite[Section 2.5]{cai2022fast}). However, when $\bar{M}$ is small, $C_s(m)$ is comparable to $C_b(m)$, and then the ratio of \eqref{eq:algo3} to \eqref{eq:algo2} will be at $O(1/N)$ asymptotically using the formulas of $\mathcal{M}^{(m)}_n$ and $\bar{\mathcal{M}}^{(m)}_n$. In this sense, the algorithm in the current work can considerably improve the computational efficiency particularly in the weak system-bath coupling regime.
%  $C_b(\bar{M}) / [C_s(\bar{M}) + C_b(\bar{M})]$

 Moreover, Algorithm \ref{algo:dyson} requires a lower memory cost compared to the algorithm in \cite{cai2022fast}. In \Cref{Line7,Line8} of Algorithm \ref{algo:dyson}, where the major computations concentrate in, one may draw a single sample $\bs_{n+1}^{(k)}$ each time and compute the corresponding $\mathscr{S}^{(k)}_{ij}$, which is added to the summation in \Cref{Line8} and can then be discarded. In this process, the memory cost can stay low since we only need to store $\mathscr{D}_{ij}^{n+1}$ and $\mathscr{K}_{ij}^{n}$ for $i,j=\pm 1$, which are essentially eight two-by-two matrices in total. The algorithm in \cite{cai2022fast}, however, requires in total $N(N+1)/2$ such matrices to store the infinite series $\mathcal{K}_n$ for all time steps. We refer readers to \cite[Section 2.4]{cai2022fast} for the details.

\section{Bold-thin-bold diagrammatic Monte Carlo method}
\label{section_semi_inchworm}
It is well known that the methods based on direct summation of the Dyson series usually suffer from the numerical sign problem \cite{loh1990sign,muhlbacher2008real,werner2009diagrammatic} for long-time simulations.
To tame the sign problem, bold lines, which are usually the sum of many diagrams, have been introduced in a number of methods to reduce the number of diagrams \cite{gull2010bold,cohen2015taming}.
Among these methods, the inchworm Monte Carlo method has been successfully applied to models including quantum impurity problems and exact non-adiabatic dynamics \cite{chen2017inchworm1,chen2017inchworm2,cai2020inchworm,cai2020numerical}.
Below, we will first review the inchworm Monte Carlo method in the context of spin-boson model. 
Our review will be based on the reformulation of $\mathcal{K}(-t,t)$, which is a perspective slightly different from the literature \cite{chen2017inchworm1, cai2020inchworm}.
We will then modify this approach to derive a novel diagrammatic Monte Carlo method that also has the iterative structure introduced in \Cref{Sec_Dyson_reuse}.

\subsection{Review of inchworm Monte Carlo method}
\label{sec_review_imc}
In open quantum systems, the bold lines are defined to be the partial resummation of the Dyson series. In \cite{chen2017inchworm1}, the bold lines are introduced as ``full propagators'' $G(s_\ii,s_\ff)$, defined as
\begin{equation}
 \label{fullPropagator2D}
 \begin{split}
  G(s_\ii,s_\ff) = G^{(0)}_s(s_\ii,s_\ff) + 
     \sum_{\substack{m=2\\ m \text{~is even}}}^{+\infty}
    \ii^m
    \int_{s_\ii \leqslant \boldsymbol{s} \leqslant s_\ff}
    \dd \boldsymbol{s}
    \left[(-1)^{\#\{\boldsymbol{s}<0\}}
    \mathcal{U}^{(0)} (s_\ii,\boldsymbol{s},s_\ff)
    \mathcal{L}_b(\boldsymbol{s})\right], \qquad & \\
   s_\ii \in [-t,t]\backslash \{0\}, \quad s_\ff \in [s_\ii, t] \backslash \{0\}, &
 \end{split}
\end{equation}
where $G^{(0)}_s(s_\ii,s_\ff)$ is given in \eqref{Gs0_definition} and the system associated functional $ \mathcal{U}^{(0)} (s_\ii,\boldsymbol{s},s_\ff)$ is formulated as \eqref{U0} with $\pm t$ replaced by $s_\ff$ and $s_\ii$. The following properties of the full propagator will be found useful later in this paper (see their proofs in \cite{cai2022fast}):
\begin{proposition}\label{G property}
 \
  \begin{itemize}
      \item \textbf{Shift invariance:} For any $\Delta t \geqslant 0$, if $s_\ff + \Delta t < 0$ or $s_\ii  > 0$, we have  
\begin{equation}\label{G shift invariance}
 G(s_\ii+\Delta t,s_\ff+ \Delta t) = G(s_\ii,s_\ff)
\end{equation}
\item \textbf{Conjugate symmetry:} For any $-t \leqslant s_\ii \leqslant s_\ff < t$, we have 
\begin{equation}\label{G conjugate symmetry}
    G(-s_\ff,-s_\ii) = G(s_\ii,s_\ff)^\dagger
\end{equation}
\item  \textbf{Jump condition:} $G(\cdot,\cdot)$ is discontinuous on the line segments
$[-t,0] \times \{0\}$ and $\{0\} \times [-t,0]$. At the origin, we have 
\begin{gather}  
\lim_{s_\ff \rightarrow s_\ii^+}  \lim_{s_\ii \rightarrow 0^+}   G(s_\ii,s_\ff)  = \lim_{s_\ii \rightarrow s_\ff^-}  \lim_{s_\ff \rightarrow 0^-}   G(s_\ii,s_\ff)  = I, \label{jump condition} \\
\lim_{s_\ii \rightarrow 0^-}  \lim_{s_\ff \rightarrow 0^+}   G(s_\ii,s_\ff)  = \hat{O}_s. \label{jump condition dyson}
\end{gather}
    \end{itemize}
 \end{proposition}
We remark that the initial condition of the integro-differential equation \eqref{integro_differential_eq} is a result of jump condition \eqref{jump condition dyson} as $\lim_{t\rightarrow 0^+}G(-t,t) = \hat{O}_s$. Later we will see that other solvers for the full propagator $G(s_\ii,s_\ff)$ may use \eqref{jump condition} as the initial conditions.   

The central idea of inchworm method \cite{cai2020inchworm} lies in the resummation of the series $\mathcal{K}(-t,t)$ in the governing equation \eqref{integro_differential_eq} using the previous calculations of the full propagators.
Diagrammatically, the full propagator $G(s_\ii,s_\ff)$ is represented by a bold line segment connecting the initial time point $s_\ii$ with the final time point $s_\ff$. This is consistent with \eqref{fullPropagator} when $s_\ii =-t$ and $s_\ff = t$. The inchworm method then replaces all the thin line segments $G^{(0)}_s(s_{\ii},s_{\ff})$ in the diagrams \eqref{diagram dyson} by bold line segments, and then removes some diagrams from the summation to maintain the equality. For example, the first diagram on the right-hand side of \eqref{diagram dyson} is replaced by 
\begin{equation}
 \label{first term inchworm}
\begin{tikzpicture}[anchor=base, baseline]
\draw [thick] (-1.4,0)--(1.4,0);
\draw[-] (-0.7,0.1) to[bend left=75] (1.4,0.1);
\fill [black] (-1.4,-0.1) rectangle (1.4,0.1);
\node [below] at (-1.4,-0.18) {$-t$};
\node [below] at (1.4,-0.18) {$t$};
\draw [thick] (-1.4,-0.12)--(-1.4,0.12);
\draw [thick] (1.4,-0.12)--(1.4,0.12);
\draw [thick,white] (-0.7,0.1) -- (-0.7,-0.1);
\node [below] at (-0.7,-0.1) {$s_1$};
\end{tikzpicture} = 
\int_{-t}^t \dd s_1 \  \sgn(s_1)  \ii^2  G(s_1,t) W_s  G(-t,s_1) B(s_1,t). 
\end{equation}
By \eqref{fullPropagator2D}, the two bold line segments above can be expressed as 
\begin{equation}
   \label{G expansion one arc}
 \begin{split}
     G(-t,s_1) = & \  G^{(0)}_s(-t,s_1) +  \int_{-t}^{s_1}  \int_{-t}^{s_3}  \dd s_2 \dd s_3\ (-1)^{\#\{s_2,s_3<0\}}  \ii^2  \times \\
     & \times G^{(0)}_s(s_3,s_1)  W_s  G^{(0)}_s(s_2,s_3) W_s G^{(0)}_s(-t,s_2)  B(s_2,s_3) + \cdots,  \\
          G(s_1,t) = & \  G^{(0)}_s(s_1,t) +  \int_{s_1}^t \int_{s_1}^{s'_3} \dd s'_2 \dd s'_3 \ (-1)^{\#\{s'_2,s'_3<0\}}  \ii^2   \times \\
          & \times G^{(0)}_s(s'_3,t) W_s G^{(0)}_s(s'_2,s'_3)W_s  G^{(0)}_s(s_1,s'_2)  B(s'_2,s'_3) + \cdots. 
 \end{split}
\end{equation}
Inserting the above formulas into \eqref{first term inchworm} yields that 
\begin{equation}
 \label{first term inchworm diagram}
 \begin{split}
 &    \begin{tikzpicture}[anchor=base, baseline]
\draw [thick] (-1.4,0)--(1.4,0);
\draw[-] (-0.7,0.1) to[bend left=75] (1.4,0.1);
\fill [black] (-1.4,-0.1) rectangle (1.4,0.1);
\node [below] at (-1.4,-0.18) {$-t$};
\node [below] at (1.4,-0.18) {$t$};
\draw [thick] (-1.4,-0.12)--(-1.4,0.12);
\draw [thick] (1.4,-0.12)--(1.4,0.12);
\draw [thick,white] (-0.7,0.1) -- (-0.7,-0.1);
\node [below] at (-0.7,-0.1) {$s_1$};
\end{tikzpicture} =  \ \begin{tikzpicture}[anchor=base, baseline] 
 \draw [thick] (-1.4,0) -- (1.4,0);
 \draw [thick] (-1.4,-0.1)--(-1.4,0.1); \draw [thick] (1.4,-0.1)--(1.4,0.1);
 \node [below] at (-1.4,-0.1) {$-t$}; 
  \node [below] at (1.4,-0.1) {$t$}; 
  \draw[-] (-1,0) to[bend left=60] (1.4,0);
  \draw plot[only marks,mark =*, mark options={color=black, scale=0.5}]coordinates {(-1,0)};
  \node [below] at (-1,0) {$s_1$};
 \end{tikzpicture}
+ \begin{tikzpicture}[anchor=base, baseline] 
 \draw [thick] (-1.4,0) -- (1.4,0);
 \draw [thick] (-1.4,-0.1)--(-1.4,0.1); \draw [thick] (1.4,-0.1)--(1.4,0.1);
 \node [below] at (-1.4,-0.1) {$-t$}; 
  \node [below] at (1.4,-0.1) {$t$}; 
  \draw[-] (-1,0) to[bend left=75] (-0.2,0);
    \draw[-] (0.6,0) to[bend left=75] (1.4,0);
  \draw plot[only marks,mark =*, mark options={color=black, scale=0.5}]coordinates {(-1,0) (-0.2,0) (0.6,0)};
  \node [below] at (-1,0) {$s_1$};  \node [below] at (-0.2,0) {$s_2$};   \node [below] at (0.6,0) {$s_3$};
\end{tikzpicture} \\
& \hspace{5cm} + 
\begin{tikzpicture}[anchor=base, baseline] 
 \draw [thick] (-1.4,0) -- (1.4,0);
 \draw [thick] (-1.4,-0.1)--(-1.4,0.1); \draw [thick] (1.4,-0.1)--(1.4,0.1);
 \node [below] at (-1.4,-0.1) {$-t$}; 
  \node [below] at (1.4,-0.1) {$t$}; 
  \draw[-] (-1,0) to[bend left=75] (1.4,0);
    \draw[-] (-0.2,0) to[bend left=75] (0.6,0);
  \draw plot[only marks,mark =*, mark options={color=black, scale=0.5}]coordinates {(-1,0) (-0.2,0) (0.6,0)};
  \node [below] at (-1,0) {$s_1$};  \node [below] at (-0.2,0) {$s_2$};   \node [below] at (0.6,0) {$s_3$};
\end{tikzpicture} 
+ 
\begin{tikzpicture}[anchor=base, baseline] 
 \draw [thick] (-1.4,0) -- (1.4,0);
 \draw [thick] (-1.4,-0.1)--(-1.4,0.1); \draw [thick] (1.4,-0.1)--(1.4,0.1);
 \node [below] at (-1.4,-0.1) {$-t$}; 
  \node [below] at (1.4,-0.1) {$t$}; 
  \draw[-] (-0.1,0) to[bend left=75] (1.4,0);
   \draw[-] (-1,0) to[bend left=75] (-0.5,0);
    \draw[-] (0.4,0) to[bend left=75] (0.9,0);
  \draw plot[only marks,mark =*, mark options={color=black, scale=0.5}]coordinates {(-1,0)(-0.5,0)(-0.1,0)(0.4,0)(0.9,0)};
    \node [below] at (-1,0) {$s_1$};
        \node [below] at (-0.5,0) {$s_2$};
  \node [below] at (-0.1,0) {$s_3$};
      \node [below] at (0.4,0) {$s_4$};    \node [below] at (0.9,0) {$s_5$};
 \end{tikzpicture} + \cdots
  \end{split}
\end{equation}
which includes all the thin diagrams where the arc connecting to $t$ does not intersect with any other arcs. These diagrams form a subset of the diagrams in \eqref{diagram dyson}, and thus \eqref{first term inchworm} can be considered as a partial sum of $\mathcal{K}(-t,t)$. For example, the first, second and fourth diagram in \eqref{diagram dyson} are included in \eqref{first term inchworm diagram} while the third diagram is not since the two arcs have an intersection.
Therefore, on the right-hand side of \eqref{diagram dyson}, after replacing the thin line segments in the first diagram with bold line segments, the second and the fourth diagrams, as well as infinite other diagrams in \eqref{first term inchworm diagram}, should be removed.
Then, for the third diagram, we again replace all the thin line segments with bold line segments:
\begin{multline*}
\begin{tikzpicture}[anchor=base, baseline]
\draw [thick] (-1.4,0)--(1.4,0);
\draw[-] (-0.7,0.1) to[bend left=75] (0.7,0.1);
\draw[-] (0,0.1) to[bend left=75] (1.4,0.1);
\fill [black] (-1.4,-0.1) rectangle (1.4,0.1);
\node [below] at (-1.4,-0.18) {$-t$};
\node [below] at (1.4,-0.18) {$t$};
\draw [thick] (-1.4,-0.12)--(-1.4,0.12);
\draw [thick] (1.4,-0.12)--(1.4,0.12);
\draw [thick,white] (-0.7,0.1) -- (-0.7,-0.1);
\draw [thick,white] (0.7,0.1) -- (0.7,-0.1);
\draw [thick,white] (0,0.1) -- (0,-0.1);
\node [below] at (-0.7,-0.1) {$s_1$};\node [below] at (0,-0.1) {$s_2$};
\node [below] at (0.7,-0.1) {$s_3$};
\end{tikzpicture} = 
\int_{-t \leqslant s_1 \leqslant s_2\leqslant s_3 \leqslant t} \dd s_1   \dd s_2  \dd s_3 \\
(-1)^{\#\{\boldsymbol{s}<0\}} \ii^4 G(s_3,t) W_s G(s_2,s_3) W_s G(s_1,s_2) W_s G(-t,s_1)  B(s_1,s_3) B(s_2,t), 
\end{multline*}
which can also be expanded into infinitely many diagrams, and we will remove these diagrams from \eqref{diagram dyson} to maintain the equality.
We can continue such a process of replacements and removals, and eventually the series $\mathcal{K}(-t,t)$ will be expressed as the sum of bold diagrams only:
\begin{equation}\label{diagram inchworm}
\begin{split}
    \mathcal{K}(-t,t) = \mathcal{K}^c(-t,t)\coloneqq
    \begin{tikzpicture}[anchor=base, baseline]
\draw [thick] (-1.4,0)--(1.4,0);
\draw[-] (-0.7,0.1) to[bend left=75] (1.4,0.1);
\fill [black] (-1.4,-0.1) rectangle (1.4,0.1);
\node [below] at (-1.4,-0.18) {$-t$};
\node [below] at (1.4,-0.18) {$t$};
\draw [thick] (-1.4,-0.12)--(-1.4,0.12);
\draw [thick] (1.4,-0.12)--(1.4,0.12);
\draw [thick,white] (-0.7,0.1) -- (-0.7,-0.1);
\node [below] at (-0.7,-0.1) {$s_1$};
\end{tikzpicture}
+ 
\begin{tikzpicture}[anchor=base, baseline]
\draw [thick] (-1.4,0)--(1.4,0);
\draw[-] (-0.7,0.1) to[bend left=75] (0.7,0.1);
\draw[-] (0,0.1) to[bend left=75] (1.4,0.1);
\fill [black] (-1.4,-0.1) rectangle (1.4,0.1);
\node [below] at (-1.4,-0.18) {$-t$};
\node [below] at (1.4,-0.18) {$t$};
\draw [thick] (-1.4,-0.12)--(-1.4,0.12);
\draw [thick] (1.4,-0.12)--(1.4,0.12);
\draw [thick,white] (-0.7,0.1) -- (-0.7,-0.1);
\draw [thick,white] (0.7,0.1) -- (0.7,-0.1);
\draw [thick,white] (0,0.1) -- (0,-0.1);
\node [below] at (-0.7,-0.1) {$s_1$};\node [below] at (0,-0.1) {$s_2$};
\node [below] at (0.7,-0.1) {$s_3$};
\end{tikzpicture} 
+ 
\begin{tikzpicture}[anchor=base, baseline]
\draw [thick] (-1.5,0)--(1.5,0);
\fill [black] (-1.5,-0.1) rectangle (1.5,0.1);
\draw[-] (-1,0.1) to[bend left=75] (0.5,0.1);
\draw[-] (-0.5,0.1) to[bend left=75] (1,0.1);
\draw[-] (0,0.1) to[bend left=75] (1.5,0.1);
\node [below] at (-1.5,-0.18) {$-t$};
\node [below] at (1.5,-0.18) {$t$};
\draw [thick] (-1.5,-0.12)--(-1.5,0.12);
\draw [thick] (1.5,-0.12)--(1.5,0.12);
\draw [thick,white] (-1,0.1) -- (-1,-0.1);
\draw [thick,white] (-0.5,0.1) -- (-0.5,-0.1);
\draw [thick,white] (-0,0.1) -- (-0,-0.1);
\draw [thick,white] (0.5,0.1) -- (0.5,-0.1);
\draw [thick,white] (1,0.1) -- (1,-0.1);
\node [below] at (-1,-0.1) {$s_1$};\node [below] at (-0.5,-0.1) {$s_2$};
\node [below] at (0,-0.1) {$s_3$};\node [below] at (0.5,-0.1) {$s_4$};
\node [below] at (1,-0.1) {$s_5$};
\end{tikzpicture} \\
+ 
\begin{tikzpicture}[anchor=base, baseline]
\draw [thick] (-1.5,0)--(1.5,0);
\fill [black] (-1.5,-0.1) rectangle (1.5,0.1);
\draw[-] (-1,0.1) to[bend left=75] (0.5,0.1);
\draw[-] (-0.5,0.1) to[bend left=75] (1.5,0.1);
\draw[-] (0,0.1) to[bend left=75] (1,0.1);
\node [below] at (-1.5,-0.18) {$-t$};
\node [below] at (1.5,-0.18) {$t$};
\draw [thick] (-1.5,-0.12)--(-1.5,0.12);
\draw [thick] (1.5,-0.12)--(1.5,0.12);
\draw [thick,white] (-1,0.1) -- (-1,-0.1);
\draw [thick,white] (-0.5,0.1) -- (-0.5,-0.1);
\draw [thick,white] (-0,0.1) -- (-0,-0.1);
\draw [thick,white] (0.5,0.1) -- (0.5,-0.1);
\draw [thick,white] (1,0.1) -- (1,-0.1);
\node [below] at (-1,-0.1) {$s_1$};\node [below] at (-0.5,-0.1) {$s_2$};
\node [below] at (0,-0.1) {$s_3$};\node [below] at (0.5,-0.1) {$s_4$};
\node [below] at (1,-0.1) {$s_5$};
\end{tikzpicture} 
+ 
\begin{tikzpicture}[anchor=base, baseline]
\draw [thick] (-1.5,0)--(1.5,0);
\fill [black] (-1.5,-0.1) rectangle (1.5,0.1);
\draw[-] (-1,0.1) to[bend left=75] (0,0.1);
\draw[-] (-0.5,0.1) to[bend left=75] (1,0.1);
\draw[-] (0.5,0.1) to[bend left=75] (1.5,0.1);
\node [below] at (-1.5,-0.18) {$-t$};
\node [below] at (1.5,-0.18) {$t$};
\draw [thick] (-1.5,-0.12)--(-1.5,0.12);
\draw [thick] (1.5,-0.12)--(1.5,0.12);
\draw [thick,white] (-1,0.1) -- (-1,-0.1);
\draw [thick,white] (-0.5,0.1) -- (-0.5,-0.1);
\draw [thick,white] (-0,0.1) -- (-0,-0.1);
\draw [thick,white] (0.5,0.1) -- (0.5,-0.1);
\draw [thick,white] (1,0.1) -- (1,-0.1);
\node [below] at (-1,-0.1) {$s_1$};\node [below] at (-0.5,-0.1) {$s_2$};
\node [below] at (0,-0.1) {$s_3$};\node [below] at (0.5,-0.1) {$s_4$};
\node [below] at (1,-0.1) {$s_5$};
\end{tikzpicture} 
+
\begin{tikzpicture}[anchor=base, baseline]
\draw [thick] (-1.5,0)--(1.5,0);
\fill [black] (-1.5,-0.1) rectangle (1.5,0.1);
\draw[-] (-1,0.1) to[bend left=75] (1,0.1);
\draw[-] (-0.5,0.1) to[bend left=75] (0.5,0.1);
\draw[-] (0,0.1) to[bend left=75] (1.5,0.1);
\node [below] at (-1.5,-0.18) {$-t$};
\node [below] at (1.5,-0.18) {$t$};
\draw [thick] (-1.5,-0.12)--(-1.5,0.12);
\draw [thick] (1.5,-0.12)--(1.5,0.12);
\draw [thick,white] (-1,0.1) -- (-1,-0.1);
\draw [thick,white] (-0.5,0.1) -- (-0.5,-0.1);
\draw [thick,white] (-0,0.1) -- (-0,-0.1);
\draw [thick,white] (0.5,0.1) -- (0.5,-0.1);
\draw [thick,white] (1,0.1) -- (1,-0.1);
\node [below] at (-1,-0.1) {$s_1$};\node [below] at (-0.5,-0.1) {$s_2$};
\node [below] at (0,-0.1) {$s_3$};\node [below] at (0.5,-0.1) {$s_4$};
\node [below] at (1,-0.1) {$s_5$};
\end{tikzpicture} + \cdots.
\end{split}
\end{equation}
In these diagrams, all the arcs are ``linked'', meaning that any two time points can be connected with each other using the arcs as ``bridges''. % Rigorously, the linkedness of a given set of pairs is described based on the following definitions:  
% \begin{definition}[Linked pairs]
% Two pairs of real numbers $(s_1, s_2)$ and $(\tau_1, \tau_2)$ satisfying $s_1
% \leqslant s_2$ and $\tau_1 \leqslant \tau_2$ are \emph{linked} if either of the
% following two statements holds:
% \begin{enumerate}
% \item $s_1 \leqslant \tau_1 \leqslant s_2$ and $\tau_1 \leqslant s_2 \leqslant \tau_2$.
% \item $\tau_1 \leqslant s_1 \leqslant \tau_2$ and $s_1 \leqslant \tau_2 \leqslant s_2$.
% \end{enumerate}
% \end{definition}
% \begin{definition}[Linked sets of pairs]
% Given two sets of pairs $\mathfrak{q}_1$ and $\mathfrak{q}_2$, we say $\mathfrak{q}_1$
% and $\mathfrak{q}_2$ are linked  if there exists $(s_1, s_2) \in \mathfrak{q}_1$ and $(\tau_1, \tau_2) \in \mathfrak{q}_2$ such that $(s_1, s_2)$ and $(\tau_1, \tau_2)$
% are linked. We say a given set of pairs $\mathfrak{q}$ is linked if it cannot be decomposed into the union of two sets of pairs that are not linked.
% \end{definition}
We remark that such resummation includes all diagrams shown in \eqref{diagram dyson} and no diagram is double-counted. One may refer to \cite[Section 3.3]{cai2020inchworm} for the rigorous proof. Mathematically, we may use the subset $\mathcal{Q}^c(\bs) \subset \mathcal{Q}(\bs)$ to denote the collection of all linked pairings. The superscript $c$ here refers to ``connected". For example, 
\begin{align*}
     &\mathcal{Q}^c(s_1,s_2) = \big\{\{(s_1,s_2)\}\big\},\\
     &  \mathcal{Q}^c(s_1,s_2,s_3,s_4) = \big\{\{(s_1,s_3),(s_2,s_4)\}\big\} \\
     &  \mathcal{Q}^c(s_1,s_2,s_3,s_4,s_5,s_6) = \big\{\{(s_1,s_4),(s_2,s_5),(s_3,s_6)\},\{(s_1,s_4),(s_2,s_6),(s_3,s_5)\},\\
     & \hspace{7cm}\{(s_1,s_3),(s_2,s_5),(s_4,s_6)\},\{(s_1,s_5),(s_2,s_4),(s_3,s_6)\}\big\}.
\end{align*}
and thus the resummation \eqref{diagram inchworm} is formulated as 
\begin{equation}
 \label{K inchworm}
 \mathcal{K}^c(-t,t) = \sum_{\substack{m=1 \\ m \text{~is odd}}}^{+\infty}  \ii^{m+1} \int_{-t\leqslant \bs \leqslant t} \dd \bs  (-1)^{\#\{\bs < 0\}} \mathcal{U}(-t, \bs , t)  \Ls_b^c(\bs,t)  
\end{equation}
where for odd $m$,
\begin{align*}
   & \mathcal{U}(-t, s_1,\cdots,s_m , t) =   \ G(s_m,t) W_s
    G(s_{m-1},s_m) W_s
    \cdots
    W_s G(s_1,s_2) 
    W_s G(-t,s_1), \\
    & \Ls_b^c(s_1,\cdots,s_{m+1})  =  \ \sum_{\mathfrak{q}\in\mathcal{Q}^c(s_1,\cdots,s_{m+1})}
    \prod_{(s_j,s_k)\in\mathfrak{q}}
    B(s_j,s_k).
\end{align*}
Note that here $m$ is odd, and thus $\mathcal{Q}^c$ has $m+1$ arguments. Similarly, the infinite series $\mathcal{K}^c(-t,t)$ can be approximated by truncating $m$ to some large odd $\bar{M}$ and evaluating the high-dimensional integrals via Monte Carlo method. Such a numerical solver is known as the inchworm Monte Carlo method. Compared with the Dyson series $\mathcal{K}^{(0)}(-t,t)$ defined in \eqref{K dyson}, working with $\mathcal{K}^c(-t,t)$ is advantageous as the its convergence with respect to $m$ is faster since each diagram in \eqref{diagram inchworm} includes infinite diagrams in \eqref{diagram dyson}. In addition, given any $m$, the inchworm Monte Carlo method only considers linked pairings and thus contains fewer diagrams than \eqref{diagram dyson}, making the direct evaluation of the bath influence functional $\Ls_b^c(s_1,\cdots,s_m,t)$ cheaper than $\Ls_b(s_1,\cdots,s_m,t)$ in the Dyson series. 

In order for efficient implementations, we examine the possibility of applying the idea of iterative calculations proposed in Section \ref{sec:dyson} to the inchworm Monte Carlo method, which requires several invariant properties for the functionals in the reformulated $\mathcal{K}^c(-t,t)$. It is not hard to observe that the bath influence functional $\Ls_b^c(s_1,\cdots,s_m,t)$ satisfies the invariance similar to \eqref{shift_properties_Lb}:
\begin{equation*}
    \mathcal{L}_b^c(\mathcal{I}_{\dt}(\boldsymbol{s}),t+\dt)
    = \mathcal{L}_b^c(\boldsymbol{s},t)
\end{equation*}
according to \eqref{properties_B}, and the shift invariance \eqref{G shift invariance} produces the similar result as \eqref{properties_Gs0}. However, when $s_\ii < 0 \leqslant s_\ff$, it is not clear whether we can write $G(s_\ii,s_\ff)$ as a linear combination of some simple basis as in \eqref{basis}. As a result, a decomposition similar to \eqref{shift_property}, which plays a key role in our iterative algorithm, is unavailable for $G(s_\ii,s_\ff)$. This indicates that special treatment for the propagators is needed when crossing the time 0 to make our fast algorithm applicable.

\subsection{The bold-thin-bold diagrammatic Monte Carlo method}
\label{Modified_IMCM}
In this subsection, 
we will modify the current inchworm Monte Carlo method 
by considering another resummation of $\mathcal{K}(-t,t)$ 
such that the idea of our iterative scheme for Dyson series can again be applied.
Instead of replacing all bare propagators $G^{(0)}_s(s_\ii,s_\ff)$ in Dyson series by full propagators, we replace them by 
    \begin{equation}
     \label{G semi}
        \widehat{G}(s_\ii,s_\ff)
        = \begin{cases}
            G(s_\ii,s_\ff),
            &\text{if }s_\ii \leqslant s_\ff < 0
            \text{ or }0\leqslant s_\ii \leqslant s_\ff \\
             G_s^{(0)}(s_\ii,s_\ff),
            &\text{ if } s_\ii < 0 \leqslant s_\ff
        \end{cases}.
    \end{equation}
Such expression allows us to define its linear basis analogous to \eqref{basis}:
    \begin{equation}
     \label{basis_2}
        \widehat{\mathscr{G}}_{ij} (s_\ii,s_\ff)
        = \begin{cases}
            G(s_\ii,s_\ff),
            &\text{ if }s_\ii \leqslant s_\ff < 0
            \text{ or }0\leqslant s_\ii \leqslant s_\ff \\
            \e^{\ii s_\ff H_s} \dyad{i}{j} \e^{\ii s_\ii H_s},
            &\text{ if } s_\ii < 0 \leqslant s_\ff
        \end{cases},
\end{equation}
which immediately gives us the similar result as \eqref{shift_property}:  
\begin{equation} \label{recurrence}
   \widehat{\mathscr{G}}_{ij}(s_\ii-\dt,s_\ff+\dt)
       = \sum_{k,l} b_{kl}^{ij} \widehat{\mathscr{G}}_{kl}(s_\ii,s_\ff) \text{~for~} s_\ii < 0 \leqslant s_\ff.
\end{equation}
With the above desired invariance satisfied, one can now repeat the analysis for Dyson series to obtain a formula similar to Theorem \ref{thm:Sscr_shift}, based on which the reuse algorithm can again be applied. We postpone the details in Section \ref{sec:semi_inchworm_numerical_method}. For now, let us focus on how the series $\mathcal{K}(-t,t)$ will be reformulated upon replacing all $G_s^{(0)}(s_\ii,s_\ff)$ by $\widehat{G}(s_\ii,s_\ff)$. In particular, we are interested in the new bath influence functional after this replacement.  

In the diagrammatic representation, all the thin line segments are replaced by the corresponding bold line segments only when they do not include time 0. Below, we will derive the integro-differential equation based on these propagators following the ``replace and remove'' procedure introduced in the previous subsection. Again, we start from the first diagram in $\mathcal{K}^{(0)}(-t,t)$ defined by \eqref{diagram dyson}. But now we split this diagram into two subdiagrams:
\begin{equation}
  \label{first thin diagram split}
 \begin{split}
    \begin{tikzpicture}[anchor=base, baseline] 
 \draw [thick] (-1.4,0) -- (1.4,0);
 \draw [thick] (-1.4,-0.1)--(-1.4,0.1); \draw [thick] (1.4,-0.1)--(1.4,0.1);
 \node [below] at (-1.4,-0.1) {$-t$}; 
  \node [below] at (1.4,-0.1) {$t$}; 
  \draw[-] (-1,0) to[bend left=60] (1.4,0);
  \draw plot[only marks,mark =*, mark options={color=black, scale=0.5}]coordinates {(-1,0)};
  \node [below] at (-1,0) {$s_1$};
 \end{tikzpicture}
 = & \ \int^0_{-t} \dd s_1 \  (-1)  \ii^2    G^{(0)}_s(s_1,t) W_s G^{(0)}_s(-t,s_1) B(s_1,t) \\
 & \hspace{3cm} + \int^t_0 \dd s_1 \   \ii^2  G^{(0)}_s(s_1,t) W_s  G^{(0)}_s(-t,s_1) B(s_1,t) \\ 
 =:  & \ 
 \begin{tikzpicture}[anchor=base, baseline]
\draw [thick] (-1.4,0)--(1.4,0);
\draw[-] (-0.7,0) to[bend left=75] (1.4,0);
\node [below] at (-1.4,-0.18) {$-t$};
\node [below] at (-0.7,-0.18) {$s_1$};
\node [below] at (0,-0.18) {0};
\draw [thick] (0,-0.1)--(0,0.1);
\node [below] at (1.4,-0.18) {$t$};
\draw [thick] (-1.4,-0.12)--(-1.4,0.12);
\draw plot[only marks,mark =*, mark options={color=black, scale=0.5}]coordinates{(-0.7,0)};
\draw [thick] (1.4,-0.12)--(1.4,0.12);
\end{tikzpicture}
+
 \begin{tikzpicture}[anchor=base, baseline]
\draw [thick] (-1.4,0)--(1.4,0);
\draw[-] (0.7,0) to[bend left=75] (1.4,0);
\node [below] at (-1.4,-0.18) {$-t$};
\node [below] at (0,-0.18) {0};
\node [below] at (0.7,-0.18) {$s_1$};
\node [below] at (1.4,-0.18) {$t$};
\draw [thick] (-1.4,-0.12)--(-1.4,0.12);
\draw [thick] (0,-0.1)--(0,0.1);
\draw plot[only marks,mark =*, mark options={color=black, scale=0.5}]coordinates{(0.7,0)};
\draw [thick] (1.4,-0.12)--(1.4,0.12);
\end{tikzpicture} .
\end{split}
\end{equation}
Here the time 0 is marked explicitly in the diagrams, and the location of $s_1$ with respect to time 0 indicates the domain of integration. When upgrading thin lines to bold lines, we will not change the line segments containing the zero point. This can be automatically achieved by replacing $G^{(0)}_s(s_\ii,s_\ff)$ with $\widehat{G}(s_\ii,s_\ff)$:
\begin{equation}
   \label{first bold diagram split}
   \begin{split}
   &  \int^0_{-t} \dd s_1 \  (-1)  \ii^2 \widehat{G}(s_1,t)   W_s  \widehat{G}(-t,s_1) B(s_1,t) + \int^t_0 \dd s_1 \    \ii^2  \widehat{G}(s_1,t) W_s \widehat{G}(-t,s_1)B(s_1,t) \\
= &  \int^0_{-t} \dd s_1 \  (-1)  \ii^2 G^{(0)}_s(s_1,t)  W_s  G(-t,s_1)B(s_1,t) + \int^t_0 \dd s_1 \    \ii^2 G(s_1,t)  W_s G^{(0)}_s(-t,s_1) B(s_1,t) \\
=: &  \begin{tikzpicture}[anchor=base, baseline]
\draw [thick] (-1.4,0)--(1.4,0);
\draw[-] (-0.7,0.1) to[bend left=75] (1.4,0.1);
\fill [black] (-1.4,-0.1) rectangle (-0.7,0.1);
\node [below] at (-1.4,-0.18) {$-t$};
\node [below] at (-0.7,-0.18) {$s_1$};
\node [below] at (0,-0.18) {0};
\node [below] at (1.4,-0.18) {$t$};
\draw [thick] (-1.4,-0.12)--(-1.4,0.12);
\draw [thick] (-0.7,-0.12)--(-0.7,0.12);
\draw [thick] (0,-0.1)--(0,0.1);\draw [thick] (1.4,-0.12)--(1.4,0.12);
\end{tikzpicture}
+ 
\begin{tikzpicture}[anchor=base, baseline]
\draw [thick] (-1.4,0)--(1.4,0);
\draw[-] (0.7,0.1) to[bend left=75] (1.4,0.1);
\fill [black] (0.7,-0.1) rectangle (1.4,0.1);
\node [below] at (-1.4,-0.18) {$-t$};
\node [below] at (0,-0.18) {0};
\node [below] at (0.7,-0.18) {$s_1$};
\node [below] at (1.4,-0.18) {$t$};
\draw [thick] (-1.4,-0.12)--(-1.4,0.12);
\draw [thick] (0,-0.1)--(0,0.1);\draw [thick] (0.7,-0.12)--(0.7,0.12);
\draw [thick] (1.4,-0.12)--(1.4,0.12);
\end{tikzpicture}.
 \end{split}
\end{equation}
Similar to \eqref{first term inchworm diagram}, one may expand the bold line segments as Dyson series by inserting the expressions \eqref{G expansion one arc} to find the thin diagrams included in \eqref{first bold diagram split}. Clearly, the two diagrams in \eqref{first thin diagram split} are included. In addition, the following diagrams with two arcs
\begin{displaymath}
\begin{tikzpicture}[anchor=base, baseline]
\draw [thick] (-1.4,0)--(1.4,0);
\draw (-1.05,0) arc[radius = 0.175,start angle = 180,end angle = 0];
\draw[-] (-0.35,0) to[bend left=75] (1.4,0);
\node [below left] at (-1.4,-0.18) {$-t$};
\node [below] at (-1.05,-0.18) {$s_1$};
\node [below] at (-0.7,-0.18) {$s_2$};
\node [below] at (-0.35,-0.18) {$s_3$};
\node [below] at (0,-0.18) {0};
\node [below] at (1.4,-0.18) {$t$};
\draw [thick] (-1.4,-0.12)--(-1.4,0.12);
\draw plot[only marks,mark =*, mark options={color=black, scale=0.5}]coordinates{(-1.05,0)};
\draw plot[only marks,mark =*, mark options={color=black, scale=0.5}]coordinates{(-0.7,0)};
\draw plot[only marks,mark =*, mark options={color=black, scale=0.5}]coordinates{(-0.35,0)};
\draw [thick] (0,-0.1)--(0,0.1);
\draw [thick] (1.4,-0.12)--(1.4,0.12);
\end{tikzpicture}
\in 
\begin{tikzpicture}[anchor=base, baseline]
\draw [thick] (-1.4,0)--(1.4,0);
\draw[-] (-0.7,0.1) to[bend left=75] (1.4,0.1);
\fill [black] (-1.4,-0.1) rectangle (-0.7,0.1);
\node [below] at (-1.4,-0.18) {$-t$};
\node [below] at (-0.7,-0.18) {$s_1$};
\node [below] at (0,-0.18) {0};
\node [below] at (1.4,-0.18) {$t$};
\draw [thick] (-1.4,-0.12)--(-1.4,0.12);
\draw [thick] (-0.7,-0.12)--(-0.7,0.12);
\draw [thick] (0,-0.1)--(0,0.1);
\draw [thick] (1.4,-0.12)--(1.4,0.12);
\end{tikzpicture}\quad ,\quad 
\begin{tikzpicture}[anchor=base, baseline]
\draw [thick] (-1.4,0)--(1.4,0);
\draw[-] (0.35,0) to[bend left=75] (1.4,0);
\draw (0.7,0) arc[radius = 0.175,start angle = 180,end angle = 0];
\node [below] at (-1.4,-0.18) {$-t$};
\node [below] at (0,-0.18) {0};
\node [below] at (0.35,-0.18) {$s_1$};
\node [below] at (0.7,-0.18) {$s_2$};
\node [below] at (1.05,-0.18) {$s_3$};
\node [below] at (1.4,-0.18) {$t$};
\draw [thick] (-1.4,-0.12)--(-1.4,0.12);
\draw [thick] (0,-0.1)--(0,0.1);
\draw plot[only marks,mark =*, mark options={color=black, scale=0.5}]coordinates{(0.35,0)};
\draw plot[only marks,mark =*, mark options={color=black, scale=0.5}]coordinates{(0.7,0)};
\draw plot[only marks,mark =*, mark options={color=black, scale=0.5}]coordinates{(1.05,0)};
\draw [thick] (1.4,-0.12)--(1.4,0.12);
\end{tikzpicture} \in 
\begin{tikzpicture}[anchor=base, baseline]
\draw [thick] (-1.4,0)--(1.4,0);
\draw[-] (0.7,0.1) to[bend left=75] (1.4,0.1);
\fill [black] (0.7,-0.1) rectangle (1.4,0.1);
\node [below] at (-1.4,-0.18) {$-t$};
\node [below] at (0,-0.18) {0};
\node [below] at (0.7,-0.18) {$s_1$};
\node [below] at (1.4,-0.18) {$t$};
\draw [thick] (-1.4,-0.12)--(-1.4,0.12);
\draw [thick] (0,-0.1)--(0,0.1);
\draw [thick] (0.7,-0.12)--(0.7,0.12);
\draw [thick] (1.4,-0.12)--(1.4,0.12);
\end{tikzpicture}
\end{displaymath}
as well as infinitely many other diagrams with more arcs also appear in \eqref{first bold diagram split}. These diagrams should be removed from the sum after replacing \eqref{first thin diagram split} with \eqref{first bold diagram split}, and these newly introduced diagrams in \eqref{first bold diagram split} which combine thin and bold line segments can again be considered as partial sums of $\mathcal{K}(-t,t)$. To proceed, we split the diagrams with two arcs in \eqref{diagram dyson} by adding the zero point: 
\begin{equation}
   \label{second thin diagram split}
  \begin{split}
 &  \begin{tikzpicture}[anchor=base, baseline] 
 \draw [thick] (-1.2,0) -- (1.2,0);
 \draw [thick] (-1.2,-0.1)--(-1.2,0.1); \draw [thick] (1.2,-0.1)--(1.2,0.1);
 \node [below] at (-1.2,-0.1) {$-t$}; 
  \node [below] at (1.2,-0.1) {$t$}; 
  \draw[-] (-0.8,0) to[bend left=75] (0,0);
    \draw[-] (0.4,0) to[bend left=75] (1.2,0);
  \draw plot[only marks,mark =*, mark options={color=black, scale=0.5}]coordinates {(-0.8,0) (0,0) (0.4,0)};
  \node [below] at (-0.8,0) {$s_1$};  \node [below] at (0,0) {$s_2$};   \node [below] at (0.4,0) {$s_3$};
\end{tikzpicture}
 = 
\boxed{ \begin{tikzpicture}[anchor=base, baseline]
\draw [thick] (-1.2,0)--(1.2,0);
\draw (-0.9,0) arc[radius = 0.15,start angle = 180,end angle = 0];
\draw[-] (-0.3,0) to[bend left=75] (1.2,0);
\node [below left] at (-1.2,-0.18) {$-t$};
\node [below] at (-0.9,-0.18) {$s_1$};
\node [below] at (-0.6,-0.18) {$s_2$};
\node [below] at (-0.3,-0.18) {$s_3$};
\node [below] at (0,-0.18) {0};
\node [below] at (1.2,-0.18) {$t$};
\draw [thick] (-1.2,-0.12)--(-1.2,0.12);
\draw plot[only marks,mark =*, mark options={scale=0.5}]coordinates{(-0.9,0)};
\draw plot[only marks,mark =*, mark options={ scale=0.5}]coordinates{(-0.6,0)};
\draw plot[only marks,mark =*, mark options={scale=0.5}]coordinates{(-0.3,0)};
\draw [thick] (0,-0.1)--(0,0.1);
\draw [thick] (1.2,-0.12)--(1.2,0.12);
\end{tikzpicture} }
+
\begin{tikzpicture}[anchor=base, baseline]
\draw [thick] (-1.2,0)--(1.2,0);
\draw (-0.8,0) arc[radius = 0.2,start angle = 180,end angle = 0];
\draw[-] (0.6,0) to[bend left=75] (1.2,0);
\node [below] at (-1.2,-0.18) {$-t$};
\node [below] at (-0.8,-0.18) {$s_1$};
\node [below] at (-0.4,-0.18) {$s_2$};
\node [below] at (0,-0.18) {0};
\node [below] at (0.6,-0.18) {$s_3$};
\node [below] at (1.2,-0.18) {$t$};
\draw [thick] (-1.2,-0.12)--(-1.2,0.12);
\draw plot[only marks,mark =*, mark options={color=black, scale=0.5}]coordinates{(-0.8,0)};
\draw plot[only marks,mark =*, mark options={color=black, scale=0.5}]coordinates{(-0.4,0)};
\draw [thick] (0,-0.1)--(0,0.1);
\draw plot[only marks,mark =*, mark options={color=black, scale=0.5}]coordinates{(0.6,0)};
\draw [thick] (1.2,-0.12)--(1.2,0.12);
\end{tikzpicture}
+
\begin{tikzpicture}[anchor=base, baseline]
\draw [thick] (-1.2,0)--(1.2,0);
\draw[-] (-0.6,0) to[bend left=75] (0.4,0);
\draw (0.8,0) arc[radius = 0.2,start angle = 180,end angle = 0];
\node [below] at (-1.2,-0.18) {$-t$};
\node [below] at (-0.6,-0.18) {$s_1$};
\node [below] at (0,-0.18) {0};
\node [below] at (0.4,-0.18) {$s_2$};
\node [below] at (0.8,-0.18) {$s_3$};
\node [below] at (1.2,-0.18) {$t$};
\draw [thick] (-1.2,-0.12)--(-1.2,0.12);
\draw plot[only marks,mark =*, mark options={color=black, scale=0.5}]coordinates{(-0.6,0)};
\draw [thick] (0,-0.1)--(0,0.1);
\draw plot[only marks,mark =*, mark options={color=black, scale=0.5}]coordinates{(0.4,0)};
\draw plot[only marks,mark =*, mark options={color=black, scale=0.5}]coordinates{(0.8,0)};
\draw [thick] (1.2,-0.12)--(1.2,0.12);
\end{tikzpicture}
+
\begin{tikzpicture}[anchor=base, baseline]
\draw [thick] (-1.2,0)--(1.2,0);
\draw (0.3,0) arc[radius = 0.15,start angle = 180,end angle = 0];
\draw (0.9,0) arc[radius = 0.15,start angle = 180,end angle = 0];
\node [below] at (-1.2,-0.18) {$-t$};
\node [below] at (0,-0.18) {0};
\node [below] at (0.3,-0.18) {$s_1$};
\node [below] at (0.6,-0.18) {$s_2$};
\node [below] at (0.9,-0.18) {$s_3$};
\node [below] at (1.2,-0.18) {$t$};
\draw [thick] (-1.2,-0.12)--(-1.2,0.12);
\draw [thick] (0,-0.1)--(0,0.1);
\draw plot[only marks,mark =*, mark options={color=black, scale=0.5}]coordinates{(0.3,0)};
\draw plot[only marks,mark =*, mark options={color=black, scale=0.5}]coordinates{(0.6,0)};
\draw plot[only marks,mark =*, mark options={color=black, scale=0.5}]coordinates{(0.9,0)};
\draw [thick] (1.2,-0.12)--(1.2,0.12);
\end{tikzpicture},\\
& \begin{tikzpicture}[anchor=base, baseline] 
 \draw [thick] (-1.2,0) -- (1.2,0);
 \draw [thick] (-1.2,-0.1)--(-1.2,0.1); \draw [thick] (1.2,-0.1)--(1.2,0.1);
 \node [below] at (-1.2,-0.1) {$-t$}; 
  \node [below] at (1.2,-0.1) {$t$}; 
  \draw[-] (-0.8,0) to[bend left=75] (0.6,0);
    \draw[-] (-0.2,0) to[bend left=75] (1.2,0);
  \draw plot[only marks,mark =*, mark options={color=black, scale=0.5}]coordinates {(-0.8,0) (-0.2,0) (0.6,0)};
  \node [below] at (-0.8,0) {$s_1$};  \node [below] at (-0.2,0) {$s_2$};   \node [below] at (0.6,0) {$s_3$};
\end{tikzpicture}
 = 
 \begin{tikzpicture}[anchor=base, baseline]
\draw [thick] (-1.2,0)--(1.2,0);
\draw[-] (-0.9,0) to[bend left=75] (-0.3,0);
\draw[-] (-0.6,0) to[bend left=75] (1.2,0);
\node [below left] at (-1.2,-0.18) {$-t$};
\node [below] at (-0.9,-0.18) {$s_1$};
\node [below] at (-0.6,-0.18) {$s_2$};
\node [below] at (-0.3,-0.18) {$s_3$};
\node [below] at (0,-0.18) {0};
\node [below] at (1.2,-0.18) {$t$};
\draw [thick] (-1.2,-0.12)--(-1.2,0.12);
\draw plot[only marks,mark =*, mark options={color=black, scale=0.5}]coordinates{(-0.9,0)};
\draw plot[only marks,mark =*, mark options={color=black, scale=0.5}]coordinates{(-0.6,0)};
\draw plot[only marks,mark =*, mark options={color=black, scale=0.5}]coordinates{(-0.3,0)};
\draw [thick] (0,-0.1)--(0,0.1);
\draw [thick] (1.2,-0.12)--(1.2,0.12);
\end{tikzpicture}
+
\begin{tikzpicture}[anchor=base, baseline]
\draw [thick] (-1.2,0)--(1.2,0);
\draw[-] (-0.8,0) to[bend left=75] (0.6,0);
\draw[-] (-0.4,0) to[bend left=75] (1.2,0);
\node [below] at (-1.2,-0.18) {$-t$};
\node [below] at (-0.8,-0.18) {$s_1$};
\node [below] at (-0.4,-0.18) {$s_2$};
\node [below] at (0,-0.18) {0};
\node [below] at (0.6,-0.18) {$s_3$};
\node [below] at (1.2,-0.18) {$t$};
\draw [thick] (-1.2,-0.12)--(-1.2,0.12);
\draw plot[only marks,mark =*, mark options={color=black, scale=0.5}]coordinates{(-0.8,0)};
\draw plot[only marks,mark =*, mark options={color=black, scale=0.5}]coordinates{(-0.4,0)};
\draw [thick] (0,-0.1)--(0,0.1);
\draw plot[only marks,mark =*, mark options={color=black, scale=0.5}]coordinates{(0.6,0)};
\draw [thick] (1.2,-0.12)--(1.2,0.12);
\end{tikzpicture}
+
\begin{tikzpicture}[anchor=base, baseline]
\draw [thick] (-1.2,0)--(1.2,0);
\draw[-] (-0.6,0) to[bend left=75] (0.8,0);
\draw[-] (0.4,0) to[bend left=75] (1.2,0);
\node [below] at (-1.2,-0.18) {$-t$};
\node [below] at (-0.6,-0.18) {$s_1$};
\node [below] at (0,-0.18) {0};
\node [below] at (0.4,-0.18) {$s_2$};
\node [below] at (0.8,-0.18) {$s_3$};
\node [below] at (1.2,-0.18) {$t$};
\draw [thick] (-1.2,-0.12)--(-1.2,0.12);
\draw plot[only marks,mark =*, mark options={color=black, scale=0.5}]coordinates{(-0.6,0)};
\draw [thick] (0,-0.1)--(0,0.1);
\draw plot[only marks,mark =*, mark options={color=black, scale=0.5}]coordinates{(0.4,0)};
\draw plot[only marks,mark =*, mark options={color=black, scale=0.5}]coordinates{(0.8,0)};
\draw [thick] (1.2,-0.12)--(1.2,0.12);
\end{tikzpicture}
+
\begin{tikzpicture}[anchor=base, baseline]
\draw [thick] (-1.2,0)--(1.2,0);
\draw[-] (0.3,0) to[bend left=75] (0.9,0);
\draw[-] (0.6,0) to[bend left=75] (1.2,0);
\node [below] at (-1.2,-0.18) {$-t$};
\node [below] at (0,-0.18) {0};
\node [below] at (0.3,-0.18) {$s_1$};
\node [below] at (0.6,-0.18) {$s_2$};
\node [below] at (0.9,-0.18) {$s_3$};
\node [below] at (1.2,-0.18) {$t$};
\draw [thick] (-1.2,-0.12)--(-1.2,0.12);
\draw [thick] (0,-0.1)--(0,0.1);
\draw plot[only marks,mark =*, mark options={color=black, scale=0.5}]coordinates{(0.3,0)};
\draw plot[only marks,mark =*, mark options={color=black, scale=0.5}]coordinates{(0.6,0)};
\draw plot[only marks,mark =*, mark options={color=black, scale=0.5}]coordinates{(0.9,0)};
\draw [thick] (1.2,-0.12)--(1.2,0.12);
\end{tikzpicture},
\\
& \begin{tikzpicture}[anchor=base, baseline] 
 \draw [thick] (-1.2,0) -- (1.2,0);
 \draw [thick] (-1.2,-0.1)--(-1.2,0.1); \draw [thick] (1.2,-0.1)--(1.2,0.1);
 \node [below] at (-1.2,-0.1) {$-t$}; 
  \node [below] at (1.2,-0.1) {$t$}; 
  \draw[-] (-0.8,0) to[bend left=75] (1.2,0);
    \draw[-] (-0.2,0) to[bend left=75] (0.6,0);
  \draw plot[only marks,mark =*, mark options={color=black, scale=0.5}]coordinates {(-0.8,0) (-0.2,0) (0.6,0)};
  \node [below] at (-0.8,0) {$s_1$};  \node [below] at (-0.2,0) {$s_2$};   \node [below] at (0.6,0) {$s_3$};
\end{tikzpicture} = 
\begin{tikzpicture}[anchor=base, baseline]
\draw [thick] (-1.2,0)--(1.2,0);
\draw[-] (-0.9,0) to[bend left=60] (1.2,0);
\draw (-0.6,0) arc[radius = 0.15,start angle = 180,end angle = 0];
\node [below left] at (-1.2,-0.18) {$-t$};
\node [below] at (-0.9,-0.18) {$s_1$};
\node [below] at (-0.6,-0.18) {$s_2$};
\node [below] at (-0.3,-0.18) {$s_3$};
\node [below] at (0,-0.18) {0};
\node [below] at (1.2,-0.18) {$t$};
\draw [thick] (-1.2,-0.12)--(-1.2,0.12);
\draw plot[only marks,mark =*, mark options={color=black, scale=0.5}]coordinates{(-0.9,0)};
\draw plot[only marks,mark =*, mark options={color=black, scale=0.5}]coordinates{(-0.6,0)};
\draw plot[only marks,mark =*, mark options={color=black, scale=0.5}]coordinates{(-0.3,0)};
\draw [thick] (0,-0.1)--(0,0.1);
\draw [thick] (1.2,-0.12)--(1.2,0.12);
\end{tikzpicture}
+
\begin{tikzpicture}[anchor=base, baseline]
\draw [thick] (-1.2,0)--(1.2,0);
\draw[-] (-0.8,0) to[bend left=75] (1.2,0);
\draw[-] (-0.4,0) to[bend left=75] (0.6,0);
\node [below] at (-1.2,-0.18) {$-t$};
\node [below] at (-0.8,-0.18) {$s_1$};
\node [below] at (-0.4,-0.18) {$s_2$};
\node [below] at (0,-0.18) {0};
\node [below] at (0.6,-0.18) {$s_3$};
\node [below] at (1.2,-0.18) {$t$};
\draw [thick] (-1.2,-0.12)--(-1.2,0.12);
\draw plot[only marks,mark =*, mark options={color=black, scale=0.5}]coordinates{(-0.8,0)};
\draw plot[only marks,mark =*, mark options={color=black, scale=0.5}]coordinates{(-0.4,0)};
\draw [thick] (0,-0.1)--(0,0.1);
\draw plot[only marks,mark =*, mark options={color=black, scale=0.5}]coordinates{(0.6,0)};
\draw [thick] (1.2,-0.12)--(1.2,0.12);
\end{tikzpicture}
+
\begin{tikzpicture}[anchor=base, baseline]
\draw [thick] (-1.2,0)--(1.2,0);
\draw[-] (-0.6,0) to[bend left=75] (1.2,0);
\draw (0.4,0) arc[radius = 0.2,start angle = 180,end angle = 0];
\node [below] at (-1.2,-0.18) {$-t$};
\node [below] at (-0.6,-0.18) {$s_1$};
\node [below] at (0,-0.18) {0};
\node [below] at (0.4,-0.18) {$s_2$};
\node [below] at (0.8,-0.18) {$s_3$};
\node [below] at (1.2,-0.18) {$t$};
\draw [thick] (-1.2,-0.12)--(-1.2,0.12);
\draw plot[only marks,mark =*, mark options={color=black, scale=0.5}]coordinates{(-0.6,0)};
\draw [thick] (0,-0.1)--(0,0.1);
\draw plot[only marks,mark =*, mark options={color=black, scale=0.5}]coordinates{(0.4,0)};
\draw plot[only marks,mark =*, mark options={color=black, scale=0.5}]coordinates{(0.8,0)};
\draw [thick] (1.2,-0.12)--(1.2,0.12);
\end{tikzpicture}
+
\boxed{ \begin{tikzpicture}[anchor=base, baseline]
\draw [thick] (-1.2,0)--(1.2,0);
\draw[-] (0.3,0) to[bend left=75] (1.2,0);
\draw (0.6,0) arc[radius = 0.15,start angle = 180,end angle = 0];
\node [below] at (-1.2,-0.18) {$-t$};
\node [below] at (0,-0.18) {0};
\node [below] at (0.3,-0.18) {$s_1$};
\node [below] at (0.6,-0.18) {$s_2$};
\node [below] at (0.9,-0.18) {$s_3$};
\node [below] at (1.2,-0.18) {$t$};
\draw [thick] (-1.2,-0.12)--(-1.2,0.12);
\draw [thick] (0,-0.1)--(0,0.1);
\draw plot[only marks,mark =*, mark options={scale=0.5}]coordinates{(0.3,0)};
\draw plot[only marks,mark =*, mark options={scale=0.5}]coordinates{(0.6,0)};
\draw plot[only marks,mark =*, mark options={scale=0.5}]coordinates{(0.9,0)};
\draw [thick] (1.2,-0.12)--(1.2,0.12);
\end{tikzpicture} } .
\end{split}
\end{equation}
As mentioned previously, the two diagrams in the boxes are already included in \eqref{first bold diagram split}, and do not need to be considered any more. For other diagrams on the right-hand sides, we again replace them with the new diagrams by changing $G_s^{(0)}$ to $\widehat{G}$, and delete the thin diagrams included in these diagrams. By repeating this process, the resummation can be written as 
\begin{multline}
    \label{diagram semi inchworm}
    \mathcal{K}(-t,t) =  \widehat{\mathcal{K}}(-t,t) \coloneqq 
    \begin{tikzpicture}[anchor=base, baseline]
\draw [thick] (-1.2,0)--(1.2,0);
\draw[-] (-0.6,0.1) to[bend left=75] (1.2,0);
\fill [black] (-1.2,-0.1) rectangle (-0.6,0.1);
\node [below] at (-1.2,-0.18) {$-t$};
\node [below] at (-0.6,-0.18) {$s_1$};
\node [below] at (0,-0.18) {0};
\node [below] at (1.2,-0.18) {$t$};
\draw [thick] (-1.2,-0.12)--(-1.2,0.12);
\draw [thick] (-0.6,-0.12)--(-0.6,0.12);
\draw [thick] (0,-0.1)--(0,0.1);
\draw [thick] (1.2,-0.12)--(1.2,0.12);
\end{tikzpicture}
+
\begin{tikzpicture}[anchor=base, baseline]
\draw [thick] (-1.2,0)--(1.2,0);
\draw[-] (0.6,0.1) to[bend left=75] (1.2,0.1);
\fill [black] (0.6,-0.1) rectangle (1.2,0.1);
\node [below] at (-1.2,-0.18) {$-t$};
\node [below] at (0,-0.18) {0};
\node [below] at (0.6,-0.18) {$s_1$};
\node [below] at (1.2,-0.18) {$t$};
\draw [thick] (-1.2,-0.12)--(-1.2,0.12);
\draw [thick] (0,-0.1)--(0,0.1);\draw [thick] (0.6,-0.12)--(0.6,0.12);
\draw [thick] (1.2,-0.12)--(1.2,0.12);
\end{tikzpicture}
+
\begin{tikzpicture}[anchor=base, baseline]
\draw [thick] (-1.2,0)--(1.2,0);
\draw (-0.8,0.1) arc[radius = 0.2,start angle = 180,end angle = 0];
\draw[-] (0.6,0.1) to[bend left=75] (1.2,0.1);
\fill [black] (-1.2,-0.1) rectangle (-0.4,0.1);
\fill [black] (0.6,-0.1) rectangle (1.2,0.1);
\node [below] at (-1.2,-0.18) {$-t$};
\node [below] at (-0.8,-0.18) {$s_1$};
\node [below] at (-0.4,-0.18) {$s_2$};
\node [below] at (0,-0.18) {0};
\node [below] at (0.6,-0.18) {$s_3$};
\node [below] at (1.2,-0.18) {$t$};
\draw [thick] (-1.2,-0.12)--(-1.2,0.12);
\draw [thick, white] (-0.8,-0.1)--(-0.8,0.1);
\draw [thick] (-0.4,-0.12)--(-0.4,0.12);
\draw [thick] (0,-0.1)--(0,0.1);\draw [thick] (0.6,-0.12)--(0.6,0.12);
\draw [thick] (1.2,-0.12)--(1.2,0.12);
\end{tikzpicture}
+
\begin{tikzpicture}[anchor=base, baseline]
\draw [thick] (-1.2,0)--(1.2,0);
\draw[-] (-0.6,0.1) to[bend left=75] (0.4,0.1);
\draw (0.8,0.1) arc[radius = 0.2,start angle = 180,end angle = 0];
\fill [black] (-1.2,-0.1) rectangle (-0.6,0.1);
\fill [black] (0.4,-0.1) rectangle (1.2,0.1);
\node [below] at (-1.2,-0.18) {$-t$};
\node [below] at (-0.6,-0.18) {$s_1$};
\node [below] at (0,-0.18) {0};
\node [below] at (0.4,-0.18) {$s_2$};
\node [below] at (0.8,-0.18) {$s_3$};
\node [below] at (1.2,-0.18) {$t$};
\draw [thick] (-1.2,-0.12)--(-1.2,0.12);
\draw [thick] (-0.6,-0.12)--(-0.6,0.12);
\draw [thick] (0,-0.1)--(0,0.1);\draw [thick] (0.4,-0.12)--(0.4,0.12);
\draw [thick, white] (0.8,-0.1)--(0.8,0.1);
\draw [thick] (1.2,-0.12)--(1.2,0.12);
\end{tikzpicture}\\
+
\begin{tikzpicture}[anchor=base, baseline]
\draw [thick] (-1.2,0)--(1.2,0);
\draw (0.3,0.1) arc[radius = 0.15,start angle = 180,end angle = 0];
\draw (0.9,0.1) arc[radius = 0.15,start angle = 180,end angle = 0];
\fill [black] (0.3,-0.1) rectangle (1.2,0.1);
\node [below] at (-1.2,-0.18) {$-t$};
\node [below] at (0,-0.18) {0};
\node [below] at (0.3,-0.18) {$s_1$};
\node [below] at (0.6,-0.18) {$s_2$};
\node [below] at (0.9,-0.18) {$s_3$};
\node [below] at (1.2,-0.18) {$t$};
\draw [thick] (-1.2,-0.12)--(-1.2,0.12);
\draw [thick] (0,-0.1)--(0,0.1);\draw [thick] (0.3,-0.12)--(0.3,0.12);
\draw [thick, white] (0.6,-0.1)--(0.6,0.1);
\draw [thick, white] (0.9,-0.1)--(0.9,0.1);
\draw [thick] (1.2,-0.12)--(1.2,0.12);
\end{tikzpicture}
+
\begin{tikzpicture}[anchor=base, baseline]
\draw [thick] (-1.2,0)--(1.2,0);
\draw[-] (-0.9,0.1) to[bend left=75] (-0.3,0.1);
\draw[-] (-0.6,0.1) to[bend left=75] (1.2,0);
\fill [black] (-1.2,-0.1) rectangle (-0.3,0.1);
\node [below left] at (-1.2,-0.18) {$-t$};
\node [below] at (-0.9,-0.18) {$s_1$};
\node [below] at (-0.6,-0.18) {$s_2$};
\node [below] at (-0.3,-0.18) {$s_3$};
\node [below] at (0,-0.18) {0};
\node [below] at (1.2,-0.18) {$t$};
\draw [thick] (-1.2,-0.12)--(-1.2,0.12);
\draw [thick, white] (-0.9,-0.1)--(-0.9,0.1);
\draw [thick, white] (-0.6,-0.1)--(-0.6,0.1);
\draw [thick] (-0.3,-0.12)--(-0.3,0.12);
\draw [thick] (0,-0.1)--(0,0.1);\draw [thick] (1.2,-0.12)--(1.2,0.12);
\end{tikzpicture}
+
\begin{tikzpicture}[anchor=base, baseline]
\draw [thick] (-1.2,0)--(1.2,0);
\draw[-] (-0.8,0.1) to[bend left=75] (0.6,0.1);
\draw[-] (-0.4,0.1) to[bend left=75] (1.2,0.1);
\fill [black] (-1.2,-0.1) rectangle (-0.4,0.1);
\fill [black] (0.6,-0.1) rectangle (1.2,0.1);
\node [below] at (-1.2,-0.18) {$-t$};
\node [below] at (-0.8,-0.18) {$s_1$};
\node [below] at (-0.4,-0.18) {$s_2$};
\node [below] at (0,-0.18) {0};
\node [below] at (0.6,-0.18) {$s_3$};
\node [below] at (1.2,-0.18) {$t$};
\draw [thick] (-1.2,-0.12)--(-1.2,0.12);
\draw [thick, white] (-0.8,-0.1)--(-0.8,0.1);
\draw [thick] (-0.4,-0.12)--(-0.4,0.12);
\draw [thick] (0,-0.1)--(0,0.1);\draw [thick] (0.6,-0.12)--(0.6,0.12);
\draw [thick] (1.2,-0.12)--(1.2,0.12);
\end{tikzpicture}
+
\begin{tikzpicture}[anchor=base, baseline]
\draw [thick] (-1.2,0)--(1.2,0);
\draw[-] (-0.6,0.1) to[bend left=75] (0.8,0.1);
\draw[-] (0.4,0.1) to[bend left=75] (1.2,0.1);
\fill [black] (-1.2,-0.1) rectangle (-0.6,0.1);
\fill [black] (0.4,-0.1) rectangle (1.2,0.1);
\node [below] at (-1.2,-0.18) {$-t$};
\node [below] at (-0.6,-0.18) {$s_1$};
\node [below] at (0,-0.18) {0};
\node [below] at (0.4,-0.18) {$s_2$};
\node [below] at (0.8,-0.18) {$s_3$};
\node [below] at (1.2,-0.18) {$t$};
\draw [thick] (-1.2,-0.12)--(-1.2,0.12);
\draw [thick] (-0.6,-0.12)--(-0.6,0.12);
\draw [thick] (0,-0.1)--(0,0.1);\draw [thick] (0.4,-0.12)--(0.4,0.12);
\draw [thick, white] (0.8,-0.1)--(0.8,0.1);
\draw [thick] (1.2,-0.12)--(1.2,0.12);
\end{tikzpicture}
+
\begin{tikzpicture}[anchor=base, baseline]
\draw [thick] (-1.2,0)--(1.2,0);
\draw[-] (0.3,0.1) to[bend left=75] (0.9,0.1);
\draw[-] (0.6,0.1) to[bend left=75] (1.2,0.1);
\fill [black] (0.3,-0.1) rectangle (1.2,0.1);
\node [below] at (-1.2,-0.18) {$-t$};
\node [below] at (0,-0.18) {0};
\node [below] at (0.3,-0.18) {$s_1$};
\node [below] at (0.6,-0.18) {$s_2$};
\node [below] at (0.9,-0.18) {$s_3$};
\node [below] at (1.2,-0.18) {$t$};
\draw [thick] (-1.2,-0.12)--(-1.2,0.12);
\draw [thick] (0,-0.1)--(0,0.1);\draw [thick] (0.3,-0.12)--(0.3,0.12);
\draw [thick, white] (0.6,-0.1)--(0.6,0.1);
\draw [thick, white] (0.9,-0.1)--(0.9,0.1);
\draw [thick] (1.2,-0.12)--(1.2,0.12);
\end{tikzpicture}\\
+
\begin{tikzpicture}[anchor=base, baseline]
\draw [thick] (-1.2,0)--(1.2,0);
\draw[-] (-0.9,0.1) to[bend left=75] (1.2,0);
\draw (-0.6,0.1) arc[radius = 0.15,start angle = 180,end angle = 0];
\fill [black] (-1.2,-0.1) rectangle (-0.3,0.1);
\node [below left] at (-1.2,-0.18) {$-t$};
\node [below] at (-0.9,-0.18) {$s_1$};
\node [below] at (-0.6,-0.18) {$s_2$};
\node [below] at (-0.3,-0.18) {$s_3$};
\node [below] at (0,-0.18) {0};
\node [below] at (1.2,-0.18) {$t$};
\draw [thick] (-1.2,-0.12)--(-1.2,0.12);
\draw [thick, white] (-0.9,-0.1)--(-0.9,0.1);
\draw [thick, white] (-0.6,-0.1)--(-0.6,0.1);
\draw [thick] (-0.3,-0.12)--(-0.3,0.12);
\draw [thick] (0,-0.1)--(0,0.1);\draw [thick] (1.2,-0.12)--(1.2,0.12);
\end{tikzpicture}
+
\begin{tikzpicture}[anchor=base, baseline]
\draw [thick] (-1.2,0)--(1.2,0);
\draw[-] (-0.8,0.1) to[bend left=75] (1.2,0.1);
\draw[-] (-0.4,0.1) to[bend left=75] (0.6,0.1);
\fill [black] (-1.2,-0.1) rectangle (-0.4,0.1);
\fill [black] (0.6,-0.1) rectangle (1.2,0.1);
\node [below] at (-1.2,-0.18) {$-t$};
\node [below] at (-0.8,-0.18) {$s_1$};
\node [below] at (-0.4,-0.18) {$s_2$};
\node [below] at (0,-0.18) {0};
\node [below] at (0.6,-0.18) {$s_3$};
\node [below] at (1.2,-0.18) {$t$};
\draw [thick] (-1.2,-0.12)--(-1.2,0.12);
\draw [thick, white] (-0.8,-0.1)--(-0.8,0.1);
\draw [thick] (-0.4,-0.12)--(-0.4,0.12);
\draw [thick] (0,-0.1)--(0,0.1);\draw [thick] (0.6,-0.12)--(0.6,0.12);
\draw [thick] (1.2,-0.12)--(1.2,0.12);
\end{tikzpicture}
+
\begin{tikzpicture}[anchor=base, baseline]
\draw [thick] (-1.2,0)--(1.2,0);
\draw[-] (-0.6,0.1) to[bend left=75] (1.2,0.1);
\draw (0.4,0.1) arc[radius = 0.2,start angle = 180,end angle = 0];
\fill [black] (-1.2,-0.1) rectangle (-0.6,0.1);
\fill [black] (0.4,-0.1) rectangle (1.2,0.1);
\node [below] at (-1.2,-0.18) {$-t$};
\node [below] at (-0.6,-0.18) {$s_1$};
\node [below] at (0,-0.18) {0};
\node [below] at (0.4,-0.18) {$s_2$};
\node [below] at (0.8,-0.18) {$s_3$};
\node [below] at (1.2,-0.18) {$t$};
\draw [thick] (-1.2,-0.12)--(-1.2,0.12);
\draw [thick] (-0.6,-0.12)--(-0.6,0.12);
\draw [thick] (0,-0.1)--(0,0.1);\draw [thick] (0.4,-0.12)--(0.4,0.12);
\draw [thick, white] (0.8,-0.1)--(0.8,0.1);
\draw [thick] (1.2,-0.12)--(1.2,0.12);
\end{tikzpicture} + \cdots .
\end{multline}
Here we only draw diagrams with up to three time points for conciseness. The formula above shows that most diagrams have a thin line including zero sandwiched by two bold lines. We therefore call the stochastic method based on \eqref{diagram semi inchworm} the ``bold-thin-bold diagrammatic Monte Carlo'' method, or the BTB method in short.

Mathematically, we use the notation $\widehat{\mathcal{Q}}(\bs)$ to denote the set of pairings with nodes $\bs$ that are preserved in \eqref{diagram semi inchworm}. Similar to $\mathcal{Q}^c(\bs)$ for the inchworm method, $\widehat{\mathcal{Q}}(\bs)$ is again a subset of $\mathcal{Q}(\bs)$. However, since the location of time 0 is taken into consideration now in each diagram, the definition of $\widehat{\mathcal{Q}}(\bs)$ relies on the distribution of $\bs$. For example,    
\begin{align*}
      &\widehat{\mathcal{Q}}(s_1,s_2) = \big\{\{(s_1,s_2)\}\big\},\\
      & \widehat{\mathcal{Q}}(s_1,s_2,s_3,s_4)  =
      \begin{cases}
            \big\{\{(s_1,s_3),(s_2,s_4)\},\{(s_1,s_4),(s_2,s_3)\}\big\},
            &\text{if } s_1 \leqslant s_2 \leqslant s_3 < 0 \\ 
           \mathcal{Q}(s_1,s_2,s_3,s_4),
            &\text{if } s_1 \leqslant s_2 < 0 \leqslant s_3 \text{~or~}  s_1 < 0 \leqslant  s_2  \leqslant s_3  \\
            \big\{\{(s_1,s_2),(s_3,s_4)\},\{(s_1,s_3),(s_2,s_4)\}\big\},
            &\text{if } 0 \leqslant s_1 \leqslant s_2 \leqslant s_3 
        \end{cases}.
\end{align*}
In general, the set $\widehat{\mathcal{Q}}(\bs)$ is given by the following definition:
\begin{definition}\label{def:improper_diagrams}
Given $\bs = (s_1,\cdots,s_m)$ and $s_1 < \cdots < s_{\ell-1} < 0 \leqslant s_{\ell} < \cdots < s_m$ for some even $m$ ($\ell=1$ if $s_1 \geqslant 0$), the set $\widehat{\mathcal{Q}}(\boldsymbol{s})$ is defined by its complement in $\mathcal{Q}(\boldsymbol{s})$:
for $\mathfrak{q}\in\mathcal{Q}(\boldsymbol{s})$,
 we have $\mathfrak{q}\not\in\widehat{\mathcal{Q}}(\boldsymbol{s})$ iff there exists a subset 
$$\mathfrak{q}' := \{(s_{i_1},s_{i_2}),\cdots,(s_{i_{2k-1}},s_{i_{2k}})\} \subset \mathfrak{q}$$ satisfying the following two properties:
\begin{enumerate}
    \item[\textbf{(a)}] $\{i_1,\cdots,i_{2k}\}$ is a set of consecutive integers in $\{1,\cdots,m\}$;
    % \item[\textbf{(b)}] Let $n_1 = \min\{i_1,\cdots,i_{2k}\}$ and $n_2 = \max\{ i_1,\cdots,i_{2k}\}$, then $n_1 < n_2 < \ell -1$ or $\ell < n_1 < n_2 < m$.
    \item[\textbf{(b)}] $n_1 < n_2 < \ell -1$ or $\ell < n_1 < n_2 < m$ where $n_1 = \min\{i_1,\cdots,i_{2k}\}$ and $n_2 = \max\{ i_1,\cdots,i_{2k}\}$.
\end{enumerate}
\end{definition}
The set $\mathcal{Q}(\boldsymbol{s})\backslash\widehat{\mathcal{Q}}(\boldsymbol{s})$ includes all possible pairings in the Dyson series \eqref{diagram dyson} that have been removed when we replace thin lines with bold lines. Therefore, the pairings in $\widehat{\mathcal{K}}(-t,t)$ is then given by $\widehat{\mathcal{Q}}(\bs)$.

To better understand the reason for such a definition, we provide in \Cref{fig:improper_diagram} two examples of pairings in $\mathcal{Q}(\boldsymbol{s})\backslash\widehat{\mathcal{Q}}(\boldsymbol{s})$. Note that for all the pairings in $\widehat{\mathcal{K}}(-t,t)$, the last time point $s_m$ always equals $t$. In \Cref{fig:improper1}, the diagram represents the set $\mathfrak{q} = \{(s_1, s_4), (s_2, s_3)\}$ with $\ell = 1$, and its subset $\mathfrak{q}' = \{(s_2, s_3)\}$ (the red arc) satisfies both conditions in \Cref{def:improper_diagrams}. The reason why this diagram should be excluded from the summation \eqref{diagram semi inchworm} is that all the thin diagrams it includes have already been counted in the second diagram in \eqref{diagram semi inchworm}. Similarly, for \Cref{fig:improper2}, which is the diagrammatic presentation of $\mathfrak{q} = \{(s_1, s_3), (s_2, s_4), (s_5, s_7), (s_6, s_8)\}$ with $\ell = 6$, we can choose $\mathfrak{q}' = \{(s_1, s_3), (s_2, s_4)\}$ and find that $\mathfrak{q} \not\in \widehat{\mathcal{Q}}(\boldsymbol{s})$. In fact, this diagram is a partial sum of the eighth diagram in \eqref{diagram semi inchworm}, and thus should be excluded. Generally, the condition \textbf{(b)} requires that the subset $\mathfrak{q}'$ does not include any end points of the ``bold sections'', which refers to the interval $[s_1, t]$ for \Cref{fig:improper1} and the intervals $[-t,s_5]$ and $[s_6,t]$ for \Cref{fig:improper2}. In this case, the diagram represented by $\mathfrak{q} \backslash \mathfrak{q}'$, which has fewer points, embraces $\mathfrak{q}$ completely, and therefore all the diagrams in $\mathcal{Q}(\boldsymbol{s})\backslash\widehat{\mathcal{Q}}(\boldsymbol{s})$ should not be considered in \eqref{diagram semi inchworm}.

%Intuitively, the definition above suggests that any diagram in Dyson series will never appear in \eqref{diagram semi inchworm} if and only if it has a subdiagram such that \textbf{(a)} the subdiagram covers contiguous time points in $\bs$; \textbf{(b)} the subdiagram is not adjacent to the thin line segment and does not include the right most time point. 

%For the pairings in $\widehat{\mathcal{K}}(-t,t)$, the last time point is always fixed at time $t$. Hereafter, we will use $\widehat{\mathcal{Q}}(\bs,t)$ and $\bar{\mathcal{Q}}^*(\bs,t)$ respectively with $\bs = (s_1,\cdots,s_m)$ for some odd $m$ to denote the collection of proper/improper pairings in \eqref{diagram semi inchworm}. Some simple examples for $\mathfrak{q}  \in \bar{\mathcal{Q}}^*(\bs,t)$ are given with the diagrammatic representation in Figure \ref{fig:improper_diagram} to visualize Definition \ref{def:improper_diagrams}, where the subdiagrams marked in red satisfy both condition \textbf{(a)} and \textbf{(b)}. Consequently, these two diagrams will not be included in \eqref{diagram semi inchworm}. In fact, both of them have already appeared in \eqref{first bold diagram split} and thus should not be double counted in higher order expansions of $\widehat{\mathcal{K}}(-t,t)$.    
\begin{figure}[h]
\centering
% \begin{displaymath}
\subfloat[]{\label{fig:improper1}%
\begin{tikzpicture}[anchor=base, baseline]
\draw [thick] (-1.6,0)--(1.6,0);
\draw[-] (0.4,0.1) to[bend left=75] (1.6,0.1);
\draw[-,red] (0.8,0.1) to[bend left=75] (1.2,0.1);
\fill [black] (0.4,-0.1) rectangle (1.6,0.1);
\node [below] at (-1.6,-0.18) {$-t$};
% \node [below] at (-1.6,-0.4) {\scriptsize $(s_{0^-})$};
\node [below] at (0,-0.18) {0};
\node [below] at (0.4,-0.18) {$s_1$};
% \node [below] at (0.4,-0.4) {\scriptsize $(s_{0^+})$};
\node [below] at (0.8,-0.18) {$s_2$};
\node [below] at (1.2,-0.18) {$s_3$};
\node [below] at (1.6,-0.18) {$t$};
\draw [thick] (-1.6,-0.12)--(-1.6,0.12);
\draw [thick] (0,-0.1)--(0,0.1);\draw [thick] (0.4,-0.12)--(0.4,0.12);
\draw [thick, white] (0.8,-0.1)--(0.8,0.1);
\draw [thick, white] (1.2,-0.1)--(1.2,0.1);
\draw [thick] (1.6,-0.12)--(1.6,0.12);
\end{tikzpicture}} \qquad
\subfloat[]{\label{fig:improper2}%
\begin{tikzpicture}[anchor=base, baseline]
\draw [thick] (-2.2,0)--(2.2,0);
\draw[-,red] (-1.8333,0.1) to[bend left=75] (-1.1,0.1);
\draw[-,red] (-1.4667,0.1) to[bend left=75] (-0.73333,0.1);
\draw[-] (-0.36667,0.1) to[bend left=75] (1.4667,0.1);
\draw[-] (0.73333,0.1) to[bend left=75] (2.2,0.1);
\fill [black] (-2.2,-0.1) rectangle (-0.36667,0.1);
\fill [black] (0.73333,-0.1) rectangle (2.2,0.1);
\node [below] at (-2.2,-0.18) {$-t$};
\node [below] at (-1.8333,-0.18) {$s_1$};
\node [below] at (-1.4667,-0.18) {$s_2$};
\node [below] at (-1.1,-0.18) {$s_3$};
\node [below] at (-0.73333,-0.18) {$s_4$};
\node [below] at (-0.36667,-0.18) {$s_5$};
\node [below] at (0,-0.18) {0};
\node [below] at (0.73333,-0.18) {$s_6$};
\node [below] at (1.4667,-0.18) {$s_7$};
\node [below] at (2.2,-0.18) {$t$};
\draw [thick] (-2.2,-0.12)--(-2.2,0.12);
\draw [thick, white] (-1.8333,-0.1)--(-1.8333,0.1);
\draw [thick, white] (-1.4667,-0.1)--(-1.4667,0.1);
\draw [thick, white] (-1.1,-0.1)--(-1.1,0.1);
\draw [thick, white] (-0.73333,-0.1)--(-0.73333,0.1);
\draw [thick] (-0.36667,-0.12)--(-0.36667,0.12);
\draw [thick] (0,-0.1)--(0,0.1);\draw [thick] (0.73333,-0.12)--(0.73333,0.12);
\draw [thick, white] (1.4667,-0.1)--(1.4667,0.1);
\draw [thick] (2.2,-0.12)--(2.2,0.12);
\end{tikzpicture}}
% \end{displaymath}
  \caption{Examples of diagram with pairings $\mathfrak{q} \in \mathcal{Q}(\boldsymbol{s}) \backslash \widehat{\mathcal{Q}}(\boldsymbol{s})$.}
 \label{fig:improper_diagram}
\end{figure}

As a comparison, the diagrams given in Figure \ref{fig:proper_diagram} are pairings in $\widehat{\mathcal{Q}}(\bs)$, which should be counted into the series \eqref{diagram semi inchworm}. Although these diagrams also have proper subsets with consecutive integer indices (red arcs in \Cref{fig:proper_diagram}), the subdiagrams all touch the ends of the bold section, and thus cannot be removed to generate fewer-point diagrams that include them. For example, the first diagram in \Cref{fig:proper_diagram} can be written as $\mathfrak{q} = \{(s_1, s_4), (s_2, s_3)\}$ with $\ell = 2$. If we remove the points $s_2$ and $s_3$ from $\mathfrak{q}$, the resulting diagram is the first diagram in \eqref{diagram semi inchworm}, which does not include $\mathfrak{q}$. Other diagrams in \Cref{fig:proper_diagram} should be preserved due to the same reason.
%as all the subdiagrams covering contiguous time points violate the condition \textbf{(b)}. As a result, they are proper pairings in $\widehat{\mathcal{Q}}(\bs,t)$ that should be included in $\widehat{\mathcal{K}}(-t,t)$. One can easily check that these diagrams cannot be generated by any existing terms in \eqref{diagram semi inchworm} with fewer arcs.

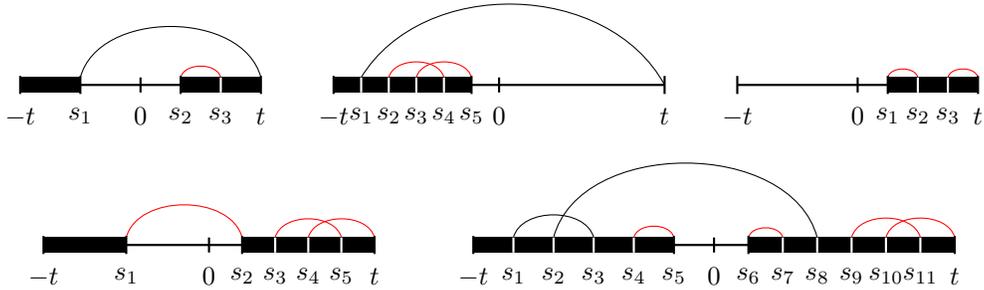
\begin{figure}[h]
\centering
% \begin{displaymath}
\begin{tikzpicture}[anchor=base, baseline]
\draw [thick] (-1.6,0)--(1.6,0);
\draw[-] (-0.8,0.1) to[bend left=75] (1.6,0.1);
\draw[-,red] (0.53333,0.1) to[bend left=75] (1.0667,0.1);
\fill [black] (-1.6,-0.1) rectangle (-0.8,0.1);
\fill [black] (0.53333,-0.1) rectangle (1.6,0.1);
\node [below] at (-1.6,-0.18) {$-t$};
\node [below] at (-0.8,-0.18) {$s_1$};
% \node [below] at (-0.8,-0.4) {\scriptsize $(s_{0^-})$};
\node [below] at (0,-0.18) {0};
% \node [below] at (0,-1) {(a)};
\node [below] at (0.53333,-0.18) {$s_2$};
% \node [below] at (0.53333,-0.4) {\scriptsize $(s_{0^+})$};
\node [below] at (1.0667,-0.18) {$s_3$};
\node [below] at (1.6,-0.18) {$t$};
\draw [thick] (-1.6,-0.12)--(-1.6,0.12);
\draw [thick] (-0.8,-0.12)--(-0.8,0.12);
\draw [thick] (0,-0.1)--(0,0.1);\draw [thick] (0.53333,-0.12)--(0.53333,0.12);
\draw [thick, white] (1.0667,-0.1)--(1.0667,0.1);
\draw [thick] (1.6,-0.12)--(1.6,0.12);
\end{tikzpicture} 
\quad   
\begin{tikzpicture}[anchor=base, baseline]
\draw [thick] (-2.2,0)--(2.2,0);
\draw[-] (-1.8333,0.1) to[bend left=60] (2.2,0);
\draw[-,red] (-1.4667,0.1) to[bend left=75] (-0.73333,0.1);
\draw[-,red] (-1.1,0.1) to[bend left=75] (-0.36667,0.1);
\fill [black] (-2.2,-0.1) rectangle (-0.36667,0.1);
\node [below] at (-2.2,-0.18) {$-t$};
\node [below] at (-1.8333,-0.18) {$s_1$};
\node [below] at (-1.4667,-0.18) {$s_2$};
\node [below] at (-1.1,-0.18) {$s_3$};
\node [below] at (-0.73333,-0.18) {$s_4$};
\node [below] at (-0.36667,-0.18) {$s_5$};
% \node [below] at (-0.36667,-0.4) {\scriptsize $(s_{0^-})$};
\node [below] at (0,-0.18) {0};
% \node [below] at (0,-1) {(b)};
\node [below] at (2.2,-0.18) {$t$};
% \node [below] at (2.2,-0.4) {\scriptsize $(s_{0^+})$};
\draw [thick] (-2.2,-0.12)--(-2.2,0.12);
\draw [thick, white] (-1.8333,-0.1)--(-1.8333,0.1);
\draw [thick, white] (-1.4667,-0.1)--(-1.4667,0.1);
\draw [thick, white] (-1.1,-0.1)--(-1.1,0.1);
\draw [thick, white] (-0.73333,-0.1)--(-0.73333,0.1);
\draw [thick] (-0.36667,-0.12)--(-0.36667,0.12);
\draw [thick] (0,-0.1)--(0,0.1);\draw [thick] (2.2,-0.12)--(2.2,0.12);
\end{tikzpicture}
% \end{displaymath}
\quad 
% \begin{displaymath}
\begin{tikzpicture}[anchor=base, baseline]
\draw [thick] (-1.6,0)--(1.6,0);
\draw[-,red] (0.4,0.1) to[bend left=75] (0.8,0.1);
\draw[-,red] (1.2,0.1) to[bend left=75] (1.6,0.1);
\fill [black] (0.4,-0.1) rectangle (1.6,0.1);
\node [below] at (-1.6,-0.18) {$-t$};
% \node [below] at (-1.6,-0.4) {\scriptsize $(s_{0^-})$};
\node [below] at (0,-0.18) {0};
% \node [below] at (0,-1) {(c)};
\node [below] at (0.4,-0.18) {$s_1$};
% \node [below] at (0.4,-0.4) {\scriptsize $(s_{0^+})$};
\node [below] at (0.8,-0.18) {$s_2$};
\node [below] at (1.2,-0.18) {$s_3$};
\node [below] at (1.6,-0.18) {$t$};
\draw [thick] (-1.6,-0.12)--(-1.6,0.12);
\draw [thick] (0,-0.1)--(0,0.1);\draw [thick] (0.4,-0.12)--(0.4,0.12);
\draw [thick, white] (0.8,-0.1)--(0.8,0.1);
\draw [thick, white] (1.2,-0.1)--(1.2,0.1);
\draw [thick] (1.6,-0.12)--(1.6,0.12);
\end{tikzpicture}
\quad  \quad 
\begin{tikzpicture}[anchor=base, baseline]
\draw [thick] (-2.2,0)--(2.2,0);
\draw[-,red] (-1.1,0.1) to[bend left=75] (0.44,0.1);
\draw[-,red] (0.88,0.1) to[bend left=75] (1.76,0.1);
\draw[-,red] (1.32,0.1) to[bend left=75] (2.2,0.1);
\fill [black] (-2.2,-0.1) rectangle (-1.1,0.1);
\fill [black] (0.44,-0.1) rectangle (2.2,0.1);
\node [below] at (-2.2,-0.18) {$-t$};
\node [below] at (-1.1,-0.18) {$s_1$};
% \node [below] at (-1.1,-0.4) {\scriptsize $(s_{0^-})$};
\node [below] at (0,-0.18) {0};
% \node [below] at (0,-1) {(d)};
\node [below] at (0.44,-0.18) {$s_2$};
% \node [below] at (0.44,-0.4) {\scriptsize $(s_{0^+})$};
\node [below] at (0.88,-0.18) {$s_3$};
\node [below] at (1.32,-0.18) {$s_4$};
\node [below] at (1.76,-0.18) {$s_5$};
\node [below] at (2.2,-0.18) {$t$};
\draw [thick] (-2.2,-0.12)--(-2.2,0.12);
\draw [thick] (-1.1,-0.12)--(-1.1,0.12);
\draw [thick] (0,-0.1)--(0,0.1);\draw [thick] (0.44,-0.12)--(0.44,0.12);
\draw [thick, white] (0.88,-0.1)--(0.88,0.1);
\draw [thick, white] (1.32,-0.1)--(1.32,0.1);
\draw [thick, white] (1.76,-0.1)--(1.76,0.1);
\draw [thick] (2.2,-0.12)--(2.2,0.12);
\end{tikzpicture} \quad  \quad 
% \end{displaymath}
\begin{tikzpicture}[anchor=base, baseline]
\draw [thick] (-3.2,0)--(3.2,0);
\draw[-] (-2.6667,0.1) to[bend left=75] (-1.6,0.1);
\draw[-] (-2.1333,0.1) to[bend left=75] (1.3714,0.1);
\draw[-,red] (-1.0667,0.1) to[bend left=75] (-0.53333,0.1);
\draw[-,red] (0.45714,0.1) to[bend left=75] (0.91429,0.1);
\draw[-,red] (1.8286,0.1) to[bend left=75] (2.7429,0.1);
\draw[-,red] (2.2857,0.1) to[bend left=75] (3.2,0.1);
\fill [black] (-3.2,-0.1) rectangle (-0.53333,0.1);
\fill [black] (0.45714,-0.1) rectangle (3.2,0.1);
\node [below] at (-3.2,-0.18) {$-t$};
\node [below] at (-2.6667,-0.18) {$s_1$};
\node [below] at (-2.1333,-0.18) {$s_2$};
\node [below] at (-1.6,-0.18) {$s_3$};
\node [below] at (-1.0667,-0.18) {$s_4$};
\node [below] at (-0.53333,-0.18) {$s_5$};
% \node [below] at (-0.53333,-0.4) {\scriptsize $(s_{0^-})$};
\node [below] at (0,-0.18) {0};
\node [below] at (0.45714,-0.18) {$s_6$};
% \node [below] at (0.45714,-0.4) {\scriptsize $(s_{0^+})$};
\node [below] at (0.91429,-0.18) {$s_7$};
\node [below] at (1.3714,-0.18) {$s_8$};
\node [below] at (1.8286,-0.18) {$s_9$};
\node [below] at (2.2857,-0.18) {$s_{10}$};
\node [below] at (2.7429,-0.18) {$s_{11}$};
\node [below] at (3.2,-0.18) {$t$};
\draw [thick] (-3.2,-0.12)--(-3.2,0.12);
\draw [thick, white] (-2.6667,-0.1)--(-2.6667,0.1);
\draw [thick, white] (-2.1333,-0.1)--(-2.1333,0.1);
\draw [thick, white] (-1.6,-0.1)--(-1.6,0.1);
\draw [thick, white] (-1.0667,-0.1)--(-1.0667,0.1);
\draw [thick] (-0.53333,-0.12)--(-0.53333,0.12);
\draw [thick] (0,-0.1)--(0,0.1);\draw [thick] (0.45714,-0.12)--(0.45714,0.12);
\draw [thick, white] (0.91429,-0.1)--(0.91429,0.1);
\draw [thick, white] (1.3714,-0.1)--(1.3714,0.1);
\draw [thick, white] (1.8286,-0.1)--(1.8286,0.1);
\draw [thick, white] (2.2857,-0.1)--(2.2857,0.1);
\draw [thick, white] (2.7429,-0.1)--(2.7429,0.1);
\draw [thick] (3.2,-0.12)--(3.2,0.12);
\end{tikzpicture}
  \caption{Examples of diagrams with pairing $\mathfrak{q}\in \widehat{\mathcal{Q}}(\bs)$.}
 \label{fig:proper_diagram}
\end{figure}

At this point, we can express the resummation $\widehat{\mathcal{K}}(-t,t)$ mathematically by 
\begin{equation}
 \label{K semi inchworm}
    \widehat{\mathcal{K}}(-t,t) = \sum_{\substack{m=1 \\ m \text{~is odd}}}^{+\infty}  \ii^{m+1} \int_{-t\leqslant \bs \leqslant t} \dd \bs  (-1)^{\#\{\bs < 0\}} \widehat{\mathcal{U}}(-t, \bs , t)  \widehat{\mathcal{L}}_b(\bs,t)  
\end{equation}
where 
\begin{align*}
   & \widehat{\mathcal{U}}(-t, s_1,\cdots,s_m , t) =   \ \widehat{G}(s_m,t) W_s
    \widehat{G}(s_{m-1},s_m) W_s
    \cdots
    W_s \widehat{G}(s_1,s_2) 
    W_s \widehat{G}(-t,s_1), \\
    & \widehat{\mathcal{L}}_b(s_1,\cdots,s_{m+1})  =  \ \sum_{\mathfrak{q}\in\widehat{\mathcal{Q}}(s_1,\cdots,s_{m+1})}
    \prod_{(s_j,s_k)\in\mathfrak{q}}
    B(s_j,s_k).
\end{align*}

The evolution of $G(-t,t)$ is again solved numerically based on the governing equation \eqref{integro_differential_eq} together with the resummation $\widehat{\mathcal{K}}(-t,t)$. This BTB method can be considered as an intermediate approach between Dyson series and inchworm method. Compared to the Dyson series, the BTB method has faster convergence with respect to $m$ as many thin line segments of the diagrams in $\widehat{\mathcal{K}}(-t,t)$ are now replaced by bold ones which include infinite subdiagrams. However, as the propagator crossing time 0 remains to be a thin line segment for the purpose of reuse algorithm, the convergence of BTB method should be slower than that of inchworm method. This also implies that the number of diagrams considered in BTB method is smaller than Dyson series but larger than the inchworm method (see Table \ref{tab:diagram_numbers}). Hence, if we directly evaluate  the three bath influence functionals $\mathcal{L}_b(\bs,t)$, $\widehat{\mathcal{L}}_b(\bs,t)$ and $\mathcal{L}_b^c(\bs,t)$ for the same $m$ by summing up the corresponding diagrams, the computation of $\mathcal{L}_b^c(\bs,t)$ is the fastest, and the computation of $\mathcal{L}_b(\bs,t)$ the slowest.

\begin{table}[h]
\centering
\begin{tabular}{| c | c | c | c | c |} 
\hline
 $m$\rule{0pt}{12pt} &
 $\vert \mathcal{Q}(\bs,t) \vert$ &
 $\vert \mathcal{Q}^c(\bs,t) \vert $&
 $\displaystyle \min_{\bs} \vert \widehat{\mathcal{Q}}(\bs,t) \vert$ &
 $\displaystyle \max_{\bs}  \vert \widehat{\mathcal{Q}}(\bs,t) \vert$
 \\ \hline
 1 & 1 & 1 & 1 & 1 \\
 3 & 3 & 1 & 2 & 3 \\
 5 & 15 & 4 & 6 & 12 \\
 7 & 105 & 27 & 36 & 66 \\
 9 & 945 & 248 & 310 & 510\\
 11 & 10395 & 2830 & 3396 & 5100 \\ \hline
\end{tabular}
\caption{Number of pairings in $\mathcal{Q}(\bs,t)$, $\mathcal{Q}^c(\bs,t)$ and $\widehat{\mathcal{Q}}(\bs,t)$ for $\bs = (s_1,\cdots,s_m)$. Note that the value of $\vert \widehat{\mathcal{Q}}(\bs,t) \vert$ relies on the distribution of $\bs$ and thus we list both the smallest and largest value.
\label{tab:diagram_numbers}}
\end{table}

Although the BTB method has more diagrams than the inchworm method, the BTB method can save a significant amount of computational cost by using the recurrence relation \eqref{recurrence}. More importantly, later we will see the inchworm method needs to deal with two time variables, while the BTB method is essentially a one-dimensional problem. The details will be discussed in the following subsection.

\subsection{Numerical scheme for the BTB method}
  \label{sec:semi_inchworm_numerical_method}
We now elaborate the iterative algorithm for BTB method. With the linear basis \eqref{basis_2}, we can define     
\begin{equation}\label{Uij*}
    \Uij(-t,\boldsymbol{s},t)
    = \widehat{\mathscr{G}}_{ij}(s_m,t) W_s \widehat{\mathscr{G}}_{ij}(s_{m-1},s_m)
    W_s \cdots W_s 
    \widehat{\mathscr{G}}_{ij}(s_1,s_2)
    W_s
    \widehat{\mathscr{G}}_{ij}(-t,s_1)
\end{equation}
and 
\begin{equation}
    \widehat{\mathscr{K}{}}_{\!\!ij}(-t,t)
    = \sum_{\substack{m=1\\m \text{ is odd}}} 
    \ii^{m+1}
    \int_{-t\leqslant \boldsymbol{s} \leqslant t}
    \dd \boldsymbol{s} (-1)^{\#\{\boldsymbol{s}<0\}}
    \Uij(-t,\boldsymbol{s},t) \widehat{\mathcal{L}}_b(\boldsymbol{s},t),
\end{equation}
which are analogous to \eqref{U0 basis} and \eqref{Sscr}, respectively. Then resummation \eqref{K semi inchworm} is then formulated as 
\begin{equation*}
    \widehat{\mathcal{K}}(-t,t) 
    = \sum_{i,j} a_{ij} \widehat{\mathscr{K}{}}_{\!\!ij}(-t,t)
\end{equation*}
where $a_{ij}$ is given as \eqref{aij}. The shift invariance \eqref{G shift invariance} of $G(s_\ii,s_\ff)$ 
and the recurrence relation \eqref{recurrence} of $\widehat{\mathscr{G}}_{ij}(s_\ii,s_\ff)$
yield
\begin{equation*}
    \Uij(-t-\dt,\mathcal{I}_{\dt}(\boldsymbol{s}),t+\dt) =  \sum_{k,l}b_{kl}^{ij} \widehat{\mathscr{U}{}}_{\!\!kl}(-t,\boldsymbol{s},t)
\end{equation*}
following the analysis for Lemma \ref{lemma:Uscr_shift}. In addition, it is also easy to check that  
\begin{displaymath}
    \widehat{\mathcal{L}}_b(\mathcal{I}_{\dt}(\boldsymbol{s}),t+\dt)
    = \widehat{\mathcal{L}}_b(\boldsymbol{s},t).
\end{displaymath}
Eventually, we may arrive at a recurrence relation similar to \eqref{Sscr_shift_property} which plays a key role in the reuse algorithm computing $\widehat{\mathcal{K}}(-t,t)$ as 
\begin{equation}\label{Sscr_shift_property_2}
    \widehat{\mathscr{K}{}}_{\!\!ij} (-t-\dt, t+\dt)
    = \sum_{k,l} b_{kl}^{ij} \widehat{\mathscr{K}{}}_{\!\!kl}(-t,t)
    + \widehat{\mathscr{D}}_{ij}^{\dt}(t+\dt),
\end{equation}
where 
\begin{equation} 
    \begin{split}
        \widehat{\mathscr{D}}_{ij}^{\dt}(t+\dt)
        =&\sum_{\substack{m=1\\m\text{ is odd}}}^{+\infty} 
        \ii^{m+1} \int_{T_{t+\dt}^{(m)}}
        \dd \boldsymbol{s} (-1)^{\#\{\boldsymbol{s} < 0\}}
        \Uij(-t-\dt,\boldsymbol{s},t+\dt)
        \widehat{\mathcal{L}}_b(\boldsymbol{s},t+\dt),
    \end{split}
\end{equation}
and the coefficient $b_{ij}$ satisfies \eqref{bij}, which depends on the time step $\Delta t$. In order to compute $\widehat{\mathscr{D}}_{ij}^{\dt}(t+\dt)$ numerically, we need the following approximations:
\begin{itemize}
\item Truncation of the infinite series at $m = \bar{M}$;
\item Application of the Monte Carlo method to compute the high-dimensional integrals;
\item A numerical scheme to calculate the system factor $\Uij(-t-\dt,\boldsymbol{s},t+\dt)$.
\end{itemize} 
While the first two points are also implemented in 
\Cref{Sec_Dyson_reuse}, the third point requires some extra effort for the BTB method. Using the shift invariance \eqref{G shift invariance} and the conjugate property \eqref{G conjugate symmetry} of the full propagator, we can rewrite the definition of $\widehat{\mathscr{G}}_{ij}(s_\ii, s_\ff)$ as
    \begin{equation} \label{hatGij}
        \widehat{\mathscr{G}}_{ij} (s_\ii,s_\ff)
        = \begin{cases}
            G^{\circ}(s_\ff-s_\ii)^{\dagger},
            &\text{ if }s_\ii < s_\ff < 0, \\
            G^{\circ}(s_\ff-s_\ii), &
            \text{ if }0 < s_\ii < s_\ff, \\
            \e^{\ii s_\ff H_s} \dyad{i}{j} \e^{\ii s_\ii H_s},
            &\text{ if } s_\ii < 0 < s_\ff,
        \end{cases}
\end{equation}
where $G^{\circ}(\cdot)$ is defined by the limit
\begin{displaymath}
G^{\circ}(s) = \lim_{\epsilon \rightarrow 0^+} G(\epsilon, s), \qquad \forall s > 0.
\end{displaymath}
Thus, according to the definition \eqref{Uij*}, the computation of $\Uij(-t,\boldsymbol{s},t)$ requires the knowledge of $G^{\circ}(s)$ for $s > 0$. In our implementation, we will first compute $G^{\circ}_k \approx G^{\circ}(k\Delta t)$ for all $k = 1,\cdots,N$. Let
\begin{equation}
\label{eq:G0}
G^{\circ}_0 = \lim_{s \rightarrow 0^+} G^{\circ}(s) = I.
\end{equation}
Then for any $s \in (0, N\Delta t]$, the value of $G^{\circ}(s)$ can be approximated by piecewise linear interpolation. The interpolant $G_I^{\circ}(s)$ can be written as
\begin{equation}\label{interpolation}
    G_I^{\circ}(s) = 
    \frac{t_{k+1}-s}{\dt}
    G^{\circ}_k 
    + \frac{s -t_k}{\dt}
    G^{\circ}_{k+1}, \quad s \in [t_k, t_{k+1}], \quad k \geqslant 0.
\end{equation}
We can then let
    \begin{equation*}
     \widehat{\mathscr{G}}_{ij}^I (s_\ii,s_\ff)
        = \begin{cases}
            G^{\circ}_I(s_\ff-s_\ii)^{\dagger},
            &\text{ if }s_\ii < s_\ff < 0, \\
            G^{\circ}_I(s_\ff-s_\ii), &
            \text{ if }0 < s_\ii < s_\ff, \\
            \e^{\ii s_\ff H_s} \dyad{i}{j} \e^{\ii s_\ii H_s},
            &\text{ if } s_\ii < 0 < s_\ff,
        \end{cases}
\end{equation*}
and thus
\begin{equation}\label{tUij*}
  \UijI(-t,\bs,t)  = \widehat{\mathscr{G}}_{ij}^I(s_m,t) W_s \widehat{\mathscr{G}}_{ij}^I(s_{m-1},s_m)  W_s \cdots W_s  \widehat{\mathscr{G}}_{ij}^I(s_1,s_2) W_s \widehat{\mathscr{G}}_{ij}^I(-t,s_1)
\end{equation}
is a practical approximation of $\Uij(-t,\bs,t)$.
At this point, we can again solve $G(-t,t)$ numerically by Heun's scheme \eqref{heun} following the procedures in Algorithm \ref{algo:dyson} where \Cref{Line7,Line8} are now replaced by the approximation of $\widehat{\mathscr{D}}_{ij}^{\dt}(-t_{n+1},t_{n+1})$: 
\begin{equation*}
     \widehat{\mathscr{D}}_{ij}^{n+1} =  \frac{1}{\mathcal{M}^{(m)}_{n+1}} \sum_{k=1}^{\mathcal{M}^{(m)}_{n+1}} \widehat{\mathscr{S}\,}{}^{\!(k)}_{\!\!\!ij}
\end{equation*}
where 
\begin{equation*}
    \widehat{\mathscr{S}\,}{}^{\!(k)}_{\!\!\!ij}
        =  \left|T^{(m)}_{t_n+\Delta t} \right| 
        \cdot \ii^{m+1} (-1)^{\#\{\bs_{n+1}^{(k)} < 0\}}
         \UijI(-t_{n+1},\bs_{n+1}^{(k)},t_{n+1})
        \widehat{\mathcal{L}}_b(\bs_{n+1}^{(k)},t_{n+1}).
\end{equation*}

The evaluation of $G^{\circ}_k$ is based on the following generalized inchworm integro-differential equation \cite{cai2020inchworm}
\begin{equation}
    \frac{\dd G^{\circ}(t)}{\dd t}
    = \ii H_s G^{\circ}(t) + \sum_{\substack{m=1\\m\text{ is odd}}}^{+\infty}
    \ii^{m+1}
    \int_{0 \leqslant \boldsymbol{s} \leqslant t}
    \dd \boldsymbol{s} \, \mathcal{U}(0,\bs,t)
    \mathcal{L}_b^c (\boldsymbol{s},t) \text{~for~} t > 0,
    \label{inchworm_integro_differential_equation}
\end{equation}
where
\begin{equation*}
\mathcal{U}(0,\bs,t) = W_s G^{\circ}(t-s_m) W_s G^{\circ}(s_m-s_{m-1}) W_s \cdots W_s G^{\circ}(s_2-s_1) W_s  G^{\circ}(s_1).
\end{equation*}
The numerical solver of \eqref{inchworm_integro_differential_equation} can again be constructed by combining Huen's method and the Monte Carlo intergration:
%which extends the governing equation for $G(-t,t)$ \eqref{integro_differential_eq} with inchworm resummation $\mathcal{K}^c(-t,t)$ \eqref{K inchworm} to an two-variable formulation that solves a general full propagator $G(s_\ii,s_\ff)$. $G_{0,k}$ for $k=1,\cdots,N$ can be obtained by solving \eqref{inchworm_integro_differential_equation} numerically by the following scheme which again combines Heun's method and Monte Carlo integration: 
\begin{equation}\label{Ginitials}
\begin{split}
    &G_{\star}^{\circ} = G^{\circ}_k + \dt \left[
    \ii H_s G^{\circ}_k + \sum_{\substack{m=1\\m \text{ is odd}}}^{\bar{M}}
    \frac{(t_k)^m}{m!\mathcal{N}_k^{(m)}}
    \sum_{i=1}^{\mathcal{N}_k^{(m)}}
    \ii^{m+1}
    W_s
    \mathcal{U}^I(0,\boldsymbol{s}_k^{(i)},t_k)
    \mathcal{L}_b^c (\boldsymbol{s}_k^{(i)},t_k)
    \right], \\
    &G^{\circ}_{\star\star} = G^{\circ}_{\star} + \dt \left[
    \ii H_s G_{\star}^{\circ} + \sum_{\substack{m=1\\m\text{ is odd}}}^{\bar{M}}
    \frac{(t_{k+1})^m}{m!\mathcal{N}_{k+1}^{(m)}}
        \sum_{i=1}^{\mathcal{N}_{k+1}^{(m)}} \ii^{m+1} W_s \mathcal{U}_{\star}^I(0,\boldsymbol{s}_{k+1}^{(i)},t_{k+1})
    \mathcal{L}_b^c (\boldsymbol{s}_{k+1}^{(i)},t_{k+1})
    \right],\\
    &G^{\circ}_{k+1} = \frac{1}{2}\left(G^{\circ}_k+G^{\circ}_{\star\star}\right).
\end{split}
\end{equation}
The initial condition of the scheme above is given as $G^{\circ}_{0} = I$ according to \eqref{eq:G0}. The interpolating functional $\mathcal{U}^I$ is defined similarly to \eqref{tUij*}:
\begin{displaymath}
\mathcal{U}^I(0,\bs,t) =  W_s G_I^{\circ}(t-s_m) W_s G_I^{\circ}(s_m-s_{m-1}) W_s \cdots W_s G_I^{\circ}(s_2-s_1) W_s  G_I^{\circ}(s_1),
\end{displaymath}
and $\mathcal{U}^I_{\star}$ in the second stage is given by the same formula as $\mathcal{U}^I$ with all $G^{\circ}_\star$ replaced by the interpolated function $G^{\circ}_{\star I}$ satisfying
\begin{displaymath}
  G^{\circ}_{\star I}(s) = 
  \begin{cases}
      G^{\circ}_{I}(s) , & \text{~if~} s \in [t_\ell,t_{\ell+1}] \text{~for~} 0 \leqslant  \ell \leqslant k-1 \\ 
       \frac{t_{k+1}-s}{\dt}
    G^{\circ}_k 
    + \frac{s -t_k}{\dt}
    G^{\circ}_{\star} , & \text{~if~} s \in [t_{k},t_{k+1}] 
  \end{cases}.
\end{displaymath}
The samples $\{\boldsymbol{s}_m^{(i)}\}_{i=1}^{\mathcal{N}_k^{(m)}}$ are drawn from uniform distribution $U(0\leqslant \boldsymbol{s}_k \leqslant t_k)$.
Similar to \eqref{Mnm1}, the number of samples $\mathcal{N}_k^{(m)}$ is again set to be proportional to the volume of the integration domain:
\begin{equation*}
    \mathcal{N}_k^{(m)} = \mathcal{N}_0
    \frac{(t_k)^m}{(m-1)!!}\mathcal{B}^{\frac{m+1}{2}}.
\end{equation*}
In our implementation, we choose $\mathcal{N}_0$ to be the same as $\mathcal{M}_0$ used in \eqref{Mnm1}.

As a summary, our BTB method consists of two stages. The first stage solves the function $G^{\circ}(t)$ from \eqref{inchworm_integro_differential_equation} using the scheme \eqref{Ginitials}, and the second stage runs a procedure similar to \Cref{algo:dyson}, with suitable alterations to change bare diagrams to the BTB diagrams. Due to the larger number of samples in the second stage, we expect that the computational cost of the second stage is dominating. Compared with \Cref{algo:dyson}, the BTB method may also have lower computational cost for large $\bar{M}$ due to the smaller number of diagrams. This will be verified by our numerical tests in the next section.

At the end of this section, we would like to compare the BTB method with the inchworm Monte Carlo method. Both methods use bold diagrams to mitigate the numerical sign problem. However, the BTB method only uses bold lines on one side of the zero, allowing us to use the shift invariance to reduce all the bold lines to a one-variable function (see \eqref{hatGij}); the inchworm method uses $G(s_\ii, s_\ff)$ with $s_\ii < 0$ and $s_\ff > 0$, so that the full propagator must be computed as a two-dimensional function. Thus, compared with the inchworm method, the BTB method reduces the number of propagators to be computed by a factor of $O(N)$. Moreover, although the inchworm method also allows reusing computed influence functionals as proposed in \cite{cai2022fast}, all the samples and the influence functionals must be stored, which may lead to a significant memory cost. The superiority of the inchworm method is its faster convergence with respect to $m$. Based on these reasons, the BTB method is believed to be superior when the coupling is not too strong, while the inchworm method is preferable when the numerical sign problem is severe.

\section{Numerical results}
\label{section_numerical_results}
In the following experiments,
 we choose the spectral density $J(\omega)$ in \eqref{bath_function_B} as
\begin{equation}
    \label{spectral_density}
    J(\omega) = \frac{\pi}{2} \sum_j \frac{c_j^2}{\omega_j} \delta(\omega - \omega_j)
\end{equation}
with $\delta$ being the Dirac delta function 
and $\omega_j,c_j$ given by
\begin{equation}
    \label{frequencies_and_coupling_intensity}
    \omega_j = - \omega_c \ln
    \left(
    1- \frac{j}{L} \left(1-\e^{-\omega_{\max}/\omega_c}\right)
    \right),
    \quad
    c_j = \omega_j \sqrt{\frac{\xi\omega_c}{L}
    \left(1-\e^{-\omega_{\max}/\omega_c}\right)},
    \quad j=1,\dots,L,
\end{equation}
where $L$ is the number of harmonic oscillators in the bath.
These parameters correspond to the bath with Ohmic spectral density \cite{makri1999linear}.
With the spectral density $J(\omega)$,
 function $B(\tau_1,\tau_2)$ \eqref{bath_function_B} becomes
\begin{equation*}
\begin{split}
    \label{Bath_function_B_with_spectral}
    B(\tau_1,\tau_2) &= \frac{1}{2}\sum_j \frac{c_j^2}{\omega_j} 
    \left(
    \coth\left(\frac{\beta\omega_j}{2}\right)
    \cos(\omega_j \Delta \tau ) - \ii \sin(\omega_j \Delta \tau )
    \right)
\end{split}
\end{equation*}
where $\Delta \tau = \vert\tau_1\vert - \vert \tau_2\vert$.
From the above expression,
 the amplitude of the function $B$
 depends on the parameters $\omega_c,\xi,\beta$ and $L$.
We therefore plot the amplitude of the function
 for the sets of parameters we use
 in this paper in \Cref{amplitude}.
From \Cref{amplitude},
 the amplitude of $B(\tau_1,\tau_2)$
 increases as either $\xi$ or $\omega_c$ increases.
Larger amplitude of function $B$ indicates stronger coupling between the system and the bath,
 which in general brings difficulty to numerical simulations.
  
\begin{figure}[ht]
    \includegraphics[width=0.6\textwidth]{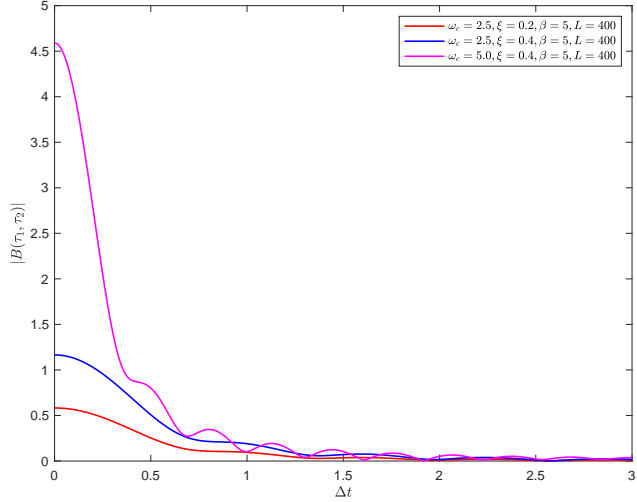}
    \caption{The amplitude of function $B(\tau_1,\tau_2)$ for different parameters.}
    \label{amplitude}
\end{figure}

For all the experiments,
 the number of harmonic oscillators is assumed to be $L = 400$, 
 the maximum frequency is chosen as $\omega_{\max} = 4\omega_c$
 and time step $\dt$ is chosen to be 0.05. 
Other parameters will be specified for each experiment.
All the experiments in this paper are run on the CPU model ``Intel\textsuperscript{\textregistered} Core\texttrademark \ i5-8257U'', and four threads are used in the numerical simulations.

\subsection{Experiments with different biases}
We first check the validity of our method
 by using the following parameters:
\begin{equation*}
    \Delta = 1,\quad
    \omega_c = 2.5,\quad
    \beta = 5,\quad
    \xi=0.2
\end{equation*}
with two different biases $\epsilon=0$ and $\epsilon=1$.
The results of both Dyson series and BTB method are shown in \Cref{Test_1_2}.
The i-QuAPI results are also plotted for reference.
\begin{figure}[ht]
    \subfloat[$\epsilon = 0$\label{Test_1}]{\includegraphics[width= 0.45\textwidth]{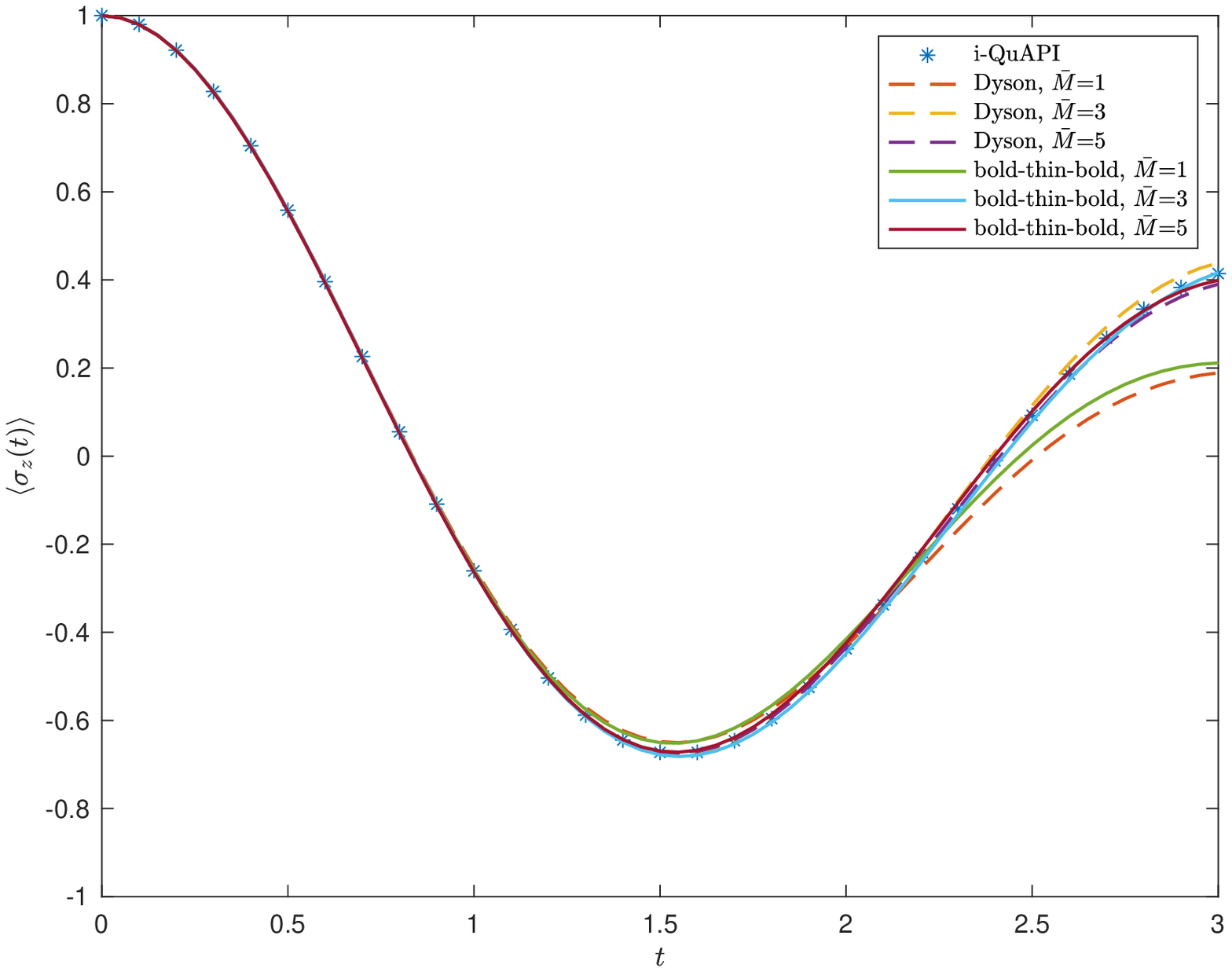}} \quad
    \subfloat[$\epsilon = 1$\label{Test_2}]{\includegraphics[width= 0.45\textwidth]{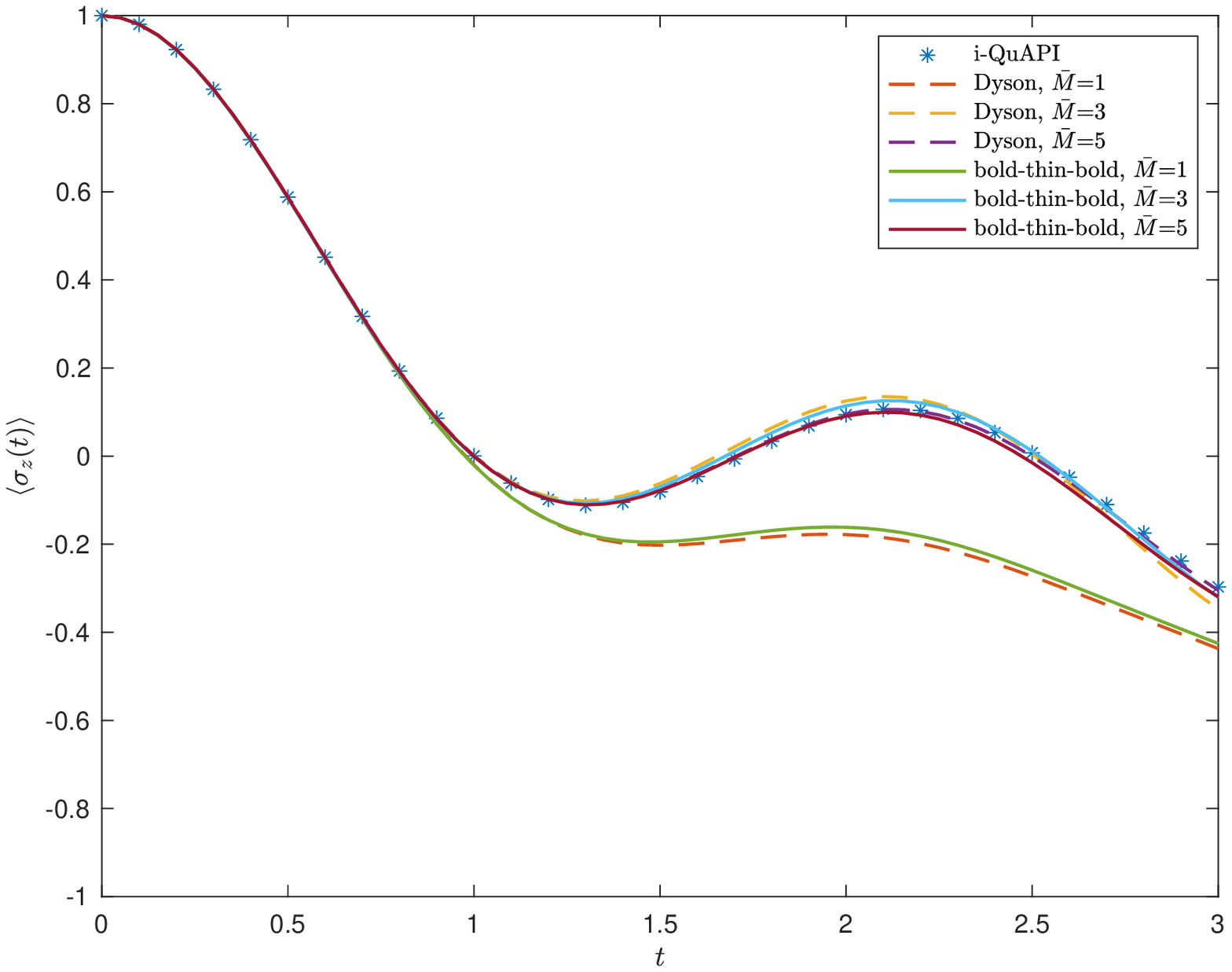}} \quad
    \caption{Evolution of $\expval{\hat{\sigma}_z(t)}$ with $\epsilon=0$ and $\epsilon=1$. 
    The number of samples is determined based on \eqref{Mnm} with $\mathcal{M}_0 = 10^6$.}
    \label{Test_1_2}
\end{figure}
\Cref{Test_1_2} shows that truncating
\eqref{K dyson}\eqref{K semi inchworm} with
$\bar{M}=1$ leads to significant numerical errors around $t = 3$ in both cases, while remarkable improvement is made when $\bar{M}$ is increased to $3$, yielding curves well matching the ones generated by i-QuAPI and $\bar{M} = 5$. The convergence towards the same curve in both methods validates our iterative algorithm for the spin-boson model with or without energy difference.

With these sets of parameters,
 the two-point correlation $B(\cdot,\cdot)$ defined in \eqref{bath_function_B} 
 has a relatively small amplitude according to \Cref{amplitude},
leading to a fast convergence with respect to $\bar{M}$ for both series in \eqref{K dyson} and \eqref{K semi inchworm}.
 In our implementation, the empirical constant $\mathcal{B}$ in \eqref{Mnm}
 is chosen as a sixth of the largest amplitude of the function $B(\tau_1,\tau_2)$, which equals $0.0971$ in these two examples. 
Consequently, the number of samples required in the Monte Carlo sampling is relatively small, and the computational time for both methods is less than 20 seconds.
In our experiments, we find that the simulation based on the Dyson series is slightly faster than the BTB method. This is because BTB method needs extra computations on $G_{k}^\circ$ by \eqref{Ginitials} which is not required for Dyson series. A closer study on the computational time will be carried out in the next subsection.

\subsection{Experiments for different frequencies}
\label{test_difference_frequencies}
One critical problem of the Monte Carlo based methods 
 for open quantum systems is the sign problem.
The analysis in \cite[Section 5]{cai2020inchworm} shows that the severity of the numerical sign problem is mainly characterized by the magnitude of the two-point correlation function $B(\cdot, \cdot)$.
To test the performance of both methods in the case where $B(\cdot,\cdot)$ has a larger amplitude, we choose the following parameters:
\begin{equation*}
    \Delta = 1,\quad
    \epsilon = 1, \quad
    \beta = 5,\quad
    \xi=0.4
\end{equation*}
with frequencies $\omega_c=2.5$ and $\omega_c=5$.
According to \Cref{amplitude}, 
 the increase of both $\xi$ and $\omega_c$
 will lead to a greater amplitude of $B(\cdot,\cdot)$, indicating
 stronger coupling between the system and the bath.
The results of both methods up to $\bar{M} = 7$ are given in \Cref{Test_3_4}.
\begin{figure}[ht]
    \subfloat[$\omega_c = 2.5$\label{Test_3}]{\includegraphics[width= 0.45\textwidth]{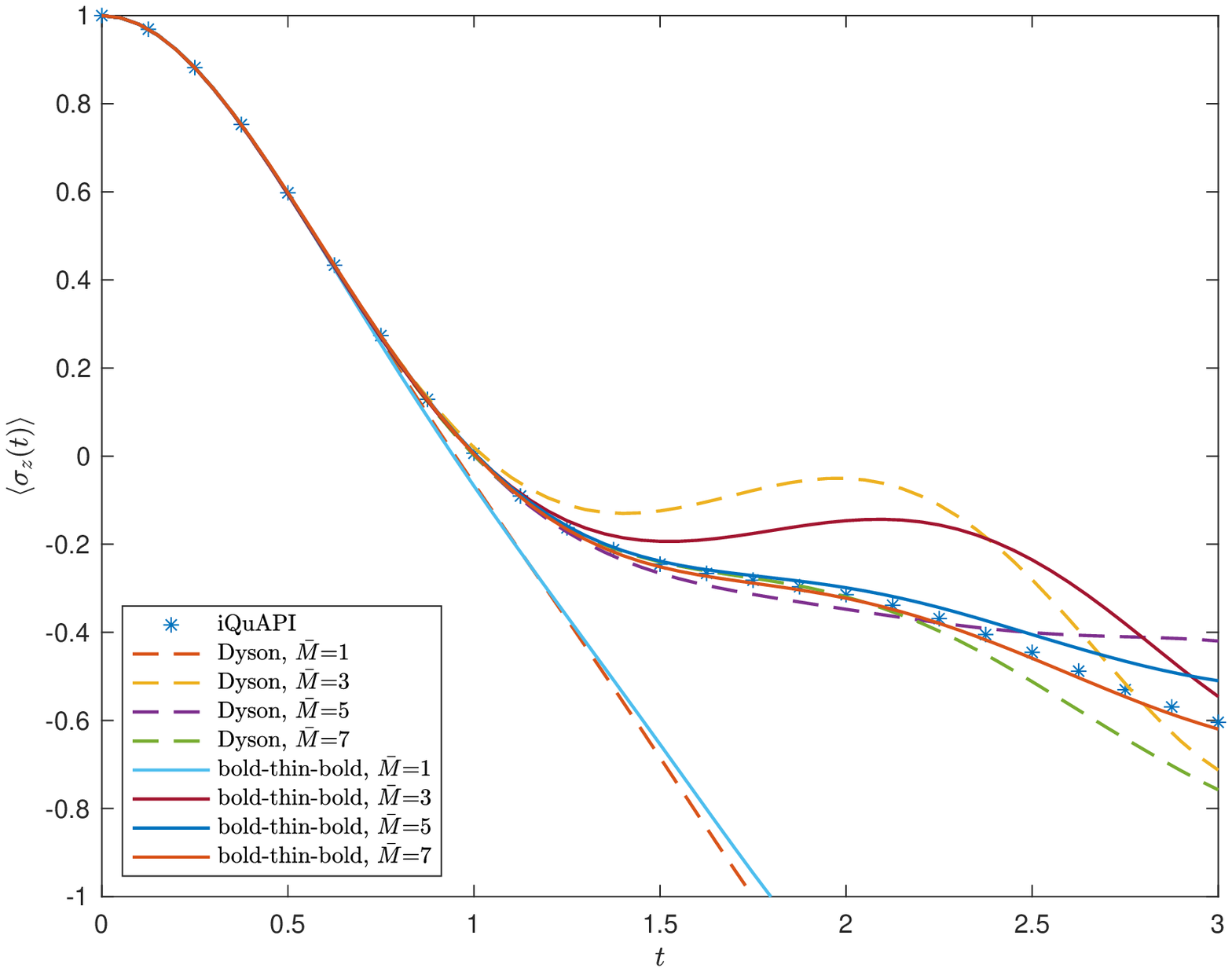}}
    \quad
    \subfloat[$\omega_c = 5$\label{Test_4}]{\includegraphics[width= 0.45\textwidth]{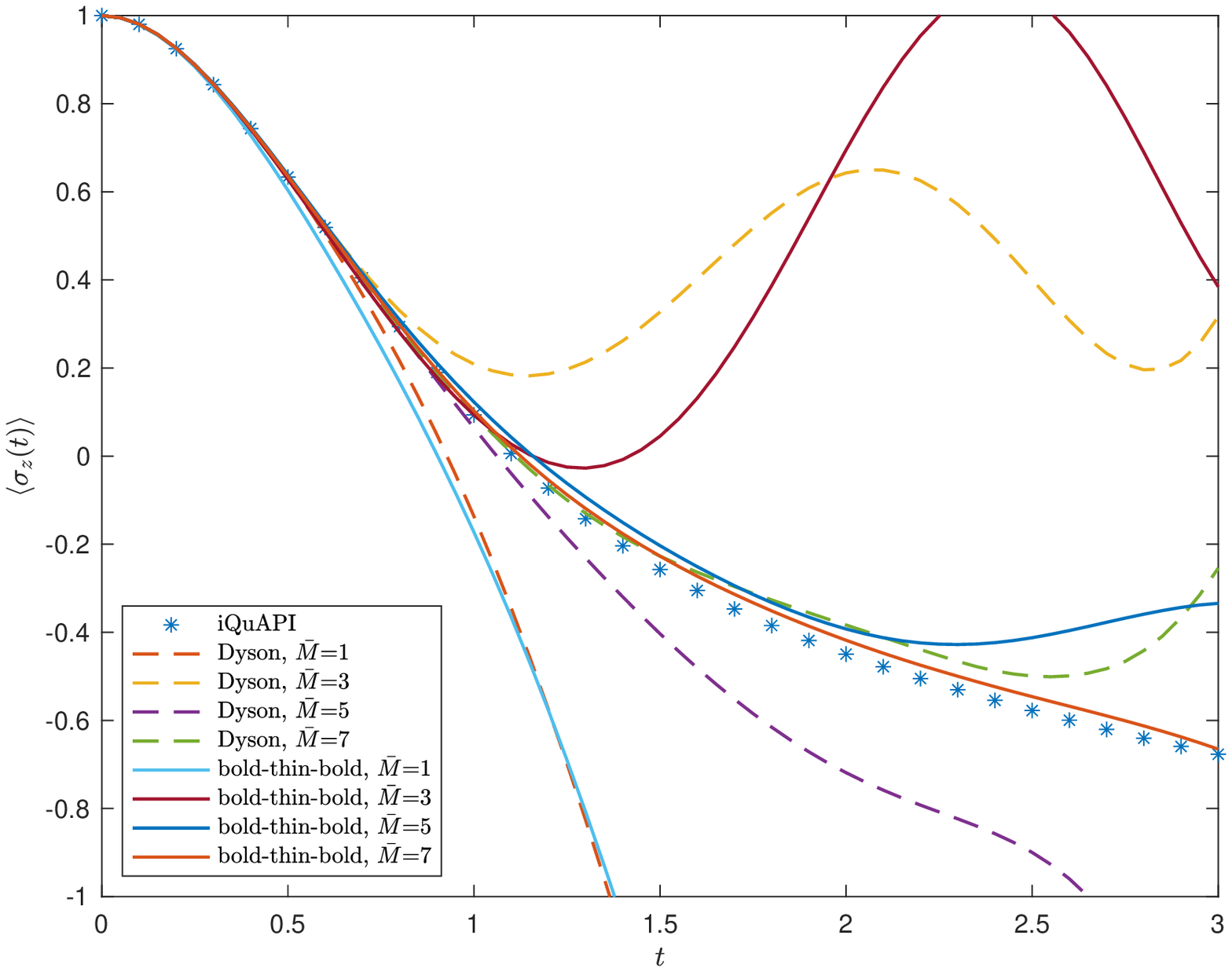}}
    \quad
    \caption{Evolution of $\expval{\hat{\sigma}_z(t)}$ with $\omega_c=2.5$ and $\omega_c=5$. 
    The number of samples is determined based on \eqref{Mnm} with $\mathcal{M}_0 = 10^6$.}
    \label{Test_3_4}
\end{figure}

% For both Dyson series method and the BTB method,
%  when the truncation dimension $\bar{M}$ increases,
%  the results are getting closer to 
%  the reference results based on the i-QuAPI.
In these experiments, early truncations of the series at
 $\bar{M}=1,3$ are insufficient to approximate the series $\mathcal{K}(-t,t)$ due to the stronger coupling.
In general, for a fixed $\bar{M}$,
 the BTB method provides more accurate results thanks to the use of bold lines.
For example, in \Cref{Test_3},
 if we compare the results with legends ``Dyson, $\bar{M}=7$'' and ``BTB, $\bar{M}=7$'',
it can be observed that the result of BTB method is valid for all $t \in [0,3]$,
while the result of Dyson series starts to deviate from the reference solution provided by i-QuAPI at $t=2$.

In addition to the accuracy,
 another advantage of BTB method over Dyson series is
 its better efficiency.
For the same value of $\bar{M}$, the computation time of BTB method is shorter since the number of diagrams in $\widehat{\mathcal{Q}}$ is 
 smaller than that in $\mathcal{Q}$ (see \Cref{tab:diagram_numbers}).
When $\bar{M}$ is large,
 the computation time of BTB method is significantly shorter compared to Dyson series.
\Cref{computation_time_compare} lists the wall clock time for the tests in \Cref{Test_3_4}.
From \Cref{computation_time_compare},
 it is clear that when $\bar{M}=7$, BTB method only takes about a third of the time
 that the Dyson series takes.
Another reason leading to this difference is that
 for each sample, Dyson series requires the complicated computation on the matrix exponential in each $\mathscr{G}_{ij}^{(0)}$. For BTB method, most of these expensive $\mathscr{G}_{ij}^{(0)}$ are replaced by the cheaper bold lines which are obtained from linear interpolations. The difference in the wall clock time between the experiments with $\omega_c=2.5$ and $\omega_c=5$
 is due to different choice of the empirical constant $\mathcal{B}$.
For $\omega_c=2.5$, the constant is set as 0.1942, while for $\omega_c=5$ the constant is 0.7652 ($\approx 0.1942 \times 3.9403$).
According to \eqref{Mnm},
 the number of samples for $\omega_c = 5$
 is about $3.9403^4\approx 240$ times as that for $\omega_c=2.5$ when $m=7$, which explains
the enormous difference between the two experiments.

\begin{table}
\centering
\begin{tabular}{| c |c c | c c|} 
 \hline
 &\multicolumn{2}{c|}{$\omega_c=2.5$} & \multicolumn{2}{c|}{$\omega_c=5$} \\
 
 $\bar{M}$ & Dyson & BTB & Dyson & BTB \\
 \hline
 1 & 3 & 17 & 8 & 56 \\ 
 3 & 24 & 34 & 312 & 290 \\
 5 & 73 & 58 & 3228 & 1698 \\
 7 & 241 & 115 & 53225 & 16086 \\
 \hline
\end{tabular}
\caption{Comparison of wall clock time (round to nearest second) for tests in \Cref{Test_3_4}.
\label{computation_time_compare}}
\end{table}

\subsection{Variance test}
The severeness of sign problem
 can be measured by the variance of the results.
Therefore, we carried our an experiment
 to numerically evaluate the variance of each method.
The parameters is the same as what we used in \Cref{test_difference_frequencies}:
\begin{equation*}
    \Delta = 1,\quad
    \epsilon = 1, \quad
    \omega_c = 2.5, \quad
    \beta = 5,\quad
    \xi=0.4.
\end{equation*}
In the test, we use the truncation dimension $\bar{M}=5$ and $\mathcal{M}_0 = 2\times 10^6$ as reference solutions, 
which are denoted by $\expval{\hat{\sigma}_z^{\mathrm{D}}(t)}$ for the method based on Dyson series and $\expval{\hat{\sigma}_z^{\mathrm{BTB}}(t)}$
for the BTB method.
We set $\mathcal{M}_0=1000$ and repeat the same experiments $N$ times.
The corresponding results are denoted by
$\expval{\hat{\sigma}_{z,k}^{\mathrm{D}}(t)}$
and $\expval{\hat{\sigma}_{z,k}^{\mathrm{BTB}}(t)}$
for $k=1,\dots,N$.
The variance of the methods are then numerically evaluated by
\begin{equation*}
    \mathrm{Var}^{\mathrm{D}}(t)
    = \frac{1}{N} \sum_{k=1}^N \left\vert
    \expval{\hat{\sigma}_{z,k}^{\mathrm{D}}(t)}
    - \expval{\hat{\sigma}_z^{\mathrm{D}}(t)}
    \right\vert^2,
    \quad
    \mathrm{Var}^{\mathrm{BTB}}(t)
    = \frac{1}{N} \sum_{k=1}^N \left\vert
    \expval{\hat{\sigma}_{z,k}^{\mathrm{BTB}}(t)}
    - \expval{\hat{\sigma}_z^{\mathrm{BTB}}(t)}
    \right\vert^2.
\end{equation*}
We then plot the result $\mathrm{Var}^{\mathrm{D}}(t)$ and $\mathrm{Var}^{\mathrm{BTB}}(t)$ in \Cref{Test_5} and their ratio in \Cref{Test_6}.
% \begin{figure}[ht]
%     \includegraphics[width=0.6\textwidth]{Figures/Test_5.eps}
%     \caption{The change of variances $\mathrm{Var}^\mathrm{D}(t)$ and $\mathrm{Var}^\mathrm{BTB}(t)$ with respect to $t$ (repeat times $N=10000$).}
%     \label{Test_5}
% \end{figure}

\begin{figure}[ht]
    \subfloat[\label{Test_5}]{\includegraphics[width= 0.45\textwidth]{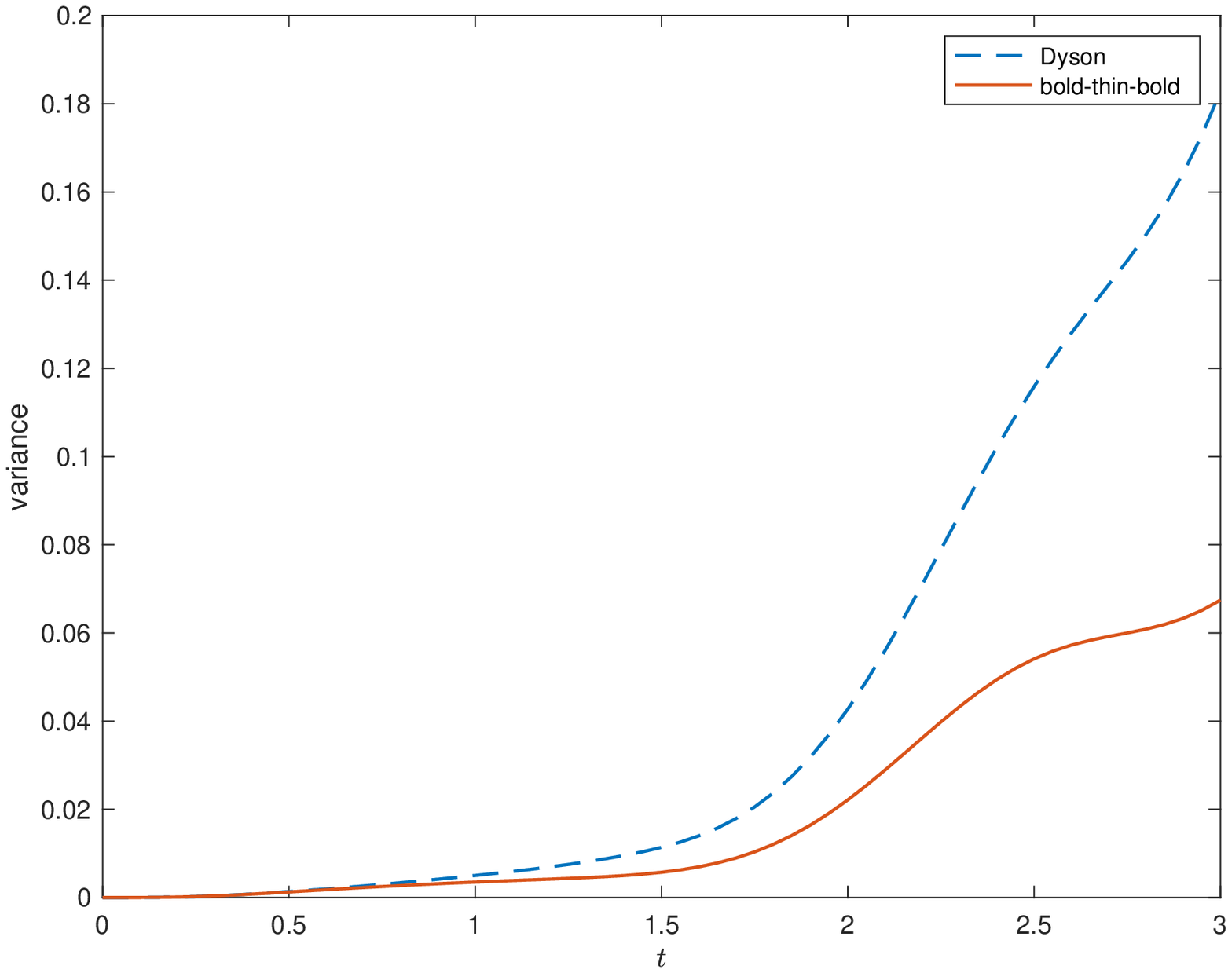}}
    \quad
    \subfloat[\label{Test_6}]{\includegraphics[width= 0.45\textwidth]{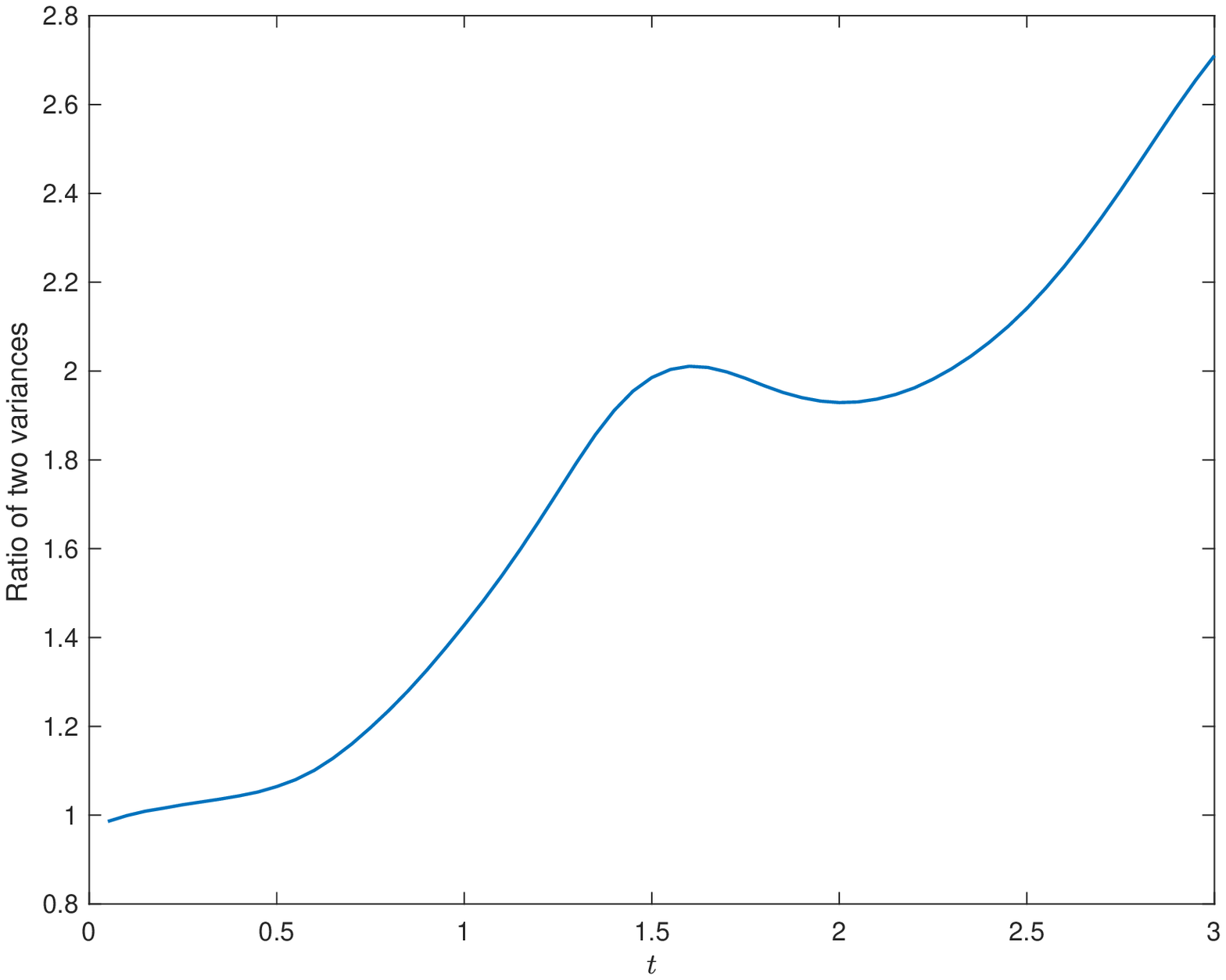}}
    \quad
    \caption{The change of variances $\mathrm{Var}^\mathrm{D}(t)$, $\mathrm{Var}^\mathrm{BTB}(t)$ and the ratio $\mathrm{Var}^\mathrm{D}(t)/\mathrm{Var}^\mathrm{BTB}(t)$ with respect to $t$ (repeat times $N=10000$).}
    \label{Test_5_6}
\end{figure}

As is shown in \Cref{Test_5},
 the variance of the BTB method 
 is lower throughout the simulation.
At $t=3$, the variance of both methods are $0.1828$ and $0.0674$ respectively.
This indicates 
 that the BTB method can better
 tame the numerical sign problem.
We also observe in \Cref{Test_6} that as time evolves, the variance ratio becomes larger, showing that
 the BTB method is more stable for longer time simulations.

\section{Conclusion}
\label{final_conclusion}
We propose two diagrammatic Monte Carlo methods which apply the recurrence relation of the infinite series in the simulations of quantum system-bath dynamics. The first approach is formulated as an integro-differential equation derived from Dyson series, which allows us to solve the evolution of the observables by some classic numerical schemes such as Runge-Kutta methods. The costliest term $\mathcal{K}(-t,t)$ in the equation is a series of high-dimensional integrals, where each integrand includes a system associated functional $\mathcal{U}^{(0)}$ and a bath influence functional $\Ls_b$, and the integrals are evaluated by the Monte Carlo method. We observe that a significant proportion of $\mathcal{K}(-t,t)$ can be expressed as a linear operator applied to $\mathcal{K}(-(t-\Delta t), t-\Delta t)$, so that our fast algorithm only needs to evaluate the remaining part at each time step. This leads to a significantly higher efficiency compared to the direct implementation without using the recurrence relation. Afterwards, inspired by the inchworm method, we develop the second approach called bold-thin-bold diagrammatic Monte Carlo, which clusters some of the diagrams in $\mathcal{K}(-t,t)$ into bold lines and hence the reformulated series form of $\mathcal{K}(-t,t)$ converges faster. The recurrence relation can again be applied to the BTB method, based on which a similar efficient algorithm is developed. These two algorithms can be considered as extensions to the idea in \cite{cai2022fast} that only reuses the information of the bath influence functionals. It can be demonstrated that our methods also have better performance in terms of memory cost. 

While the two methods are mainly discussed in the framework of spin-boson model, the idea can also be generalized to any quantum system interacting with harmonic bath where the two-point correlation $B(\tau_1,\tau_2)$ has a similar formula as \eqref{bath_function_B} which only relies on the time difference $\Delta \tau = |\tau_1|-|\tau_2|$. We also point out that as the BTB method is preferable for the simulations of relatively weak system-bath coupling where $\bar{M}$ is often not too large. In this case, it is worth considering replacing the Monte Carlo method with some deterministic numerical quadrature and impose some truncation of the memory length to avoid the numerical sign problem. This will be studied in our future work.

% In our current framework, 
%  the two methods use Monte Carlo sampling
%  to estimate the high dimensional integrals.
% The number of samples relies on the empirical constant $\mathcal{B}$.
% An idea to increase the computational efficiency especially for strong coupling case
%  is to develop a quadrature formula on the special domain $T_{t+\dt}^{(m)}$
%  so that numerical sign problem caused by Monte Carlo sampling can be avoided.
% Another natural idea for future work is to migrate the idea to more complicated systems
% other than the spin-boson model.
% In a general open quantum system,
%  the bath correlation function has a similar form that only depends on the time difference $\Delta \tau = |\tau_1|-|\tau_2|$. \ki{Is the statement of the bath correlation function correct?}

\bibliographystyle{abbrv}
\bibliography{myBib.bib}

\end{document}